\documentclass[a4paper]{article}
\usepackage{a4wide}
\usepackage{graphicx,xcolor}
\graphicspath{{}}
\usepackage{caption}
\usepackage{subcaption}
\usepackage{empheq}
\usepackage{amsthm}
\usepackage{amssymb, latexsym, mathrsfs}
\usepackage{amsmath}
\usepackage{dsfont}
\usepackage{enumerate}
\usepackage{algorithm}
\usepackage{color}
\usepackage{transparent}
\usepackage{subfig}
\usepackage{url}
\usepackage[affil-it]{authblk}
\usepackage{booktabs}

\usepackage[american,cuteinductors,smartlabels]{circuitikz}
\usepackage{tikz}
\usepackage{schemabloc}
\usetikzlibrary{shapes,arrows}

\theoremstyle{plain}

\newtheorem{assum}{Assumption}
\newtheorem{prbl}{Problem}

\newtheorem{rmk}{Remark}
\newtheorem{prop}{Proposition}

\newtheorem{defn}{Definition}

\newcommand{\Rset}{\mathbb{R}}

\newcommand{\hd}{{\hat{d}}}

\newcommand{\hx}{{\hat{x}}}

\newcommand{\bd}{{\bar{d}}}

\newcommand{\CC}{{\mathcal{C}}}
\newcommand{\DD}{{\mathcal{D}}}

\newcommand{\NN}{{\mathcal{N}}}

\newcommand{\PP}{{\mathcal{P}}}

\newcommand{\diag}{{\mbox{diag}}}                  
\newcommand{\abs}[1]{{|{#1}|}}                     
\newcommand{\norme}[2]{||{#1}||_{#2}}            
\newcommand{\rank}{\mbox{rank}}                    

\newcommand{\mbf}[1]{\mathbf{#1}}                  

\newcommand{\Zero}{\textbf{0}}

\newcommand{\subss}[2]{{#1}_{[#2]}}

\newcommand{\dx}{{\dot x}}

\newcommand{\matr}[1]{
\begin{bmatrix}
    #1
\end{bmatrix}
}

\begin{document}

     \title{Plug-and-play voltage and frequency control of islanded microgrids with meshed topology\thanks{The research leading to these results has received funding from the European Union Seventh Framework Programme [FP7/2007-2013]  under grant agreement n$^\circ$ 257462 HYCON2 Network of excellence.}}

    \author{Stefano Riverso%
       \thanks{Electronic address: \texttt{stefano.riverso@unipv.it}; Corresponding author}} 

     \author{Fabio Sarzo%
       \thanks{Electronic address: \texttt{fabio.sarzo01@universitadipavia.it}}}

     \author{Giancarlo Ferrari-Trecate%
       \thanks{Electronic address: \texttt{giancarlo.ferrari@unipv.it}} }

     \affil{Dipartimento di Ingegneria Industriale e dell'Informazione\\Universit\`a degli Studi di Pavia}
     \date{\textbf{Technical Report}\\ September, 2014}

     \maketitle

     \begin{abstract}
         In this paper we propose a new decentralized control scheme for Islanded microGrids (ImGs) composed by the interconnection of Distributed Generation Units (DGUs). Local controllers regulate voltage and frequency at the Point of Common Coupling (PCC) of each DGU and they are able to guarantee stability of the overall ImG. The control design procedure is decentralized, since, besides two global scalar quantities, the synthesis of a local controller uses only information on the corresponding DGU and lines connected to it. Most important, our design procedure enables Plug-and-Play (PnP) operations: when a DGU is plugged in or out, only DGUs physically connected to it have to retune their local controllers. We study the performance of the proposed controllers simulating different scenarios in MatLab/Simulink and using performance indexes proposed in IEEE standards. 

     \end{abstract}
     
     \newpage

      \section{Introduction}
           In recent years, research on Islanded microGrids (ImG) has received major attention. ImGs are self-sufficient micro grids composed of several Distributed Generation Units (DGUs) designed to operate safely and reliably in absence of a connection with the main grid. Besides fostering the use of renewable generation, ImGs bring distributed generation sources close to loads and allow power to be delivered to rural areas, remote lands, islands or harsh environments \cite{Lasseter2002,Lasseter2004,Guerrero2013}. The interest in ImGs is also motivated by microgrids that normally operate in grid-connected mode and that can be switched off-grid for guaranteeing users remain powered in presence of grid faults. In particular, for buildings such hospitals and airports, ImGs offer a interesting solution for emergency generation since, differently from common diesel generators, power is produced and delivered to the main grid in absence of faults.

          For grid-connected microgrids, voltage and frequency are set by the main grid. However, in islanded mode, voltage and frequency control must be controlled by DGUs. This is a challenging task, especially if one allows for (a) meshed topology with the goal increasing redundancy and robustness to line faults; (b) decentralized regulation of voltage and frequency, meaning that each DGU is equipped with a local controller and controllers do not communicate in real-time.

          As reviewed in \cite{Guerrero2013} many available decentralized regulators are based on droop control \cite{Chandorkar1996,Katiraei2005,Piagi2006,Guerrero2007,Mehrizi-Sani2010,Guerrero2011}. The main drawback of applying the droop method to ImGs is that frequency and amplitude deviations can be heavily affected by loads. For these reasons, a secondary control layer to restore system frequency and voltage to nominal values is needed \cite{Chandorkar1996,Katiraei2006,IEEE2011}. 

          Stability is another critical issue in ImGs controlled in a decentralized way \cite{Guerrero2013}. The key challenge is to guarantee stability is not spoiled by the interaction among DGUs and, in the context of droop control, this issue has been investigated only recently \cite{Simpson-Porco2013}. For regulators not based on droop control, almost all studies focused on radial ImGs (i.e. DGUs are not connected in a loop fashion) while control of ImGs with meshed topology is still largely unexplored \cite{Guerrero2013}.

          In this paper we consider the design of decentralized regulators for meshed ImGs with a view on decentralization of the synthesis procedure. More specifically, we develop a Plug-and-Play (PnP) design algorithm where the synthesis of a local controller for a DGU requires parameters of transmission lines connected to it, the knowledge of two global scalar parameters, but not specific information about any other DGU. This implies that when a DGU is plugged in or out, only DGUs physically connected to it have to retune their local controllers.

          PnP control design for general linear constrained systems has been proposed in \cite{Riverso2013c,Riverso2014a} and \cite{Riverso2014_thesis}. PnP design for ImGs is however different since it is based on the concept of neutral interactions \cite{Lunze1992} rather than on robustness against subsystem coupling. Furthermore, for achieving neutral interactions among DGUs, we exploit Quasi-Stationary Line (QSL) approximations of line dynamics \cite{Venkatasubramanian1995}.   

          Our theoretical results are backed up by simulations using realistic models of Voltage Source Converters (VSCs), associated filters and transformers. As a first testbed, we consider two radially connected DGUs \cite{Moradi2010} and show that, in spite of QSL approximations, PnP controllers exhibit very good performances (measured as in the IEEE standard \cite{IEEE2009}) in terms of voltage tracking and robustness to nonlinear and unbalanced loads. We then consider a meshed ImG with 10 DGUs and discuss the real-time plugging in and out of a DGU.

          The paper is organized as follows. In Section \ref{sec:Model} we present dynamical models of ImGs and introduce the adopted line approximation. In Section \ref{sec:PnPctrl} we exploit the notion of neutral interactions for designing decentralized controllers and we discuss how to perform PnP operations. Moreover we will show how to improve performance and how to allow PnP capabilities of the proposed controllers. In Section \ref{sec:Simresults} we study performance of PnP controllers through simulation case studies. Section \ref{sec:conclusions} is devoted to some conclusions.

     \clearpage

     \section{Microgrid model}
          \label{sec:Model}
          In this section, we present dynamical models of ImGs used in this paper. For the sake of clearness, we first introduce an ImG consisting of two parallel DGUs and then generalize the model to ImGs composed of $N$ DGUs.\\
          As in \cite{Karimi2008,Moradi2010,Etemadi2012a,Etemadi2012,Babazadeh2013}, we consider the microgrid in Figure \ref{fig:schemairandist} where two DGUs, generally denoted with $i$ and $j$, are connected through a three-phase line with non-zero impedance ($R_{ij}$, $L_{ij}$). Each DGU is composed of a DC voltage source (representing a generic renewable resource), a voltage source converter (VSC), a series filter described by a resistance $R_{t}$ and an inductance $L_{t}$ and a step-up transformer (Y-$\Delta$) which connects the DGU to the remainder of the electrical network at Point of Common Coupling (PCC). Transformer parameters, except the transformation ratio $k$, are included in $R_{t}$ and $L_{t}$.
          \begin{figure}[!htb]
            \centering
            \ctikzset{bipoles/length=0.7cm}
\begin{circuitikz}[scale=1,transform shape]
\ctikzset{current/distance=1}
\draw
node[transformer] (Ti) at (0,0) {}
(Ti.center) node{$k_i$}
node[transformer] (Tj) at ($(Ti.center)+(5.4,0)$) {}
(Tj.center) node{$k_j$}

node[ocirc] (Aibattery) at ([xshift=-3.14cm,yshift=-0.2cm]Ti.A1) {}
node[ocirc] (Bibattery) at ([xshift=-3.14cm,yshift=0.2cm]Ti.A2) {}
(Bibattery) to [battery] (Aibattery) {}
node [rectangle,draw,minimum width=1cm,minimum height=2.4cm] (vsci) at ($0.5*(Aibattery)+0.5*(Bibattery)+(0.8,0)$) {\scriptsize{VSC $i$}}
(Aibattery) to [short] ([xshift=0.3cm]Aibattery)
(Bibattery) to [short] ([xshift=0.3cm]Bibattery)

 node[ocirc] (Ai) at ($(Aibattery)+(1.54,0.2)$) {}
 node[ocirc] (Bi) at ($(Bibattery)+(1.54,-0.2)$) {}
(Ai) to [short] ([xshift=-0.24cm]Ai)
(Bi) to [short] ([xshift=-0.24cm]Bi)
(Ai) to [R, l=\scriptsize{$R_{ti}$}] ($(Ai)+(0.8,0)$) {}
to [L, l=\scriptsize{$L_{ti}$}]($(Ai)+(1.5,0)$){}
to [short,i=\scriptsize{$I_{ti}$}]($(Ai)+(1.6,0)$){}
(Bi) to [short] (Ti.A2);
\begin{scope}[shorten >= 10pt,shorten <= 10pt,]
\draw[<-] (Ai) -- node[right] {$V_{ti}$} (Bi);
\end{scope};

\draw
($(Ti.B1)+(0.25,0.25)$) node[anchor=south]{\scriptsize{$PCC_i$}}
($(Ti.B1)+(0.25,0)$) node[anchor=south]{\scriptsize{$V_i$}}
($(Ti.B1)+(0.25,0)$) node[ocirc](PCCi){}
($(Ti.B1)+(0,-0.5)$) to [I](Ti.B2)--($(Ti.B2)+(0.5,0)$)
($(Ti.B1)+(0,-0.5)$) to [short,i<=\scriptsize{$I_{Li}$}] (Ti.B1)--($(Ti.B1)+(0.5,0)$)
($(Ti.B1)+(0.5,0)$) to [C, l=\scriptsize{$C_{ti}$}] ($(Ti.B2)+(0.5,0)$)

($(Ti.B1)+(0.5,0)$) to [short] ($(Ti.B1)+(0.7,0)$)
($(Ti.B1)+(0.7,0)$) to [short,i^<=\scriptsize{$I_{ij}$}] ($(Ti.B1)+(0.8,0)$)
($(Ti.B1)+(0.8,0)$) to [R, l=\scriptsize{$R_{ij}$}] ($(Ti.B1)+(1.6,0)$) {}
to [L, l=\scriptsize{$L_{ij}$}]($(Ti.B1)+(2.4,0)$){}
($(Ti.B1)+(2.4,0)$)--($(Ti.B1)+(2.5,0)$) to [short,i=\scriptsize{$I_{ji}$}] ($(Ti.B1)+(2.6,0)$)
($(Ti.B1)+(2.6,0)$) to [short] ($(Ti.B1)+(2.8,0)$)
($(Ti.B2)+(0.5,0)$) to [short] ($(Ti.B2)+(2.8,0)$)

($(Tj.A1)-(0.25,-0.2)$) node[anchor=south]{\scriptsize{$PCC_j$}}
($(Tj.A1)-(0.25,0)$) node[anchor=south]{\scriptsize{$V_j$}}
($(Tj.A1)-(0.25,0)$) node[ocirc](PCCj){}
($(Ti.B2)+(3.3,0)$)--($(Ti.B2)+(2.8,0)$) to [C, l=\scriptsize{$C_{tj}$}] ($(Ti.B1)+(2.8,0)$)--($(Ti.B1)+(3.3,0)$)
($(Ti.B1)+(3.3,-0.5)$) to [I]($(Ti.B2)+(3.3,0)$)
(Tj.A1) to [short,i=\scriptsize{$I_{Lj}$}] ($(Tj.A1)+(0,-0.5)$)

(Tj.B1) to [short,i<=\scriptsize{$I_{tj}$}]($(Tj.B1)+(0.1,0)$){}
($(Tj.B1)+(0.1,0)$) to [L, l=\scriptsize{$L_{tj}$}] ($(Tj.B1)+(0.8,0)$) {}
($(Tj.B1)+(0.8,0)$) to [R, l=\scriptsize{$R_{tj}$}] ($(Tj.B1)+(1.6,0)$) {}
 node[ocirc] (Aj) at ($(Tj.B1)+(1.6,0)$) {}
 node[ocirc] (Bj) at ($(Tj.B2)+(1.6,0)$) {}
(Aj) to [short] ([xshift=0.24cm]Aj)
(Bj) to [short] ([xshift=0.24cm]Bj)
(Bj) to [short] (Tj.B2);
\begin{scope}[shorten >= 10pt,shorten <= 10pt,]
\draw[<-] (Aj) -- node[left] {$V_{tj}$} (Bj);
\end{scope};

\draw
node[ocirc] (Ajbattery) at ([xshift=3.14cm,yshift=-0.2cm]Tj.B1) {}
node[ocirc] (Bjbattery) at ([xshift=3.14cm,yshift=0.2cm]Tj.B2) {}
(Bjbattery) to [battery] (Ajbattery) {}
node [rectangle,draw,minimum width=1cm,minimum height=2.4cm] (vsci) at ($0.5*(Ajbattery)+0.5*(Bjbattery)-(0.8,0)$) {\scriptsize{VSC $j$}}
(Ajbattery) to [short] ([xshift=-0.3cm]Ajbattery)
(Bjbattery) to [short] ([xshift=-0.3cm]Bjbattery)

node [rectangle,draw,minimum width=6.38cm,minimum height=3.3cm,dashed,label=\small\textbf{DGU $i$}] (DGUi) at ($0.5*(Aibattery)+0.5*(Bibattery)+(2.9,0)$) {}
node [rectangle,draw,minimum width=6.38cm,minimum height=3.3cm,dashed,label=\small\textbf{DGU $j$}] (DGUj) at ($0.5*(Ajbattery)+0.5*(Bjbattery)-(2.9,0)$) {}
node [rectangle,draw,minimum width=1.6cm,minimum height=3.3cm,dashed,label=\small\textbf{Line $ij$ and $ji$}] (Lineij) at ($0.5*(DGUi.center)+0.5*(DGUj.center)+(0,0)$){}

;\end{circuitikz}
            \caption{Electrical scheme of an ImG composed of two radially connected DGUs with unmodeled loads.}
            \label{fig:schemairandist}
          \end{figure}
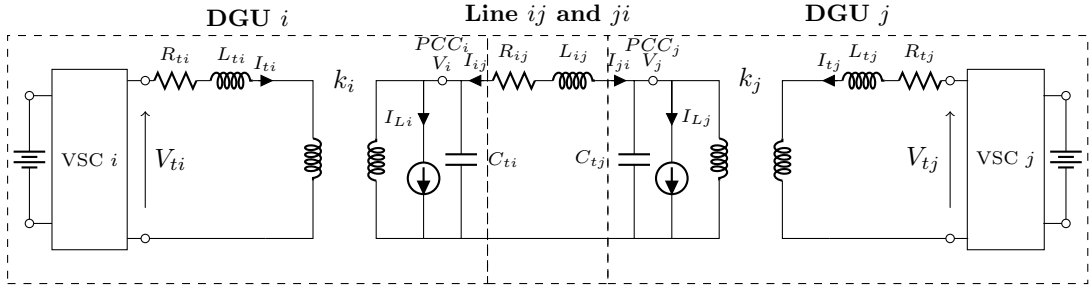
          
          Each DGU provides real and reactive power for its local loads connected to the PCC. We assume loads are unknown and, similarly to \cite{Babazadeh2013}, we treat load currents $I_L$ as disturbances for the DGUs. As shown in Figure \ref{fig:schemairandist}, at the PCC of each area we use a shunt capacitance $C_t$ for attenuating the impact of high-frequency harmonics of the load voltage. From the dynamical equations of this scheme in the \emph{abc}-frame, after applying the Park's Transformation \cite{Park1929} we obtain the following a \emph{dq}-frame rotating with speed $\omega_0$
          \begin{subequations}
            \label{eq:sysdistdq}            
            \begin{empheq}[left=$DGU \emph{i}\quad$\empheqlbrace]{align}
              \label{eq:sysdistdqB}\frac{dV_{i,dq}}{dt} + j\omega_0 V_{i,dq} &= \frac{k_i}{C_{ti}}I_{ti,dq}+\frac{1}{C_{ti}}I_{{ij},dq}-\frac{1}{C_{ti}}I_{Li,dq} \\
              \label{eq:sysdistdqA}\frac{dI_{ti,dq}}{dt} + j\omega_0 I_{ti,dq} &= -\frac{R_{ti}}{L_{ti}}I_{ti,dq}-\frac{k_i}{L_{ti}}V_{i,dq}+\frac{1}{L_{ti}}V_{ti,dq}
            \end{empheq}
            \begin{empheq}[left=$Line \emph{ij}\quad$\empheqlbrace]{align}
              \label{eq:sysdistdqC}\frac{dI_{ij,dq}}{dt} + j\omega_0 I_{ij,dq} &= \frac{1}{L_{ij}}V_{j,dq}-\frac{R_{ij}}{L_{ij}}I_{{ij},dq}-\frac{1}{L_{ij}}V_{i,dq}
            \end{empheq}
            \begin{empheq}[left=$Line \emph{ji}\quad$\empheqlbrace]{align}
              \label{eq:sysdistdqCji}\frac{dI_{ji,dq}}{dt} + j\omega_0 I_{ji,dq} &= \frac{1}{L_{ji}}V_{i,dq}-\frac{R_{ji}}{L_{ji}}I_{{ji},dq}-\frac{1}{L_{ji}}V_{j,dq}
            \end{empheq}
            \begin{empheq}[left=$DGU \emph{j}\quad$\empheqlbrace]{align}
              \label{eq:sysdistdqD}\frac{dV_{j,dq}}{dt} + j\omega_0 V_{j,dq} &= \frac{k_j}{C_{tj}}I_{tj,dq}+\frac{1}{C_{tj}}I_{{ji},dq}-\frac{1}{C_{tj}}I_{Lj,dq} \\
              \label{eq:sysdistdqE}\frac{dI_{tj,dq}}{dt} + j\omega_0 I_{tj,dq} &= -\frac{R_{tj}}{L_{tj}}I_{tj,dq}-\frac{k_j}{L_{tj}}V_{j,dq}+\frac{1}{L_{tj}}V_{tj,dq}
            \end{empheq}
          \end{subequations}
          Each state in \eqref{eq:sysdistdq}   can be split in two parts (the real component \emph{d-} and the imaginary component \emph{q-} of \emph{dq} reference frame, respectively). Note that in \eqref{eq:sysdistdqC} and \eqref{eq:sysdistdqCji}, one obtains two opposite line currents $I_{ij}$ and $I_{ji}$ so as to have a reference current entering in each DGU. In order to guarantee that $I_{ij}(t)=-I_{ji}(t)$, $\forall t\geq 0$, we introduce the following assumption.
          \begin{assum}
            \label{ass:lines}
            Initial states verify $I_{ij,dq}(0)=-I_{ji,dq}(0)$. Moreover it holds $L_{ij}=L_{ji}$ and $R_{ij}=R_{ji}$.
          \end{assum}
          \begin{rmk}
            System in \eqref{eq:sysdistdqC}, \eqref{eq:sysdistdqCji} can be seen as an expansion of the line model one can obtain by fixing a single reference direction for the line current and introducing a single state variable. For a definition of expansion of a system we defer the reader to Section 3.4 in \cite{Lunze1992}. System \eqref{eq:sysdistdq} can also be viewed as a system of differential-algebraic equations, given by \eqref{eq:sysdistdqB}-\eqref{eq:sysdistdqC}, \eqref{eq:sysdistdqD}, \eqref{eq:sysdistdqE} and $I_{ij}(t)=-I_{ji}(t)$.
          \end{rmk}
          Using the proposed notation for the lines, both DGU models have the same structure. Modeling the load current $I_{L*,dq}$, $*\in\{i,j\}$ as a disturbance, \eqref{eq:sysdistdq} can be represented through the linear system
          \begin{equation}
            \label{eq:sysdistABCDM}
            \begin{aligned}
              \dot{x}(t) &= Ax(t)+Bu(t)+Md(t)\\
              y(t)       &= Cx(t)
            \end{aligned}
          \end{equation}
          where $x=[V_{i,d},V_{i,q},I_{ti,d},I_{ti,q},I_{{ij},d},I_{{ij},q},I_{{ji},d},I_{{ji},q},V_{j,d},V_{j,q},I_{tj,d},I_{tj,q}]^T$, $u=[V_{ti,d},V_{ti,q},V_{tj,d},V_{tj,q}]^T$, $d=[I_{Li,d},I_{Li,q},I_{Lj,d},I_{Lj,q}]^T$, $y=[V_{i,d},V_{i,q},V_{j,d},V_{j,q}]^T$, are, respectively, the state, input, disturbance, and output of the system. The matrices in \eqref{eq:sysdistABCDM} are obtained from \eqref{eq:sysdistdq} and are given in Appendix \ref{sec:AppMasterSlave}. Model \eqref{eq:sysdistABCDM} is often referred in literature to as \emph{master-slave} model, due to the fact that, for control purposes, the state of the line is controlled by one DGU only \cite{Babazadeh2013}. In the next section, we propose an approximate model that allows one to describe each DGU as a dynamical system affected directly by state of the other DGU, hence avoiding the need of using the line current in the DGU state equations.

          \subsection{QSL model}
               \label{sec:newmodel}               
               As in equation (T1.10) in \cite{Venkatasubramanian1995}, we set $\frac{d I_{{ij},dq}}{dt}=0$ and $\frac{d I_{{ji},dq}}{dt}=0$ (see also \cite{Akagi2007} and references therein). Then, \eqref{eq:sysdistdqC} and \eqref{eq:sysdistdqCji} give the QSL model
               \begin{equation}
                 \label{eq:staticline}
                 \begin{aligned}
                   \bar{I}_{ij,dq} = \frac{V_{j,dq}}{(R_{ij}+ j\omega_0 L_{ij})} -\frac{V_{i,dq}}{(R_{ij}+ j\omega_0 L_{ij})}\\
                   \bar{I}_{ji,dq} = \frac{V_{i,dq}}{(R_{ji}+ j\omega_0 L_{ji})} -\frac{V_{j,dq}}{(R_{ji}+ j\omega_0 L_{ji})}\\
                 \end{aligned}
               \end{equation}
               We then replace variables $I_{ij,dq}$ and $I_{ji,dq}$, in \eqref{eq:sysdistdqA} and \eqref{eq:sysdistdqE} with the right-hand side of \eqref{eq:staticline}. Splitting complex \emph{dq} quantities in their \emph{d} and \emph{q} components one obtains the new model for DGU $i$
               \begin{equation}
                 \label{eq:newDGU}
                 \text{DGU}~i\quad\left\lbrace
                   \begin{aligned}
                     \frac{dV_{i,d}}{dt} &= \omega_0 V_{i,q} +\frac{k_i}{C_{ti}}I_{ti,d}-\frac{1}{C_{ti}}I_{Li,d}+\frac{1}{C_{ti}}\bar{I}_{ij,d}\\
                     \frac{dV_{i,q}}{dt} &= -\omega_0 V_{i,d} +\frac{k_i}{C_{ti}}I_{ti,q}-\frac{1}{C_{ti}}I_{Li,q}+\frac{1}{C_{ti}}\bar{I}_{ij,q}\\
                     \frac{dI_{ti,d}}{dt} &= -\frac{k}{L_{ti}}V_{i,d}-\frac{R_{ti}}{L_{ti}}I_{ti,d}+\omega_0 I_{ti,q} +\frac{1}{L_{ti}}V_{ti,d}\\
                     \frac{dI_{ti,q}}{dt} &= -\frac{k}{L_{ti}}V_{i,q}-\omega_0I_{ti,d}-\frac{R_{ti}}{L_{ti}} I_{ti,q} +\frac{1}{L_{ti}}V_{ti,q}\\
                   \end{aligned}
                 \right.
               \end{equation}
               Similarly, switching indexes $i$ and $j$ in \eqref{eq:newDGU} one obtains the model of DGU $i$
               \begin{equation}
                 \label{eq:subsysDGUi}
                 \subss{\Sigma}{i}^{DGU} :
                 \left\lbrace
                   \begin{aligned}
                     \subss{\dot{x}}{i}(t) &= A_{ii}\subss{x}{i}(t) + B_{i}\subss{u}{i}(t)+M_{i}\subss{d}{i}(t)+ \subss\xi i(t)\\
                     \subss{y}{i}(t)       &= C_{i}\subss{x}{i}(t)\\
                     \subss{z}{i}(t)       &= H_{i}\subss{y}{i}(t)\\
                   \end{aligned}
                 \right.
               \end{equation}
               where $\subss{x}{i}=[V_{i,d},V_{i,q},I_{ti,d},I_{ti,q}]^T$, $\subss{u}{i} = [V_{ti,d},V_{ti,q}]^T$, $\subss{d}{i} = [I_{Li,d},I_{Li,q}]^T$, $\subss{z}{i} = [V_{i,d},V_{i,q}]^T$ are the state, the control input, the exogenous input and the controlled variables. The measurable output is $\subss y i(t)$ and we assume $\subss{y}{i}=\subss{x}{i}$. Furthermore $\subss\xi i(t)=A_{ij}\subss x j$ is the coupling with DGU $j$. The matrices of $\subss{\Sigma}{i}^{DGU}$ are obtained from \eqref{eq:newDGU} and they are collected in Appendix \ref{sec:AppMasterMaster}. As for the line, we have the subsystem 
               \begin{equation}
                 \label{eq:subsysLine}
                 \subss{\Sigma}{ij}^{Line} :
                 \left\lbrace
                   \subss{\dot{x}}{l,ij}(t) = A_{ll,ij}\subss{x}{ll,ij}(t) + A_{li,ij}\subss{x}{i}(t) + A_{lj,ij}\subss{x}{j}(t)\\
                 \right.
               \end{equation}
                where $\subss{x}{l,ij}=[I_{ij,d},I_{ij,q}]^T$ is the state of the line. The matrices of \eqref{eq:subsysLine} are obtained from \eqref{eq:sysdistdqC} and collected in Appendix \ref{sec:AppMasterMaster}. From \eqref{eq:subsysDGUi} and \eqref{eq:subsysLine}, the overall model of the microgrid in Figure \ref{fig:schemairandist} is
                \begin{equation}
                  \label{eq:overallmodeltwoDGU}
                  \begin{aligned}
                    \begin{bmatrix}
                      \subss{\dx}{i} \\
                      \subss{\dx}{j} \\
                      \subss{\dx}{l,ij} \\
                      \subss{\dx}{l,ji}
                    \end{bmatrix} 
                    &= 
                    \begin{bmatrix}
                      A_{ii} & A_{ij} & 0 & 0 \\
                      A_{ji} & A_{jj} & 0 & 0 \\
                      A_{li,ij} & A_{lj,ij} & A_{ll,ij} & 0 \\
                      A_{li,ji} & A_{lj,ji} & 0 & A_{ll,ji}
                    \end{bmatrix}
                    \begin{bmatrix}
                      \subss{x}{i} \\
                      \subss{x}{j} \\
                      \subss{x}{l,ij} \\
                      \subss{x}{l,ji}
                    \end{bmatrix}
                    +
                    \begin{bmatrix}
                      B_{i} & 0\\
                      0 & B_{i} \\
                      0 & 0  
                    \end{bmatrix}
                    \begin{bmatrix}
                      \subss{u}{i} \\
                      \subss{u}{j}  
                    \end{bmatrix}
                    +
                    \begin{bmatrix}
                      M_{i} & 0\\
                      0 & M_{j} \\
                      0 & 0  
                    \end{bmatrix}
                    \begin{bmatrix}
                      \subss{d}{i} \\
                      \subss{d}{j}  
                    \end{bmatrix}
                    \\		
                    \begin{bmatrix}
                      \subss{y}{i}\\
                      \subss{y}{j}
                    \end{bmatrix}
                    &=
                    \begin{bmatrix}
                      C_{1} & 0 & 0 & 0 \\
                      0 & C_{2} & 0 & 0 \\
                    \end{bmatrix}
                    \begin{bmatrix}
                      \subss{x}{i} \\
                      \subss{x}{j} \\
                      \subss{x}{l,ij} 
                    \end{bmatrix}
                    \\
                    \begin{bmatrix}
                      \subss{z}{i}\\
                      \subss{z}{j}
                    \end{bmatrix}
                    &=
                    \begin{bmatrix}
                      H_{i} & 0\\
                      0 & H_{j}
                    \end{bmatrix}
                    \begin{bmatrix}
                      \subss{y}{i} \\
                      \subss{y}{j} \\ 
                    \end{bmatrix}.
                  \end{aligned}
                \end{equation}
                
                \begin{rmk}
                  \label{rmk:Astruct}
                  Note that $A$ has following the block-triangular structure
                  \begin{equation*}
                    \label{eq:blktriangmatrix}
                    A=\left[\begin{array}{cc|cc}
			A_{ii} & A_{ij} & 0 & 0  \\
			A_{ji} & A_{jj} & 0 & 0 \\ \hline	
			A_{li,ij} & A_{lj,ij} & A_{ll,ij} & 0 \\
                        A_{li,ji} & A_{lj,ji} & 0 & A_{ll,ji}
                      \end{array}\right]
                  \end{equation*}
                  and hence, its eigenvalues are the union of the eigenvalues of $\matr{ A_{ii} & A_{ij} \\ A_{ji} & A_{jj} }$, $A_{ll,ij}$ and $A_{ll,ji}$. Note that $A_{ll,ij}=A_{ll,ji}$. In particular, positivity of line parameters implies that the line dynamics is asymptotically stable. Therefore stability of \eqref{eq:overallmodeltwoDGU} depends on the stability of local DGUs interconnected through the QSL model \eqref{eq:staticline}. Furthermore, designing decentralized controllers $\subss{u}{*}= k_*(\subss y *)$, $*\in\{i,j\}$, such that the connection of the DGUs is asymptotically stable, makes the overall closed-loop model of the microgrid asymptotically stable as well. QSL approximations have been widely studied in literature, especially in the field of power theory (see e.g. \cite{Venkatasubramanian1995,Akagi2007} and references therein).
                \end{rmk}
                
          \subsection{QSL model of a microgrid composed of $N$ DGUs}
               Next, we generalize model \eqref{eq:overallmodeltwoDGU} to microgrids composed of $N$ DGUs. Let $\DD=\{1,\ldots,N\}$. Two DGUs $i$ and $j$ are neighbours if there is a transmission line connecting them and we denote with $\NN_i\subset\DD$ the set of neighbours of DGU $i$. The dynamics of DGU $i$, can be then described by model \eqref{eq:subsysDGUi} setting $\subss\xi i =\sum_{j\in\NN_i} A_{ij}\subss{x}{j}(t)$. The new matrices of $\subss{\Sigma}{i}^{DGU}$ are provided in Appendix \ref{sec:AppNDGunit}. The overall QSL ImG model is given by
               \begin{subequations}
                 \label{eq:stdform}
                 \begin{align}
                   \label{eq:stdformA}\mbf{\dot{x}}(t) &= \mbf{Ax}(t) + \mbf{Bu}(t)+ \mbf{Md}(t)\\
                   \label{eq:stdformB}\subss{\dot{x}}{l,ij}(t) &= A_{ll,ij}\subss{x}{l,ij}(t) + A_{li,ij}\subss{x}{i}(t) + A_{lj,ij}\subss{x}{j}(t),~\forall i\in\DD,~\forall j\in\NN_i
                 \end{align}
               \end{subequations}
               where $\mbf x = (\subss x 1,\ldots,\subss x N)\in\Rset^{4N}$, $\mbf u = (\subss u 1,\ldots,\subss u N)\in\Rset^{2N}$, $\mbf d = (\subss d 1,\ldots,\subss d N)\in\Rset^{2N}$, and matrices $\mbf{A}$, $\mbf{B}$, $\mbf M$, $A_{ll,ij}$, $A_{li,ij}$ and $A_{lj,ij}$ are reported in Appendices \ref{sec:AppMasterMaster} and \ref{sec:AppNDGunit}. Note that $\mbf x$ is not influenced by any line state $\subss x {l,ij}$ and therefore, the collective model embracing all states has a block-triangular structure as in \eqref{eq:overallmodeltwoDGU}. We equip \eqref{eq:stdform} with the equations
               \begin{equation}
                 \label{eq:stdformOut}
                 \begin{aligned}
                   \mbf{y}(t)       &= \mbf{Cx}(t)\\
                   \mbf{z}(t)       &= \mbf{Hy}(t)
                 \end{aligned}
               \end{equation}
               where $\mbf y = (\subss y 1,\ldots,\subss y N)\in\Rset^{4N}$ are the measured variables and $\mbf z = (\subss z 1,\ldots,\subss z N)\in\Rset^{2N}$ are the controlled variables (see Appendix \ref{sec:AppNDGunit} for definition of matrices $\mbf C$ and $\mbf H$). Since neither $\mbf y$ nor $\mbf z$ depend upon states $\subss x {l,ij}$, and, since $\subss x {l,ij}$ does not influence $\mbf x$ too, equations \eqref{eq:stdformB} will be omitted in the sequel.               

\clearpage

     \section{Plug-and-Play decentralized voltage control}
	  \label{sec:PnPctrl}
	  \subsection{Decentralized control scheme with integrators}
               \label{sec:ctrlint}
               In order to track a constant set-point $\mbf{z_{ref}}(t)$, when $\mbf{d}(t)$ is constant, we augment the ImG model with integrators \cite{Skogestad1996}. For zeroing the steady-state error, it must hold
               \begin{equation}
                 \begin{aligned}
                   \Zero &= \mbf{A\bar{x}}+\mbf{B\bar{u}}+\mbf{M\bd}\\
                   \mbf{z_{ref}} &= \mbf{HC\bar{x}}
                 \end{aligned}
               \end{equation}
               \begin{equation}
                 \label{eq:cond_integrators}
                 \Gamma\begin{bmatrix}
                   \mbf{\bar{x}}\\
                   \mbf{\bar{u}}
                 \end{bmatrix}
                 =\begin{bmatrix}
                   \Zero & \mbf{-M}\\
                   \mbf{I} & \Zero
                 \end{bmatrix}
                 \begin{bmatrix}
                   \mbf{z_{ref}}\\
                   \mbf{\bd}
                 \end{bmatrix},\qquad
                 \Gamma = \begin{bmatrix}
                   \mbf{A} & \mbf{B}\\
                   \mbf{HC} & \Zero
                 \end{bmatrix} \in \Rset^{6N\times 6N}
               \end{equation}
               where $\mbf{\bar{x}}$ and $\mbf{\bar{u}}$ are equilibrium states and inputs.

               \begin{prop}
                 \label{prop:target}
                 Given $\mbf{z_{ref}}$ and $\mbf{\bd}$, vectors $\mbf{\bar{x}}$ and $\mbf{\bar{u}}$ that satisfy \eqref{eq:cond_integrators} always exist.
               \end{prop} 
               \begin{proof}
                 From \cite{Skogestad1996}, $\mbf{\bar{x}},\mbf{\bar{u}}$ verifying \eqref{eq:cond_integrators} exist if and only if the following two conditions are fulfilled.
                 \begin{enumerate}[(i)]
                 \item\label{enu:cond_integrators_p1} The number of controlled variables is not greater than the number of control inputs.
                 \item\label{enu:cond_integrators_p2} $\rank{(\Gamma)}=6N$. This is equivalent to require that the system under control has no invariant zeros.
                 \end{enumerate}
                 Condition (\ref{enu:cond_integrators_p1}) is verified since, in \eqref{eq:overallmodeltwoDGU}, $\subss u i$ and $\subss z i$ have the same size, $\forall i\in\DD$. Condition (\ref{enu:cond_integrators_p2}) can be easily proved using the definition of matrices $\mbf A$, $\mbf B$, $\mbf C$ and $\mbf H$ in \eqref{eq:stdform} and \eqref{eq:stdformOut} and the fact that electrical parameters are positive.
               \end{proof}
               \begin{figure}[!htb]
                 \centering
                 \tikzstyle{input} = [coordinate]
\tikzstyle{output} = [coordinate]
\tikzstyle{guide} = []
\tikzstyle{block} = [draw, rectangle, minimum height=1cm]

\begin{tikzpicture}
  \sbEntree{zref1}
  \sbDecaleNoeudy[3]{zref1}{zrefj}
  \sbDecaleNoeudy[3]{zrefj}{zrefN}
  \node [block, right of=zrefj,node distance=11cm,minimum height=4cm, minimum width=2cm] (microgrid) {\textbf{Microgrid}};
  \node [guide, right of=zrefj,xshift=-0.5cm] (zrefjline) {};
  \draw [draw] (zrefjline) -| node{$\vdots$} (zrefjline);
  
  \sbComph{sumret1}{zref1}
  \sbBloc{integrator1}{$\int dt$}{sumret1}
  \sbBloc[2.5]{controller1}{$K_1$}{integrator1}
  \sbRelier[$\subss{z_{ref}}{1}$]{zref1}{sumret1}
  \sbRelier{sumret1}{integrator1}
  \sbRelier[$\subss v 1$]{integrator1}{controller1}
  \node [guide, left of=microgrid,yshift=1.05cm,xshift=0.125cm] (u1end) {};
  \sbRelier[$\subss u 1$]{controller1}{u1end}
  
  \sbComp{sumretN}{zrefN}
  \sbBloc{integratorN}{$\int dt$}{sumretN}
  \sbBloc[4.5]{controllerN}{$K_N$}{integratorN}
  \sbRelier[$\subss{z_{ref}}{N}$]{zrefN}{sumretN}
  \sbRelier{sumretN}{integratorN}
  \sbRelier[$\subss v N$]{integratorN}{controllerN}
  \node [guide, left of=microgrid,yshift=-1.05cm,xshift=0.125cm] (uNend) {};
  \sbRelier[$\subss u N$]{controllerN}{uNend}
  
  \node [output, below of=microgrid,yshift=-1.8cm,xshift=0.5cm] (d1out) {};
  \node [output, below of=microgrid,yshift=-1.0cm,xshift=0.5cm] (d1outstart) {};
  \draw [draw,->,>=latex'] (d1out) -| node[yshift=-0.2cm]{$\subss d 1$} (d1outstart);
  
  \node [output, below of=microgrid,yshift=-1.5cm] (djout) {};
  \draw [draw] (djout) -| node {$\ldots$} (djout);
  
  \node [output, below of=microgrid,yshift=-1.8cm,xshift=-0.5cm] (dNout) {};
  \node [output, below of=microgrid,yshift=-1.0cm,xshift=-0.5cm] (dNoutstart) {};
  \draw [draw,->,>=latex'] (dNout) -| node[yshift=-0.2cm]{$\subss d N$} (dNoutstart);

  \node [output, above of=microgrid,yshift=1.8cm,xshift=0.5cm] (y1out) {};
  \node [output, above of=microgrid,yshift=1.0cm,xshift=0.5cm] (y1outstart) {};
  \node [output, above of=microgrid,yshift=1.5cm] (yjout) {};
  \node [output, above of=microgrid,yshift=1.2cm,xshift=-0.5cm] (yNout) {};
  \node [output, above of=microgrid,yshift=1.0cm,xshift=-0.5cm] (yNoutstart) {};
  \draw [draw] (y1outstart) -- node {} (y1out);
  \draw [draw,->,>=latex'] (y1out) -| node[yshift=0.2cm]{$\subss y 1$} (controller1);
  \draw [draw] (yjout) -| node {$\ldots$} (yjout);
  \draw [draw] (yNoutstart) -- node {} (yNout);
  \draw [draw,->,>=latex'] (yNout) -| node[yshift=0.2cm]{$\subss y N$} (controllerN);
  
  \node [guide, right of=microgrid,yshift=1.05cm,xshift=-.125cm] (z1) {};
  \node [guide, right of=microgrid,yshift=1.165cm,xshift=0.5cm] (z1near) {};
  \node [guide, right of=microgrid,yshift=1.05cm,xshift=1cm] (z1end) {};
  \draw [draw,->,>=latex',near end,swap] (z1) -- node[xshift=0.6cm] {$\subss{z} 1$} (z1end);
  \sbRenvoi[-6.8]{z1near}{sumret1}{}

  \node [guide, right of=microgrid,yshift=-1.05cm,xshift=-.125cm] (zN) {};
  \node [guide, right of=microgrid,yshift=-0.95cm,xshift=0.5cm] (zNnear) {};
  \node [guide, right of=microgrid,yshift=-1.05cm,xshift=1cm] (zNend) {};
  \draw [draw,->,>=latex',near end,swap] (zN) -- node[xshift=0.6cm] {$\subss{z} N$} (zNend);
  \sbRenvoi[6.8]{zNnear}{sumretN}{}
  
  \node [guide, right of=microgrid,xshift=0.5cm] (zj) {};
  \draw [draw] (zj) -| node {$\vdots$} (zj);
  
  \node [guide, left of=microgrid,xshift=-0.5cm] (uj) {};
  \draw [draw] (uj) -| node {$\vdots$} (uj);
  
\end{tikzpicture}
                 \caption{Control scheme with integrators for overall microgrid model.}
                 \label{fig:schemaint}
               \end{figure}
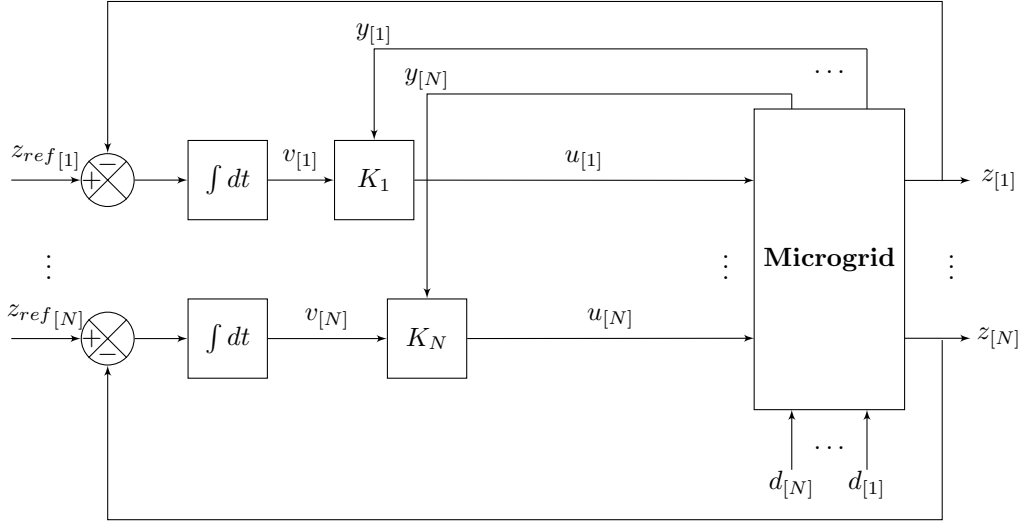
               The dynamics of the integrators is (see Figure \ref{fig:schemaint})
               \begin{equation}
                 \subss{\dot{v}}{i}(t) = \subss{e}{i}(t) = \subss{z_{ref}}{i}(t)-\subss{z}{i}(t) = \subss{z_{ref}}{i}(t)-H_{i}C_{i}\subss{x}{i}(t),
               \end{equation}
               and hence, the DGU model augmented with integrators is
               \begin{equation}
                 \label{eq:modelDGUgen-aug}
                 \subss{\hat{\Sigma}}{i}^{DGU} :
                 \left\lbrace
                   \begin{aligned}
                     \subss{\dot{\hat{x}}}{i}(t) &= \hat{A}_{ii}\subss{\hat{x}}{i}(t) + \hat{B}_{i}\subss{u}{i}(t)+\hat{M}_{i}\subss{\hat{d}}{i}(t)+ \sum_{j\in\NN_i}\hat{A}_{ij}\subss{\hat{x}}{j}(t)\\
                     \subss{\hat{y}}{i}(t)       &= \hat{C}_{i}\subss{\hat{x}}{i}(t)\\
                     \subss{z}{i}(t)       &= \hat{H}_{i}\subss{\hat{y}}{i}(t)
                   \end{aligned}
                 \right.
               \end{equation}
               where $\subss{\hat{x}}{i}=[\subss{x^T} i,v_{i,d},v_{i,q}]^T\in\Rset^6$ is the state, $\subss{\hat{y}}{i}=\subss{\hat{x}}{i}\in\Rset^6$ is the measurable output and $\subss{\hat{d}}{i}=[\subss{d}{i},\subss{z_{ref}}{i}]^T\in\Rset^4$ collects the exogenous signals (the current of the load and the reference signals, in $dq$ coordinates). Moreover, matrices $\hat{A}_{ii}, \hat{A}_{ij}, \hat{B}_{i}, \hat{C}_{i}, \hat{M}_{i}$ and $\hat{H}_{i}$  are defined as
               \begin{equation}
                 \label{eq:augith}
                 \hat{A}_{ii}=\begin{bmatrix}
                   A_{ii} & 0\\
                   -H_{i}C_{i} & 0
                 \end{bmatrix}
                 \hat{A}_{ij}=\begin{bmatrix}
                   A_{ij} &0\\
                   0&0
                 \end{bmatrix}
                 \hat{B}_{i}=\begin{bmatrix}
                   B_{i}\\
                   0
                 \end{bmatrix}
                 \hat{C}_{i}=\begin{bmatrix}
                   C_{i} & 0\\
                   0 & I
                 \end{bmatrix}
                 \hat{M}_{i}=\begin{bmatrix}
                   M_{i} & 0 \\
                   0 & I_2
                 \end{bmatrix}
                 \hat{H}_{i}=\begin{bmatrix}
                   H_{i} & 0
                 \end{bmatrix}.
               \end{equation}
               The following proposition guarantees that the pair $(\hat{A}_{ii},\hat{B}_{i})$ is controllable, hence system \eqref{eq:modelDGUgen-aug} can be stabilized.
               \begin{prop}
                 \label{prop:local_controllability}
                 The pair $(\hat{A}_{ii},\hat{B}_i)$ is controllable.
               \end{prop}
               \begin{proof}
                 Using the definition of controllability matrix, we have that
                 \begin{equation}
                   \label{eq:ctrb}
                   \begin{aligned}
                     \hat{M}^C_i &= \begin{bmatrix}
                       \hat B_{i} & \hat{A}_{ii}\hat{B}_i & \hat{A}_{ii}^2\hat{B}_i & \hat{A}_{ii}^3\hat{B}_i & \hat{A}_{ii}^4\hat{B}_i & \hat{A}_{ii}^5\hat{B}_i
                     \end{bmatrix}\\
                     &= \begin{bmatrix}
                       B_{i} & A_{ii}B_i & A_{ii}^2B_i & A_{ii}^3B_i & A_{ii}^4B_i & A_{ii}^5B_i \\
                       0 & -H_iC_iB_i & -H_iC_iA_{ii}B_i & -H_iC_iA_{ii}^2B_i & -H_iC_iA_{ii}^3B_i & -H_iC_iA_{ii}^4B_i
                     \end{bmatrix}\\
                     &= \underbrace{\begin{bmatrix}
                         A_{ii} & B_i \\ -H_iC_i & 0 
                       \end{bmatrix}}_{\hat{M}^C_{i,1}}\underbrace{\begin{bmatrix}
                         0 & B_i & A_{ii}B_i & A_{ii}^2B_i & A_{ii}^3B_i & A_{ii}^4B_i \\ I & 0 & 0 & 0 & 0 & 0 
                       \end{bmatrix}}_{\hat{M}^C_{i,2}}.
                   \end{aligned}
                 \end{equation}
                 Matrices $\hat{M}^C_{i,1}$ and $\hat{M}^C_{i,2}$ have always full rank, since all electrical parameters are positive, hence $\rank(\hat{M}^C_i)=6$. Therefore the pair $(\hat{A}_{ii},\hat{B}_i)$ is controllable.
               \end{proof}
               The overall augmented system is obtained from \eqref{eq:modelDGUgen-aug} as
               \begin{equation}
                 \label{eq:sysaugoverall}
                 \left\lbrace
                   \begin{aligned}
                     \mbf{\dot{\hat{x}}}(t) &= \mbf{\hat{A}\hat{x}}(t) + \mbf{\hat{B}u}(t)+ \mbf{\hat{M}\hat{d}}(t)\\
                     \mbf{\hat{y}}(t)       &= \mbf{\hat{C}\hat{x}}(t)\\
                     \mbf{z}(t)       &= \mbf{\hat{H}\hat{y}}(t)
                   \end{aligned}
                 \right.
               \end{equation}
               where $\mbf{\hat{x}}$, $\mbf{\hat{y}}$ and $\mbf{\hat{d}}$ collect variables $\subss{\hat{x}}{i}$, $\subss{\hat{y}}{i}$ and $\subss{\hat{d}}{i}$ respectively, and matrices $\mbf{\hat{A}}, \mbf{\hat{B}}, \mbf{\hat{C}}, \mbf{\hat{M}}, \mbf{\hat{H}}$ collect matrices in \eqref{eq:augith}.

	  \subsection{Decentralized PnP control based on neutral interactions}
               In this section, we present a decentralized control approach that ensures asymptotic stability for the augmented system \eqref{eq:sysaugoverall}. Furthermore, local controllers are synthesized in a decentralized fashion allowing PnP operations. As shown in Appendix \ref{sec:AppUnstable}, decentralized design of local controllers can fail to guarantee voltage and frequency stability of the whole ImG, if coupling among DGUs is neglected.\\
               We equip each DGU $\subss{\hat{\Sigma}}{i}^{DGU}$ with the following state-feedback controller
               \begin{equation}
                 \label{eq:ctrldec}
                 \subss{\CC}{i}:\qquad \subss{u}{i}(t)=K_{i}\subss{\hat{y}}{i}(t)=K_{i}\subss{\hat{x}}{i}(t)
               \end{equation}
               where $K_{i}\in\Rset^{2\times6}$. Note that controllers $\subss{\CC}{i}$, $i\in\DD$ are decentralized since the computation of $\subss{u}{i}(t)$ requires the state of $\subss{\hat{\Sigma}}{i}^{DGU}$ only. Let nominal subsystems be given by \eqref{eq:modelDGUgen-aug} without coupling terms. We design local controllers $\subss{\CC}{i}$ such that the nominal closed-loop subsystem  
               \begin{equation}
                 \label{eq:modelDGUgen-aug-closed}
                 \left\lbrace
                   \begin{aligned}
                     \subss{\dot{\hat{x}}}{i}(t) &= (\hat{A}_{ii}+ \hat{B}_{i}K_{i})\subss{\hat{x}}{i}(t)+\hat{M}_{i}\subss{\hd}{i}(t)\\
                     \subss{\hat{y}}{i}(t)       &= \hat{C}_{i}\subss{\hat{x}}{i}(t)\\
                     \subss{z}{i}(t)       &= \hat{H}_{i}\subss{\hat{y}}{i}(t)
                   \end{aligned}
                 \right.
               \end{equation}
               is asymptotically stable. From Lyapunov theory, we can achieve this aim if there exists a symmetric matrix $P_{i}\in\Rset^{6\times6}, P_{i}\ge 0$ such that 
               \begin{equation}
                 \label{eq:Lyapeqnith}
                 (\hat{A}_{ii}+\hat{B}_{i}K_{i})^T P_{i}+P_{i}(\hat{A}_{ii}+\hat{B}_{i}K_{i})<0.
               \end{equation}
               Using \eqref{eq:sysaugoverall} and \eqref{eq:ctrldec}, the closed-loop system (accounting for coupling terms) is
               \begin{equation}
                 \label{eq:sysaugoverallclosed}
                 \left\lbrace
                   \begin{aligned}
                     \mbf{\dot{\hat{x}}}(t) &= (\mbf{\hat{A}+\hat{B}K})\mbf{\hat{x}}(t)+ \mbf{\hat{M}\hd}(t)\\
                     \mbf{\hat{y}}(t)       &= \mbf{\hat{C}\hat{x}}(t)\\
                     \mbf{z}(t)       &= \mbf{\hat{H}\hat{y}}(t)
                   \end{aligned}
                 \right.
               \end{equation}
               where $\mbf{K}=\diag(K_{1},\dots,K_{N})$. In the sequel, we will refer to \eqref{eq:sysaugoverallclosed} as the closed-loop QSL microgrid. System \eqref{eq:sysaugoverallclosed} is asymptotically stable if the matrix $\mbf{P}=\diag(P_{1},\dots,P_{N})$ satisfies 
               \begin{equation}
                 \label{eq:Lyapeqnoverall}
                 (\mbf{\hat{A}}+\mbf{\hat{B}K})^T \mbf{P}+\mbf{P}(\mbf{\hat{A}}+\mbf{\hat{B}K})<0.
               \end{equation}
               Note that \eqref{eq:Lyapeqnith} does not imply \eqref{eq:Lyapeqnoverall}, i.e. coupling terms might spoil stability of \eqref{eq:sysaugoverallclosed}. In order to derive conditions such that \eqref{eq:Lyapeqnith} guarantees \eqref{eq:Lyapeqnoverall} we will exploit the concept of neutral interactions between subsystems (see Chapter 7 in \cite{Lunze1992}). To this purpose let us define $\mbf{\hat{A}_{D}}=\diag(\hat{A}_{ii},\dots,\hat{A}_{NN})$ and $\mbf{\hat{A}_{C}=\hat{A}-\hat{A}_{D}}$.
               \begin{defn}
                 \label{def:neutrality}
                 DGU interactions are \emph{neutral} if the matrix $\mbf{\hat{A}_{C}}$ can be factorized as
                 \begin{equation}
                   \mbf{\hat{A}_{C}=SP}
                 \end{equation}
                 where $\mbf{S}$ is a skew-symmetric matrix (i.e. $\mbf{S=-S^T}$).
               \end{defn}
               In order to ensure asymptotic stability of \eqref{eq:sysaugoverallclosed}, we will exploit the following assumptions.
               \begin{assum} 
                 \label{ass:ctrl}
                 \begin{enumerate}[(i)]
                 \item\label{assum:line} The shunt capacitances at all PCCs are identical, i.e. $C_{ti}=C_{s}$, $\forall i\in\DD$.
                 \item\label{assum:pstruct} Decentralized controllers $\subss{\CC}{i}$, $i\in\DD$ are designed such that \eqref{eq:Lyapeqnith} holds with 
                   \begin{equation}
                     \label{eq:pstruct}
                     P_{i}=\left( \begin{array}{cc|cccc}
                         \eta & 0 & 0 & 0& 0 & 0\\
                         0 & \eta & 0 & 0& 0 & 0\\ \hline
                         0 & 0 & \bullet & \bullet& \bullet & \bullet\\
                         0 & 0 & \bullet & \bullet& \bullet & \bullet\\
                         0 & 0 & \bullet & \bullet& \bullet & \bullet\\
                         0 & 0 & \bullet & \bullet& \bullet & \bullet
                       \end{array}\right).
                   \end{equation}
                   where $\bullet$ denotes arbitrary entries and $\eta>0$ is a parameter common to all matrices $P_i,~i\in\DD$.
                 \item\label{assum:etasmallest} It holds $\frac{\eta R_{ij}}{C_{s}Z_{ij}^2}\approx 0$, $\forall i\in\DD$, $\forall j\in\NN_i$, where $Z_{ij}=\abs{R_{ij}+j\omega_0L_{ij}}$.
                 \end{enumerate}
               \end{assum}
               Assumption \ref{ass:ctrl}-(\ref{assum:line}) provides a reference ImG model for which closed-loop asymptotic stability will be shown in Proposition \ref{prop:ctrldec} below. However, integrators in the control loop guarantee robustness of stability \cite{Skogestad1996} with respect to small deviations of capacitances $C_{ti}$ from the common value $C_s$. This feature will be shown through simulations in Section \ref{sec:differentC}. We also highlight that in electrical networks sometimes there are blocks of capacitors positioned at various PCCs that can be switched in steps so as to tune their total capacitance. In this case, Assumption \ref{ass:ctrl}-(\ref{assum:line}) can be directly fulfilled. As for Assumption \ref{ass:ctrl}-(\ref{assum:pstruct}) we show later that checking the existence of $P_i$ as in \eqref{eq:pstruct} and $K_i$ fulfilling \eqref{eq:Lyapeqnith} amounts to solving a convex optimization problem. Here, we just highlight that $\eta>0$ and $C_s>0$ are the only global parameters that must be known for designing local controllers. 
               \begin{rmk}
                 Assumption \ref{ass:ctrl}-(\ref{assum:etasmallest}) can be fulfilled in different ways. When an upper bound to all ratios $\frac{R_{ij}}{Z_{ij}^2}$ (which depend upon line parameters only) is known, it is enough to set the control design parameter $\eta$ sufficiently small. If, however, lines are mainly inductive, one has $\frac{R_{ij}}{Z_{ij}^2}\approx 0$ by construction and bigger values of $\eta$ can be used for synthesizing local controllers.
               \end{rmk}
               \begin{prop}
                 \label{prop:ctrldec}
                 Let Assumption \ref{ass:ctrl} holds. Then, DGU interactions are neutral and the overall closed-loop system \eqref{eq:sysaugoverallclosed} is asymptotically stable.
               \end{prop}               
               \begin{proof}
                 We have to prove that \eqref{eq:Lyapeqnoverall} holds, i.e.
                 \begin{equation}
                   \label{eq:Lyap2part}
                   \underbrace{\mbf{(\hat{A}_{D}+\hat{B}K)}^T \mbf{P+P(\hat{A}_{D}+\hat{B}K)}}_{(a)}+\underbrace{\mbf{\hat{A}_{C}}^T \mbf{P+P\hat{A}_{C}}}_{(b)}<0.
                 \end{equation}
                 Note that term $(a)$ is a block diagonal matrix that collects on the diagonal all left hand sides of \eqref{eq:Lyapeqnith}. Hence term $(a)$ is a negative definite matrix. Next we prove that term $(b)$ is zero. Considering the terms $P_j\hat{A}_{ij}$ and using Assumption \ref{ass:ctrl}-(\ref{assum:etasmallest}), we obtain
                 \begin{equation}
                   \label{eq:couplingslyap}
                   P_i\hat{A}_{ij}=P_j\hat{A}_{ji}=
                   \renewcommand\arraystretch{2}
                   \begin{pmatrix}
                     \frac{\eta R_{ij}}{C_sZ_{ij}^2} & \frac{\eta X_{ij}}{C_sZ_{ij}^2}   & 0 & 0& 0 & 0 \\
                     -\frac{\eta X_{ij}}{C_sZ_{ij}^2} & \frac{\eta R_{ij}}{C_sZ_{ij}^2} & 0 & 0& 0 & 0 \\
                     0 & 0 & 0 & 0& 0 & 0\\
                     0 & 0 & 0 & 0& 0 & 0\\
                     0 & 0 & 0 & 0& 0 & 0\\
                     0 & 0 & 0 & 0& 0 & 0
                   \end{pmatrix}\approx \begin{pmatrix}
                     0 & \frac{\eta X_{ij}}{C_sZ_{ij}^2}   & 0 & 0& 0 & 0 \\
                     -\frac{\eta X_{ij}}{C_sZ_{ij}^2} & 0 & 0 & 0& 0 & 0 \\
                     0 & 0 & 0 & 0& 0 & 0\\
                     0 & 0 & 0 & 0& 0 & 0\\
                     0 & 0 & 0 & 0& 0 & 0\\
                     0 & 0 & 0 & 0& 0 & 0
                   \end{pmatrix},
                 \end{equation}
                 where $X_{ij}=\omega_0L_{ij}$. Hence term $(b)$ in \eqref{eq:Lyap2part} can be approximated using blocks given in \eqref{eq:couplingslyap}. In the following, with a little abuse of notation, we consider coupling terms $\hat{A}_{ij}$ with zero elements on the diagonal. In this case, there always exists a skew-symmetric matrix $S_{ij}$ such that $\hat{A}_{ij}=S_{ij}P_{j}$, e.g. $S_{ij}=\frac{1}{\eta}\hat{A}_{ij}$. Hence matrix $\mbf{S}$ composed of blocks $S_{ij}$ is skew-symmetric and such that $\mbf{\hat{A}_{C}=SP}$. Since $\mbf{\hat{A}_{C}}^T =\mbf{-PS}$, for term $(b)$ we have
                 \begin{equation*}
                   \begin{aligned}
                     (b)&= \mbf{\hat{A}_{C}}^T \mbf{P+P\hat{A}_{C}}\\
                     &= \mbf{PS}^T \mbf{P+PSP}\\
                     &= \mbf{-PSP+PSP}=0
                   \end{aligned}
                 \end{equation*}
                 We have then shown that inequality \eqref{eq:Lyap2part} holds.
               \end{proof}

               The main problem that still has to be solved for designing local controller $\subss{\CC}{i}$ is the following one.
               \begin{prbl}
                 \label{prbl:designPrbl}
                 Compute a matrix $K_{i}$ such that system \eqref{eq:modelDGUgen-aug-closed} is asymptotically stable and Assumption \ref{ass:ctrl}-(\ref{assum:pstruct}) is verified, i.e. \eqref{eq:Lyapeqnith} holds for a matrix $P_i$ structured as in \eqref{eq:pstruct}.
               \end{prbl}
               We solve Problem \ref{prbl:designPrbl} computing a matrix $K_{i}$, verifying  
               \begin{equation}
                 \label{eq:Lyapdecr}
                 (\hat{A}_{ii}+\hat{B}_{i}K_{i})^{T}P_{i}+P_{i}(\hat{A}_{ii}+\hat{B}_{i}K_{i})+\gamma_{i}^{-1}I\le 0
               \end{equation} 
               where $P_{i}$ is defined in \eqref{eq:pstruct} and $\gamma_{i}>0$ enforces a certain degree of robust stability for the origin of the closed-loop subsystem \eqref{eq:modelDGUgen-aug}, see \cite{Boyd1994}. Using the Schur complement, we can rewrite \eqref{eq:Lyapdecr} as 
               \begin{equation}
                 \label{eq:LyapSchur}
                 \begin{bmatrix}
                   (\hat{A}_{ii}+\hat{B}_{i}K_{i})^{T}P_{i}+P_{i}(\hat{A}_{ii}+\hat{B}_{i}K_{i}) & I\\
                   I & -\gamma_{i}I
                 \end{bmatrix}\le 0
               \end{equation}
               This inequality is nonlinear in $P_{i}$ and $K_{i}$. As in \cite{Boyd1994}, we introduce new matrices
               \begin{equation}
                 \label{eq:eqschur}
                 \begin{aligned}
                   Y_{i}&=P_{i}^{-1}\\
                   G_{i}&=K_{i}Y_{i}.
                 \end{aligned}
               \end{equation}
               Note that $Y_{i}$ has the same structure of $P_{i}$. By pre- and post-multiplying \eqref{eq:LyapSchur} with $\left[\begin{smallmatrix}
                   Y_{i} & 0\\
                   0 & I 
                 \end{smallmatrix}\right]$ and using \eqref{eq:eqschur} we obtain
               \begin{equation}
                 \begin{bmatrix}
                   Y_{i}\hat{A}_{ii}^{T}+G_{i}^{T}\hat{B}_{i}^{T}+\hat{A}_{ii}Y_{i}+\hat{B}_{i}G_{i} & Y_{i} \\
                   Y_{i} & -\gamma_{i}I
                 \end{bmatrix}\le 0
               \end{equation}
               Consider the following optimization problem
               \begin{subequations}
                 \label{eq:optproblem}
                 \begin{align}
                  \mathcal{O}: \min_{\substack{Y_{i},G_{i},\gamma_{i},\beta_{i},\delta_{i}}}\quad &\alpha_{i1}\gamma_{i}+\alpha_{i2}\beta_{i}+\alpha_{i3}\delta_{i}\nonumber \\
                   \label{eq:Ystruct}&Y_{i}=\left[ \begin{smallmatrix}
                       \eta^{-1} & 0 & 0 & 0& 0 & 0\\
                       0 & \eta^{-1} & 0 & 0& 0 & 0\\ 
                       0 & 0 & \bullet & \bullet& \bullet & \bullet\\
                       0 & 0 & \bullet & \bullet& \bullet & \bullet\\
                       0 & 0 & \bullet & \bullet& \bullet & \bullet\\
                       0 & 0 & \bullet & \bullet& \bullet & \bullet
                     \end{smallmatrix}\right]>0\\
                   \label{eq:LMIstab}&\begin{bmatrix}
                     Y_{i}\hat{A}_{ii}^{T}+G_{i}^{T}\hat{B}_{i}^{T}+\hat{A}_{ii}Y_{i}+\hat{B}_{i}G_{i} & Y_{i} \\
                     Y_{i} & -\gamma_{i}I
                   \end{bmatrix}\le 0\\
                   \label{eq:Gcostr}&\begin{bmatrix}
                     -\beta_{i}I & G_{i}^T\\
                     G_{i} & -I
                   \end{bmatrix}<0\\
                   \label{eq:Ycostr}&\begin{bmatrix}
                     Y_{i} & I\\
                     I & \delta_{i}I
                   \end{bmatrix}>0\\
                   &\gamma_{i}>0,\quad\beta_{i}>0\quad\delta_{i}>0
                 \end{align}
               \end{subequations} 
               where $\alpha_{i1}$, $\alpha_{i2}$ and $\alpha_{i3}$ represent positive weights and $\bullet$ are arbitrary entries.

               All constraints in \eqref{eq:optproblem} are Linear Matrix Inequalities (LMI) and therefore the optimization problem is convex and it can be efficiently solved with state-of-art LMI solvers \cite{Boyd1994}. Note that constraint \eqref{eq:Ystruct} guarantees that $P_i$ has the structure prescribed by Assumption \ref{ass:ctrl}-(\ref{assum:pstruct}). Moreover, stability of the nominal closed-loop subsystem \eqref{eq:modelDGUgen-aug-closed} is guaranteed by \eqref{eq:LMIstab}. In order to prevent $\norme{K_{i}}{2}$ from becoming too large we add the bounds $\norme{G_{i}}{2}<\sqrt{\beta_{i}}$ and $\norme{Y_{i}^{-1}}{2}<\delta_{i}$ that, via Schur complement, correspond to constraints \eqref{eq:Gcostr} and \eqref{eq:Ycostr}. These bound imply $\norme{K_{i}}{2}<\sqrt{\beta_i}\delta_{i}$ and then affect the magnitude of control variables. 

               We highlight that the constraints in \eqref{eq:optproblem} depend upon local fixed matrices ($\hat{A}_{ii}$, $\hat{B}_{i}$) and local design parameters ($\alpha_{i1}$, $\alpha_{i2}$, $\alpha_{i3}$). Therefore, once global parameters $\eta$ and $C_{s}$ are fixed, the computation of controller $\subss{\CC}{i}$ does not influence the computation of controllers $\subss{\CC}{j}, j\neq i$. If problem $\PP_i$ is feasible, then we obtain controller $\subss\CC i$ through $K_{i}=G_{i}Y_{i}^{-1}$. Moreover since all assumptions in Proposition \ref{prop:ctrldec} are verified, the overall closed-loop system \eqref{eq:sysaugoverallclosed} is asymptotically stable.

          \subsection{Enhancements of local controllers for improving performances}
               In the previous section we have shown how to design decentralized controllers $\subss{\CC}{i}$ guaranteeing asymptotic stability for the overall closed-loop system \eqref{eq:sysaugoverallclosed}. For improving performance in transient, we enhance controllers $\subss{\CC}{i}$ with feed-forward terms to
               \begin{enumerate}[(i)]
               \item pre-filter reference signals;
               \item compensate measurable disturbances.
               \end{enumerate}

               \subsubsection{Pre-filtering of reference signal}
                    \label{sec:pre-filter}
                    Pre-filtering is used to widen the bandwidth so as to speed up the response of the system. For each nominal closed-loop subsystem \eqref{eq:modelDGUgen-aug-closed}, the transfer function $\subss{F}{i}(s)$, from $\subss{z_{ref}}{i}(t)$ to the controlled variable  $\subss{z}{i}(t)$ is
                    \begin{equation}
                      \label{eq:F(s)}
                      \subss{F}{i}(s)=(\hat{H}_i \hat{C}_i) (sI-(\hat{A}_{ii}+\hat{B}_{i}K_i ))^{-1}\begin{bmatrix}
                        0 \\
                        I_2
                      \end{bmatrix} 
                    \end{equation}
                    Using a feedforward compensator $\subss{\tilde{C}}{i}(s)$, we filter the reference signal $\subss{z_{ref}}{i}(t)$ as shown in Figure \ref{fig:prefilter}. 
                    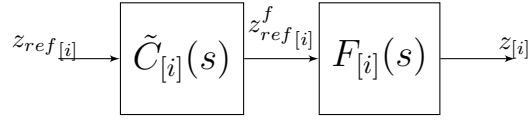
\begin{figure}[!htb]
                      \centering
                      \begin{tikzpicture}
  \begin{Large}
  \sbEntree{zrefi}
  \sbBloc{compensatori}{$\subss{\tilde{C}}{i}(s)$}{zrefi}
  \sbBloc{Fi}{$\subss{F}{i}(s)$}{compensatori}
  \sbSortie{zi}{Fi}
  \sbRelier{zrefi}{compensatori}
  \sbRelier{Fi}{zi}
  \end{Large}
  \sbNomLien{zrefi}{$\subss{z_{ref}} i$}
  \sbNomLien{zi}{$\subss z i$}
  \sbRelier[$\subss{z_{ref}^f} i$]{compensatori}{Fi}
\end{tikzpicture}
                      \caption{Block diagram of closed-loop DGU $i$ with prefilter.}
                      \label{fig:prefilter}
                    \end{figure}
                    The new transfer function from $\subss{z_{ref}}{i}(t)$ to $\subss{z}{i}(t)$ is
                    \begin{equation}
                      \label{eq:C(s)F(s}
                      \subss{\tilde{F}}{i}(s)=\subss{\tilde{C}}{i}(s)\subss{F}{i}(s)
                    \end{equation}
                    Defining a  desired transfer function $\subss{\tilde{F}}{i}(s)$ for each subsystem, we can compute, from \eqref{eq:C(s)F(s}, the pre-filter $\subss{\tilde{C}}{i}(s)$ as
                    \begin{equation}
                      \label{eq:prefilter}
                      \subss{\tilde{C}}{i}(s) = \subss{\tilde{F}}{i}(s)\subss{\hat{F}}{i}(s)^{-1}
                    \end{equation}
                    under the following conditions \cite{Skogestad1996}:
                    \begin{itemize}
                    \item $\subss{\hat{F}}{i}(s)$ must not have Right-Half-Plane (RHP) zeros that would become RHP poles of $\subss{\tilde{C}}{i}(s)$, making it unstable;
                    \item $\subss{\hat{F}}{i}(s)$ must not contain a time delay, otherwise $\subss{\tilde{C}}{i}(s)$ would have a predictive action
                    \item $\subss{\tilde{C}}{i}(s)$ must be realizable, i.e. it must have more poles than zeros.
                    \end{itemize}
                    If these conditions are verified, the filter $\subss{\tilde{C}}{i}(s)$ given by \eqref{eq:prefilter} is realizable and asymptotically stable (this condition is essential since $\subss{\tilde{C}}{i}(s)$ works in open-loop). Furthermore, since $\subss{\hat{F}}{i}(s)$ is asymptotically stable (thanks to the stabilizing controllers $\subss{\CC}{i}$ designed solving the problem $\PP_i$), the closed-loop system including filters $\subss{\tilde{C}}{i}(s)$ is asymptotically stable as well. Moreover, when some of the previous conditions do not hold, formula \eqref{eq:prefilter} cannot be used. Still, the compensator $\subss{\tilde{C}}{i}(s)$ can be designed for a given bandwidth, as shown in \cite{Skogestad1996}. 

	       \subsubsection{Compensation of measurable disturbances}
	            \label{sec:compensator}
                    The second enhancement is the compensation of measurable disturbances. Since we assumed load dynamics is not known, we have modeled, for each subsystem, the \emph{dq} components of load currents as a measurable disturbance $\subss{d}{i}(t)$. Considering new local controllers $\subss{\tilde{\CC}}{i}$ defined as
                    \begin{equation}
                      \subss{\tilde{\CC}}{i}:\quad \subss{u}{i}=K_{i}\subss{\hat{x}}{i}(t)+\subss{\tilde{u}}{i}(t)
                    \end{equation}
                    that are just controllers $\subss{\CC}{i}$ in \eqref{eq:ctrldec} with the additional term $\subss{\tilde{u}}{i}(t)$, we can rewrite \eqref{eq:modelDGUgen-aug-closed} as
                    \begin{equation}
                      \label{eq:subsysDGi-thAUGCLComp}
                      \subss{\tilde{\Sigma}}{i}^{DGU} :
                      \left\lbrace
                        \begin{aligned}
                          \subss{\dot{\hat{x}}}{i}(t) &= (\hat{A}_{ii}+ \hat{B}_{i}K_{i})\subss{\hat{x}}{i}(t)+\hat{M}_{i}\subss{\hat{d}}{i}(t)+\hat{B_{i}}\subss{\tilde{u}}{i}(t)\\
                          \subss{\hat{y}}{i}(t)       &= \hat{C}_{i}\subss{\hat{x}}{i}(t)\\
                          \subss{z}{i}(t)       &= \hat{H}_{i}\subss{\hat{y}}{i}(t)
                        \end{aligned}
                      \right..
                    \end{equation}
                    We now use the new input $\subss{\tilde{u}}{i}(t)$ to compensate the measurable disturbance $\subss{d}{i}(t)$ (recall that $\subss{\hat d}{i} = [\subss{d^T}{i}~\subss{z_{ref}^T}{i} ]^T$). From \eqref{eq:subsysDGi-thAUGCLComp}, the transfer function from the disturbance $\subss{d}{i}(t)$ to the controlled variable $\subss{z}{i}(t)$ is
                    \begin{equation}
                      \label{eq:Gd(s)}
                      G^{d}_{i}(s)=(\hat{H}_i \hat{C}_i)(sI-(\hat{A}_{ii}+\hat{B}_{i}K_i))^{-1}			\begin{bmatrix}
                        M_{i}\\
                        0
                      \end{bmatrix}. 
                    \end{equation}
                    Moreover, the transfer function from the new input $\subss{\tilde{u}}{i}(t)$ to the controlled variable $\subss{z}{i}(t)$ is
                    \begin{equation}
                      \label{eq:G(s)}
                      G_{i}(s)=(\hat{H}_i \hat{C}_i)(sI-(\hat{A}_{ii}+\hat{B}_{i}K_i ))^{-1}\hat{B}_i . 
                    \end{equation}
                    If we combine \eqref{eq:Gd(s)} and \eqref{eq:G(s)}, we obtain
                    \begin{equation}
                      \subss{z}{i}(s)=G_{i}(s)\subss{\tilde{u}}{i}(s)+G^{d}_{i}(s)\subss{d}{i}(s).
                    \end{equation}
                    In order to zero the effect of the disturbance on the controlled variable, we set
                    \begin{equation}
                      \subss{\tilde{u}}{i}(s)=N_{i}(s)\subss{d}{i}(s)
                    \end{equation}
                    where 
                    \begin{equation}
                      \label{eq:compensator}
                      \subss{N}{i}(s)=-G_{i}(s)^{-1}G^{d}_{i}(s)
                    \end{equation}
                    is the transfer function of the compensator. Note that $\subss{N}{i}(s)$ is well defined under the following conditions \cite{Skogestad1996}:
                    \begin{itemize}
                    \item $\subss{G}{i}(s)$ must not have RHP zeros that would become RHP poles of $\subss{N}{i}(s)$;
                    \item $\subss{G}{i}(s)$ must not contain a time delay, otherwise $\subss{N}{i}(s)$ would have a predictive action
                    \item $\subss{N}{i}(s)$ must be realizable, i.e. it must have more poles than zeros.
                    \end{itemize}
                    In this way, we can ensure that the compensator $\subss{N}{i}(s)$ is asymptotically stable, hence preserving asymptotic stability of the system. When some of the previous conditions do not hold, formula \eqref{eq:compensator} cannot be used and perfect compensation cannot be achieved. Still, the compensator $\subss{N}{i}(s)$ can be designed to reject disturbances within a given bandwidth, as shown in \cite{Skogestad1996}. The overall control scheme with the addition of the compensators is shown in Figure \ref{fig:compensator}.
                    \begin{figure}[!htb]
                      \centering
                      \tikzstyle{input} = [coordinate]
\tikzstyle{output} = [coordinate]
\tikzstyle{guide} = []
\tikzstyle{block} = [draw, rectangle, minimum height=1cm]

\begin{tikzpicture}
  \sbEntree{zref1}
  \sbDecaleNoeudy[3]{zref1}{zrefj}
  \sbDecaleNoeudy[3]{zrefj}{zrefN}
  \node [block, right of=zrefj,node distance=11cm,minimum height=4cm, minimum width=2cm] (microgrid) {\textbf{Microgrid}};
  \node [guide, right of=zrefj,xshift=-0.5cm] (zrefjline) {};
  \draw [draw] (zrefjline) -| node{$\vdots$} (zrefjline);
  
  \sbComph{sumret1}{zref1}
  \sbBloc{integrator1}{$\int dt$}{sumret1}
  \sbBloc[2.5]{controller1}{$K_1$}{integrator1}
  \sbSumb[5]{sumdist1}{controller1}
  \sbRelier[$\subss{z_{ref}}{1}$]{zref1}{sumret1}
  \sbRelier{sumret1}{integrator1}
  \sbRelier[$\subss v 1$]{integrator1}{controller1}
  \sbRelier{controller1}{sumdist1}
  \node [guide, left of=microgrid,yshift=1.05cm,xshift=0.125cm] (u1end) {};
  \sbRelier[$\subss u 1$]{sumdist1}{u1end}
  
  \sbComp{sumretN}{zrefN}
  \sbBloc{integratorN}{$\int dt$}{sumretN}
  \sbBloc[4.5]{controllerN}{$K_N$}{integratorN}
  \sbSumb[7.5]{sumdistN}{controllerN}
  \sbRelier[$\subss{z_{ref}}{N}$]{zrefN}{sumretN}
  \sbRelier{sumretN}{integratorN}
  \sbRelier[$\subss v N$]{integratorN}{controllerN}
  \sbRelier{controllerN}{sumdistN}
  \node [guide, left of=microgrid,yshift=-1.05cm,xshift=0.125cm] (uNend) {};
  \sbRelier[$\subss u N$]{sumdistN}{uNend}
  
  \node [output, below of=microgrid,yshift=-3.0cm,xshift=0.5cm] (d1out) {};
  \node [output, below of=microgrid,yshift=-2.8cm,xshift=0.5cm] (d1outnear) {};
  \node [output, below of=microgrid,yshift=-1.0cm,xshift=0.5cm] (d1outstart) {};
  \draw [draw,->,>=latex'] (d1out) -| node[yshift=-0.2cm]{$\subss d 1$} (d1outstart);
  \node [block, below of=uNend,yshift=-1.75cm,xshift=0.0cm] (N1) {$N_1(s)$};
  \draw [draw,->,>=latex',near end,swap] (d1outnear) -- (N1);
  \draw [draw,->,>=latex',near end,swap] (N1) -| node[xshift=0.35cm] {$\subss{\tilde u} 1$} (sumdist1);
  
  \node [output, below of=microgrid,yshift=-1.5cm] (djout) {};
  \draw [draw] (djout) -| node {$\ldots$} (djout);
  
  \node [output, below of=microgrid,yshift=-1.8cm,xshift=-0.5cm] (dNout) {};
  \node [output, below of=microgrid,yshift=-1.6cm,xshift=-0.5cm] (dNoutnear) {};
  \node [output, below of=microgrid,yshift=-1.0cm,xshift=-0.5cm] (dNoutstart) {};
  \draw [draw,->,>=latex'] (dNout) -| node[yshift=-0.2cm]{$\subss d N$} (dNoutstart);
  \node [block, below of=sumdistN,yshift=-0.55cm,xshift=0.0cm] (NN) {$N_N(s)$};
  \draw [draw,->,>=latex',near end,swap] (dNoutnear) -- (NN);
  \draw [draw,->,>=latex',near end,swap] (NN) -- node[xshift=0.35cm,yshift=-0.2cm] {$\subss{\tilde u} N$} (sumdistN);

  \node [output, above of=microgrid,yshift=1.8cm,xshift=0.5cm] (y1out) {};
  \node [output, above of=microgrid,yshift=1.0cm,xshift=0.5cm] (y1outstart) {};
  \node [output, above of=microgrid,yshift=1.5cm] (yjout) {};
  \node [output, above of=microgrid,yshift=1.2cm,xshift=-0.5cm] (yNout) {};
  \node [output, above of=microgrid,yshift=1.0cm,xshift=-0.5cm] (yNoutstart) {};
  \draw [draw] (y1outstart) -- node {} (y1out);
  \draw [draw,->,>=latex'] (y1out) -| node[yshift=0.2cm]{$\subss y 1$} (controller1);
  \draw [draw] (yjout) -| node {$\ldots$} (yjout);
  \draw [draw] (yNoutstart) -- node {} (yNout);
  \draw [draw,->,>=latex'] (yNout) -| node[yshift=0.2cm]{$\subss y N$} (controllerN);
  
  \node [guide, right of=microgrid,yshift=1.05cm,xshift=-.125cm] (z1) {};
  \node [guide, right of=microgrid,yshift=1.165cm,xshift=0.5cm] (z1near) {};
  \node [guide, right of=microgrid,yshift=1.05cm,xshift=1cm] (z1end) {};
  \draw [draw,->,>=latex',near end,swap] (z1) -- node[xshift=0.6cm] {$\subss{z} 1$} (z1end);
  \sbRenvoi[-6.8]{z1near}{sumret1}{}

  \node [guide, right of=microgrid,yshift=-1.05cm,xshift=-.125cm] (zN) {};
  \node [guide, right of=microgrid,yshift=-0.95cm,xshift=0.5cm] (zNnear) {};
  \node [guide, right of=microgrid,yshift=-1.05cm,xshift=1cm] (zNend) {};
  \draw [draw,->,>=latex',near end,swap] (zN) -- node[xshift=0.6cm] {$\subss{z} N$} (zNend);
  \sbRenvoi[10]{zNnear}{sumretN}{}
  
  \node [guide, right of=microgrid,xshift=0.5cm] (zj) {};
  \draw [draw] (zj) -| node {$\vdots$} (zj);
  
  \node [guide, left of=microgrid,xshift=-0.5cm] (uj) {};
  \draw [draw] (uj) -| node {$\vdots$} (uj);
  
\end{tikzpicture}
                      \caption{Overall control scheme with compensation of measurable disturbances $\subss{d}{i}(s)$.}
                      \label{fig:compensator}
                    \end{figure}
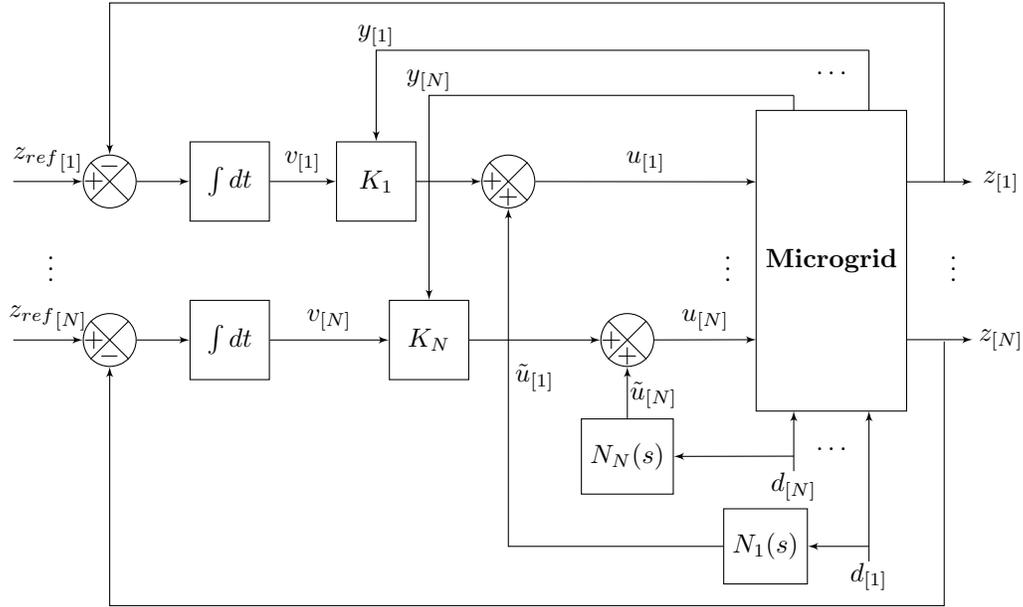

	  \subsection{Algorithm for the design of local controllers}
               \label{sec:algorithm}
                In Algorithm \ref{alg:ctrl_design}, we summarize the procedure for the design of local controller $\subss\CC i$ and compensators $\subss{\tilde{C}}{i}$ and $\subss{N}{i}$.
               \begin{algorithm}[!htb]
                 \caption{Design of controller $\subss{\CC}{i}$ and compensators $\subss{\tilde{C}}{i}$ and $\subss{N}{i}$ for subsystem $\subss{\hat{\Sigma}}{i}^{DGU}$}
                 \label{alg:ctrl_design}
                 \textbf{Input:} DGU $\subss{\hat{\Sigma}}{i}^{DGU}$ as in \eqref{eq:modelDGUgen-aug} and desired closed-loop transfer function $\subss{\tilde{F}}{i}(s)$\\
                 \textbf{Output:} Controller $\subss{\CC}{i}$ and, optionally, pre-filter $\subss{\tilde{C}}{i}$ and compensator $\subss{N}{i}$\\
                 \begin{enumerate}[(A)]
                 \item\label{enu:stepAalgCtrl} Find $K_i$ solving the LMI problem \eqref{eq:optproblem}. If it is not feasible \textbf{stop} (the controller $\subss\CC i$ cannot be designed).\\
                   
                   \textbf{Optional steps}
                 \item\label{enu:stepBalgCtrl} Given the desired closed-loop transfer function $\subss{\tilde{F}}{i}(s)$, design the asymptotically stable local pre-filter $\subss{\tilde{C}}{i}$ as in \eqref{eq:prefilter}. 
                 \item\label{enu:stepCalgCtrl} Design the asymptotically stable local compensator $\subss{N}{i}$ as in \eqref{eq:compensator}.
                 \end{enumerate}
               \end{algorithm}

          \subsection{Plug-and-Play operations}
               \label{sec:PnP}
               In this section, we discuss the operations for updating the controllers when DGUs are added to or removed from the ImG. The goal is to preserve stability of the new closed-loop system. As a starting point, we consider a microgrid composed by subsystems $\subss{\hat{\Sigma}}{i}^{DGU}, i\in\DD$ equipped with local controllers and compensators produced by Algorithm \ref{alg:ctrl_design}.

               \paragraph{Plugging-in operation}
                    Consider the plug-in of a new DGU $\subss{\hat{\Sigma}}{N+1}^{DGU}$ described by matrices, $\hat{A}_{N+1\:N+1}$, $\hat{B}_{N+1}$, $\hat{C}_{N+1}$, $\hat{M}_{N+1}$, $\hat{H}_{N+1}$ and $\{\hat{A}_{N+1\:j}\}_{j\in\NN_{N+1}}$. In particular, $\NN_{N+1}$ identifies the DGUs that are directly coupled to $\subss{\hat{\Sigma}}{N+1}^{DGU}$ through transmission lines and $\{\hat{A}_{N+1\:j}\}_{j\in\NN_{N+1}}$ are the corresponding coupling terms. For designing controller $\subss{\CC}{N+1}$ and compensators $\subss{\tilde{C}}{N+1}$ and $\subss{N}{N+i}$, we execute Algorithm \ref{alg:ctrl_design}. We note that DGUs $\subss{\hat{\Sigma}}{j}^{DGU}$, $j\in\NN_{N+1}$, have the new neighbour $\subss{\hat{\Sigma}}{N+1}^{DGU}$. Therefore,  the redesign of controllers $\subss{\CC}{j}$ and compensators $\subss{\tilde{C}}{j}$ and $\subss{N}{j}$, $\forall j\in\NN_{N+1}$ is needed because matrices $\hat{A}_{jj}$, $j\in\NN_{N+1}$ change (see definition of $\hat{A}_{jj}$ in \eqref{eq:Aii}). 

                    In conclusion, the plug-in of $\subss{\hat{\Sigma}}{N+1}^{DGU}$ is allowed only if Algorithm \ref{alg:ctrl_design} does not stop in Step \ref{enu:stepAalgCtrl} when computing controllers $\subss{\CC}{k}$ for all $k\in\NN_{N+1}\cup\{N+1\}$. Note that, the redesign is not propagated further in the network, i.e. even without changing controllers $\subss{\CC}{i}$, $\subss{\tilde{C}}{i}$ and $\subss{N}{i}$, $i\not\in \{N+1\}\cup\NN_{N+1}$ asymptotic stability of the new overall closed-loop QSL ImG model is ensured.

               \paragraph{Unplugging operation}  
                    We consider the unplugging of DGU $\subss{\hat{\Sigma}}{k}^{DGU}$, $k\in\DD$. Matrix $\hat{A}_{jj}$ of each $\subss{\hat{\Sigma}}{j}^{DGU}$, $ j\in\NN_{k}$ changes due to the disconnection of $\subss{\hat{\Sigma}}{k}^{DGU}$ from the network. For this reason, for each $ j\in\NN_{k}$, the redesign through Algorithm \ref{alg:ctrl_design} of controllers $\subss{\CC}{j}$ and compensators $\subss{\tilde{C}}{j}$ and $\subss{N}{j}$, $j\in\NN_{N+1}$, is needed and unplugging of $\subss{\hat{\Sigma}}{k}^{DGU}$ is allowed only if these operations can be successfully terminated. As for the plugging-in operation, the re-design of local controllers $\subss\CC j$, $j\notin\NN_{k}$ is not required.

                    \begin{rmk}
                      Existing contributions on decentralized control for ImG fall in two main categories. The first one comprises centralized design procedures where the use of the whole ImG model allows one to guarantee voltage and frequency stability \cite{Etemadi2012a,Etemadi2012,Simpson-Porco2013}. The second one embraces decentralized design approaches, often based on droop control, where tuning the parameters of local regulators does not require any piece of global information about the ImG model \cite{Guerrero2007,Moradi2010,Guerrero2013,Shafiee2014}. In this case, however, stability is seldom guaranteed. Our control design algorithm bridges the gap between the above categories, meaning that it is decentralized but, at the same time, capable to provide closed-loop stability. We also highlight that all decentralized design approaches falling in the second category allow for PnP operations. However they do not guarantee stability is preserved when a new DGU is plugged in or out. As regards the first category of contributions, all control architectures that require a centralized design do not allow for PnP operations. Indeed, when a new DGU is added, the execution of centralized design procedure can modify existing controllers of all other DGUs.
                    \end{rmk}

     \clearpage
     
     \section{Simulation results}
          \label{sec:Simresults}
          In this section, we study performance brought about by PnP controllers described in Section \ref{sec:PnPctrl} using two ImGs. Simulations have been conducted in MatLab/Simulink using the SimPowerSystem Toolbox and the PnPMPC-toolbox \cite{Riverso2012g}. First, we consider the ImG in Figure \ref{fig:schemairandist} with only two DGUs and discuss performance in tracking step references as well as robustness to disturbances and various load dynamics, including highly nonlinear and unbalanced loads. Then, we consider the ImG in Figure \ref{fig:10areas} composed of $10$ DGUs and we show that stability of the whole microgrid is guaranteed. Moreover, we give an example of the PnP capabilities of our control approach. Simulation results show that PnP controllers can guarantee good dynamical performance.
      
     	  \subsection{Scenario 1}
               \label{sec:scenario1}
     	       In this Scenario, we consider the ImG in Figure \ref{fig:schemairandist} composed of two DGUs and, for the sake of simplicity, we set $i=1$ and $j=2$. Parameters values are collected in Table \ref{tbl:par2areas} in Appendix \ref{sec:AppElectrPar}. We highlight that they are comparable to those used in \cite{Karimi2008,Moradi2010,Babazadeh2013}. We used standard blocks of the MatLab/SimPowerSystem Toolbox for simulating VSCs and associated filter, for the (Y/$\Delta$) transformer, for the tie-lines and for the loads. Therefore, although control design is based on a linear model, in this and next sections, more realistic (and often nonlinear) models of microgrid components are used for assessing the performance of PnP controllers. For each DGU, we execute Algorithm \ref{alg:ctrl_design} in order to design local controllers and compensators, where the desired closed-loop transfer function $\subss{\tilde{F}}{i}(s)$, $i=\{1,2\}$ has been chosen as a low-pass filter with DC gain equal to 0 dB and bandwidth equal to 1kHz. The choice of this bandwidth is a trade-off between guaranteeing sufficiently fast transients, and not interfering with the modulation frequency of the inverter.\\
               Through Step \ref{enu:stepAalgCtrl} of Algorithm \ref{alg:ctrl_design}, we obtain two decentralized controllers $\subss{\CC}{i}$, $i=\{1,2\}$ that stabilize the overall microgrid. The eigenvalues of the closed-loop QSL microgrid \eqref{eq:sysaugoverallclosed} are represented in Figure \ref{fig:closedLoop2Areas}. Solving Step \ref{enu:stepBalgCtrl} of Algorithm \ref{alg:ctrl_design} we obtain two asymptotically stable local pre-filters $\subss{\tilde{C}}{i}$, $i=\{1,2\}$ whose singular values are shown in Figure \ref{fig:singularValuePreFilter2Areas}. From Figure \ref{fig:singularValueClosedLoop2Areas}, we also note that the singular values of the overall closed-loop transfer function $F(s)$ with the addition of the two pre-filters (in green) coincide with the frequency response of desired closed-loop transfer function $\tilde{F}(s)$ previously defined. In the last Step of the Algorithm \ref{alg:ctrl_design}, we obtain, for each DGU, an asymptotically stable disturbance compensator $\subss{N}{i}$, $i=\{1,2\}$ as in \eqref{eq:compensator}. The singular values of the two compensators are shown in Figure \ref{fig:singularValueCompensator2Areas}. Next, we evaluate the performance of the decentralized voltage control scheme with different tests.
               \begin{figure}[!htb]
                 \centering
                 \begin{subfigure}[!htb]{0.48\textwidth}
                   \centering
                   \includegraphics[width=1\textwidth, height=130pt]{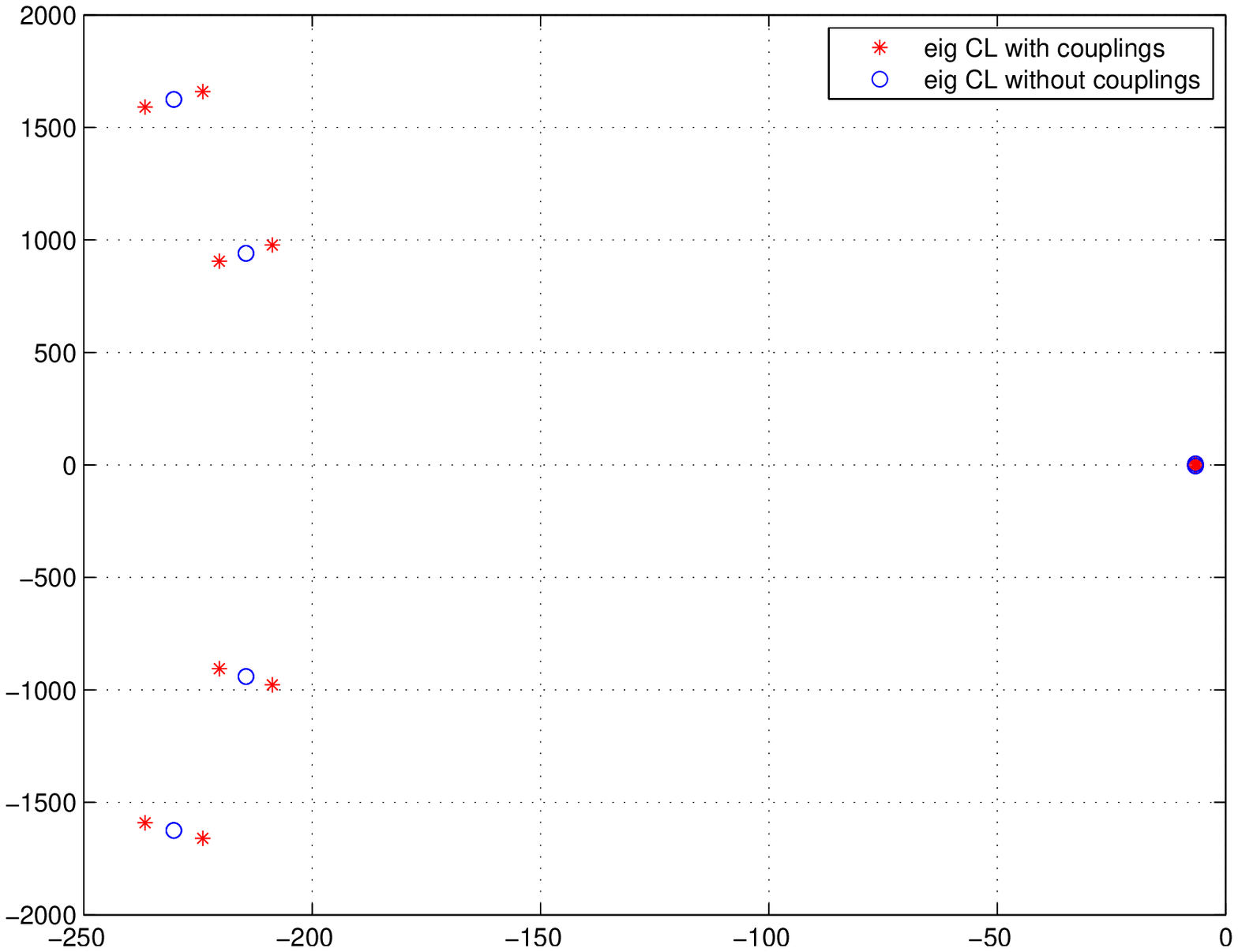}
                   \caption{Eigenvalues of the closed-loop (CL) QSL microgrid with (red) and without (blue) couplings.}
                   \label{fig:closedLoop2Areas}
                 \end{subfigure}
                 \begin{subfigure}[!htb]{0.48\textwidth}
                   \centering
                   \includegraphics[width=1\textwidth, height=130pt]{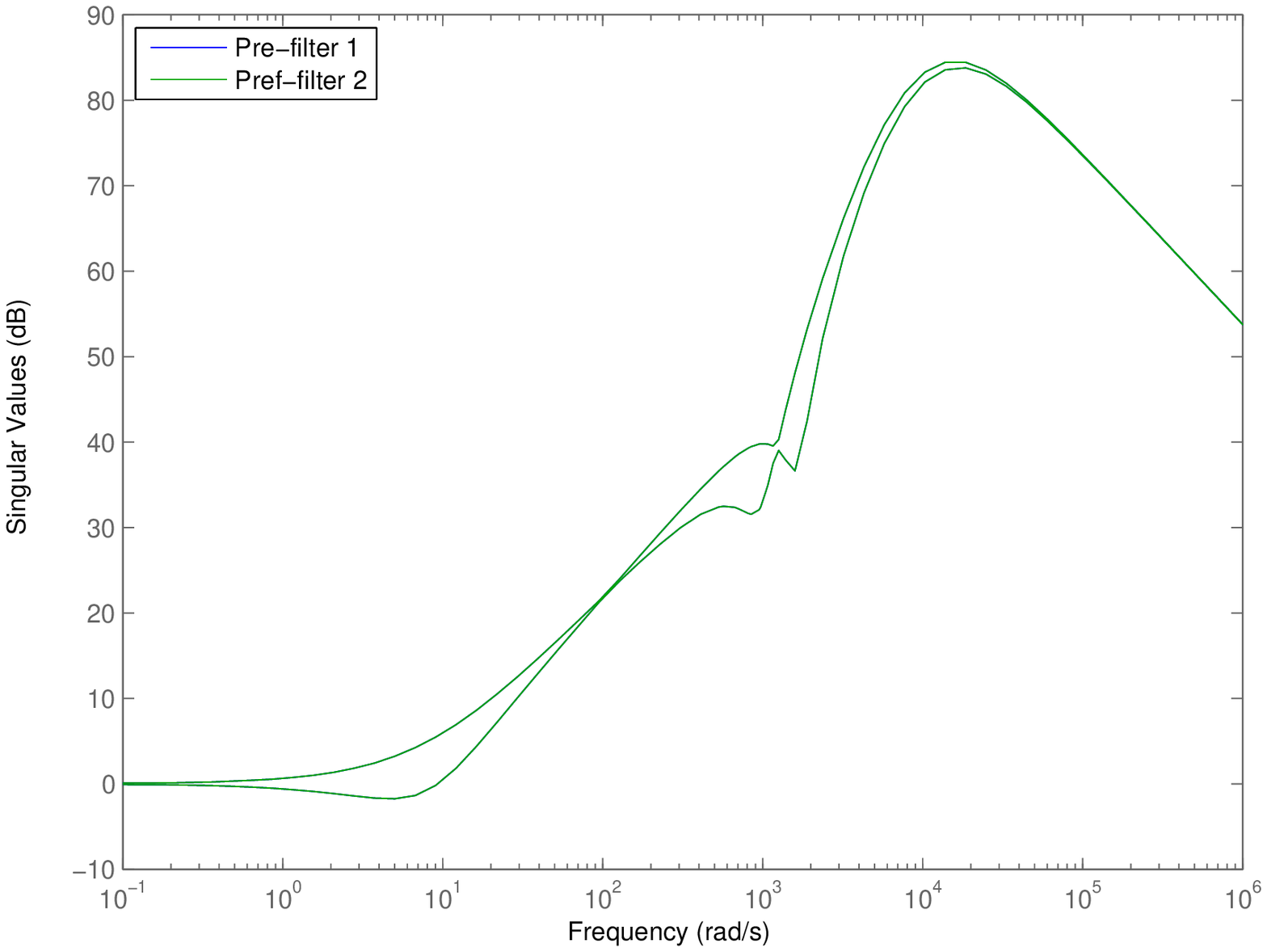}
                   \caption{Singular values of the pre-filters.}
                   \label{fig:singularValuePreFilter2Areas}
                 \end{subfigure}
                 \begin{subfigure}[!htb]{0.48\textwidth}
                   \centering
                   \includegraphics[width=1\textwidth, height=130pt]{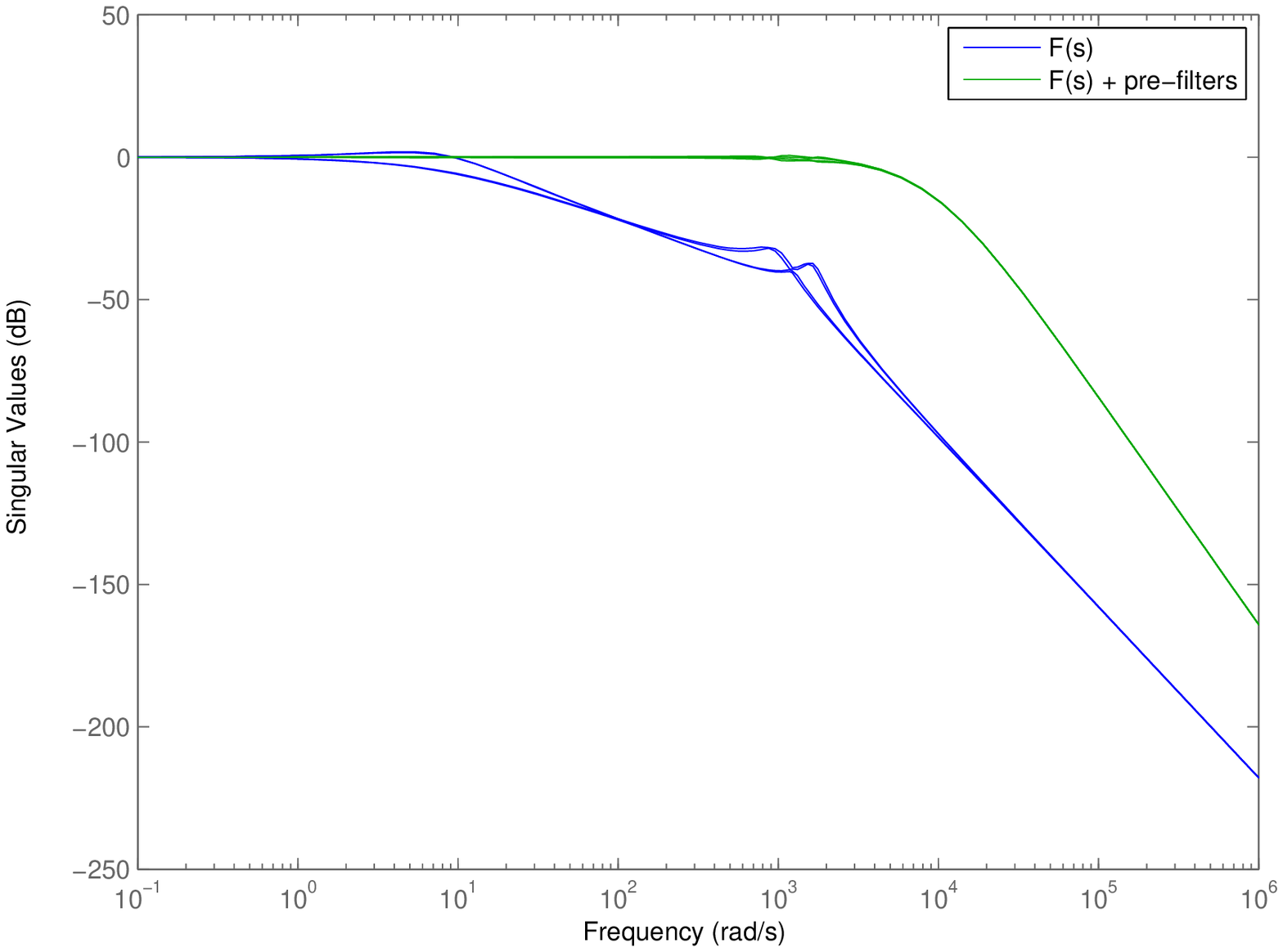}
                   \caption{Singular values of $F(s)$ with (green) and without (blue) pre-filters.}
                   \label{fig:singularValueClosedLoop2Areas}
                 \end{subfigure}
                 \begin{subfigure}[!htb]{0.48\textwidth}
                   \includegraphics[width=1\textwidth, height=130pt]{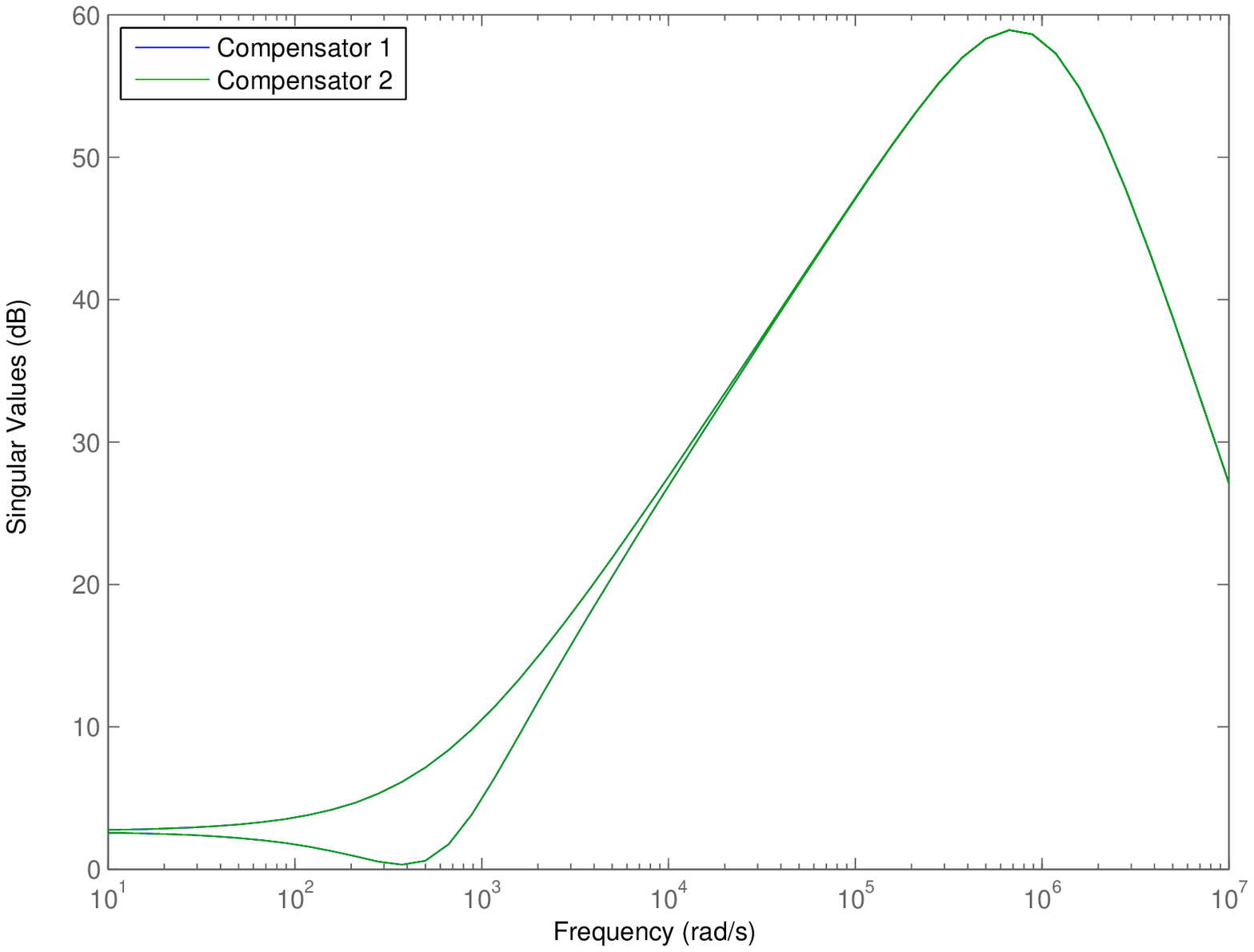}
                   \caption{Singular values of disturbances compensators.}
                   \label{fig:singularValueCompensator2Areas}
                 \end{subfigure}
                 \caption{Features of PnP controllers for Scenario 1.}
                 \label{fig:closedLoop2areas}
               \end{figure}

               \subsubsection{Voltage tracking for DGU 1}
                    \label{sec:scenario1voltagetrack1}
                    In the first test, we assess the performance in tracking step changes in the \emph{dq} voltage reference at $PCC_1$. For each DGU, we use an RL parallel load with constant parameters $R=76~\Omega$ and $L=111.9$ mH. The \emph{d} and \emph{q} components of the voltage at $PCC_{1}$ are initially set at 0.2 per-unit (pu) and 0.6 pu and those of $PCC_{2}$ are set at 0.5 pu and 0.7 pu, respectively. The reference signals of DGU 1, i.e. $V_{1,d~ref}$ and $V_{1,q~ref}$, are affected by two step changes: the \emph{d} component of the load voltage steps up to 0.3 pu at $t=0.5$ s and the \emph{q} component steps down to 0.5 pu at $t=1.5$ s. Figure \ref{fig:trackdqDGU1} shows the dynamic responses of the two DGUs to these changes. In particular, Figures \ref{fig:trackdqDGU1_DGU1} and \ref{fig:trackdqDGU1_DGU2} show good tracking performances with small interactions between the two DGUs. Figures \ref{fig:trackdqDGU1_DGU1_changes1} and \ref{fig:trackdqDGU1_DGU1_changes2} show the instantaneous voltage at $PCC_1$ in the abc frame, during the two step changes of the reference signals. Note that the proposed decentralized control strategy ensures an excellent tracking of the references in less than two cycles.   
                    \begin{figure}[!htb]
                      \centering
                      \begin{subfigure}[!htb]{0.48\textwidth}
                        \centering
                        \includegraphics[width=1\textwidth, height=120pt]{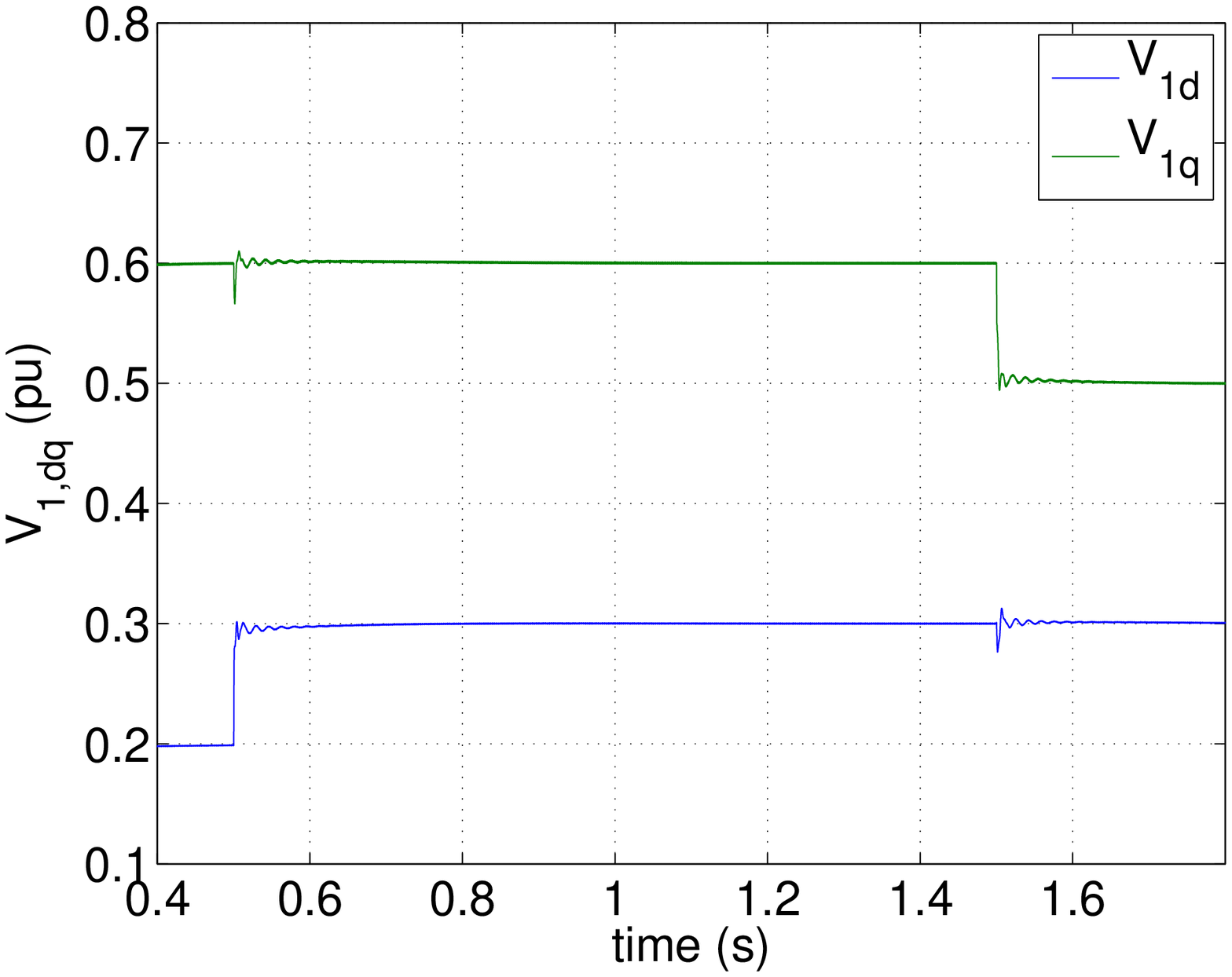}
                        \caption{\emph{d} and \emph{q} components of the load voltage at $PCC_1$.}
                        \label{fig:trackdqDGU1_DGU1}
                      \end{subfigure}
                      \begin{subfigure}[!htb]{0.48\textwidth}
                        \centering
                        \includegraphics[width=1\textwidth, height=120pt]{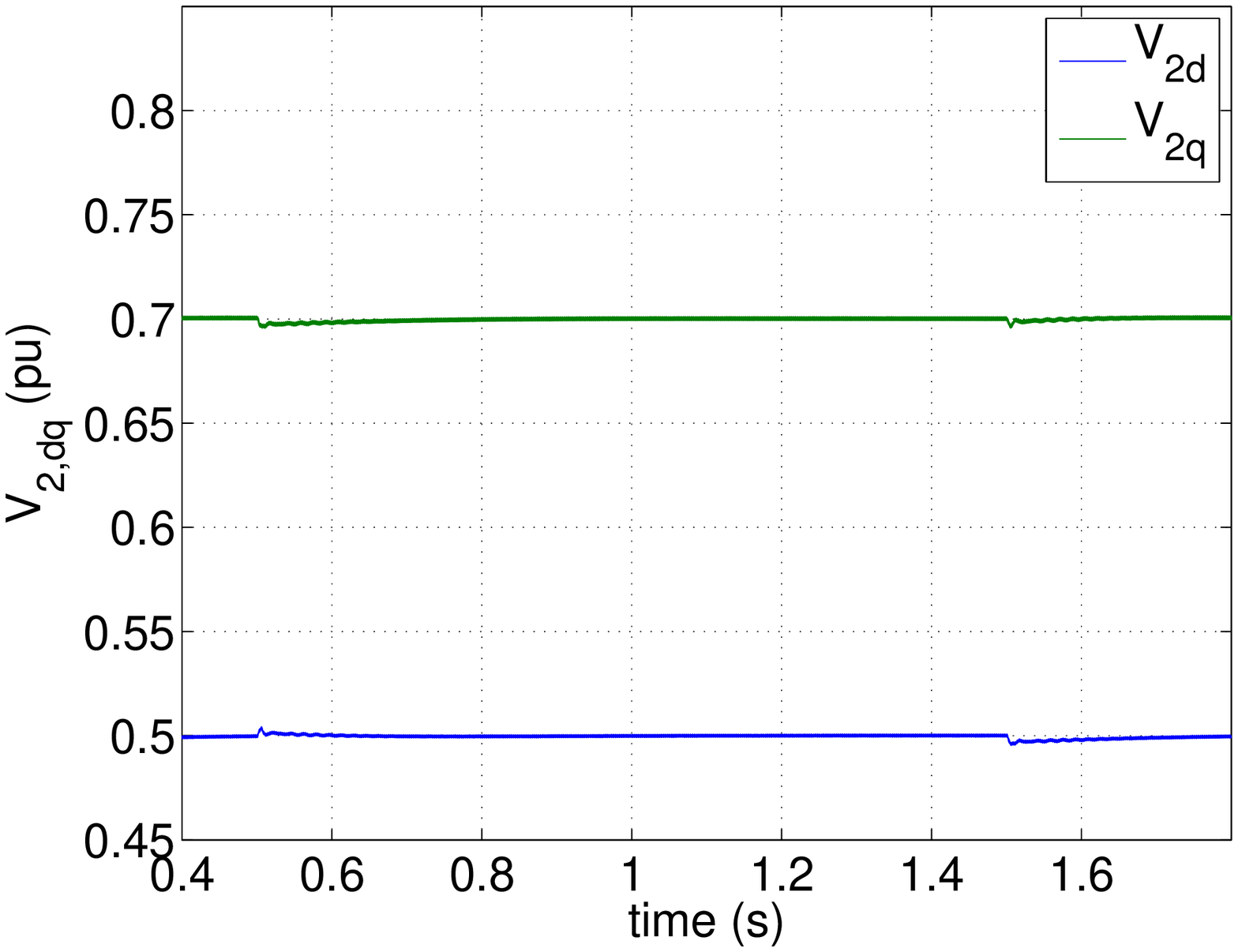}
                        \caption{\emph{d} and \emph{q} components of the load voltage at $PCC_2$.}
                        \label{fig:trackdqDGU1_DGU2}
                      \end{subfigure}
                      \begin{subfigure}[!htb]{0.48\textwidth}
                        \centering
                        \includegraphics[width=1\textwidth, height=120pt]{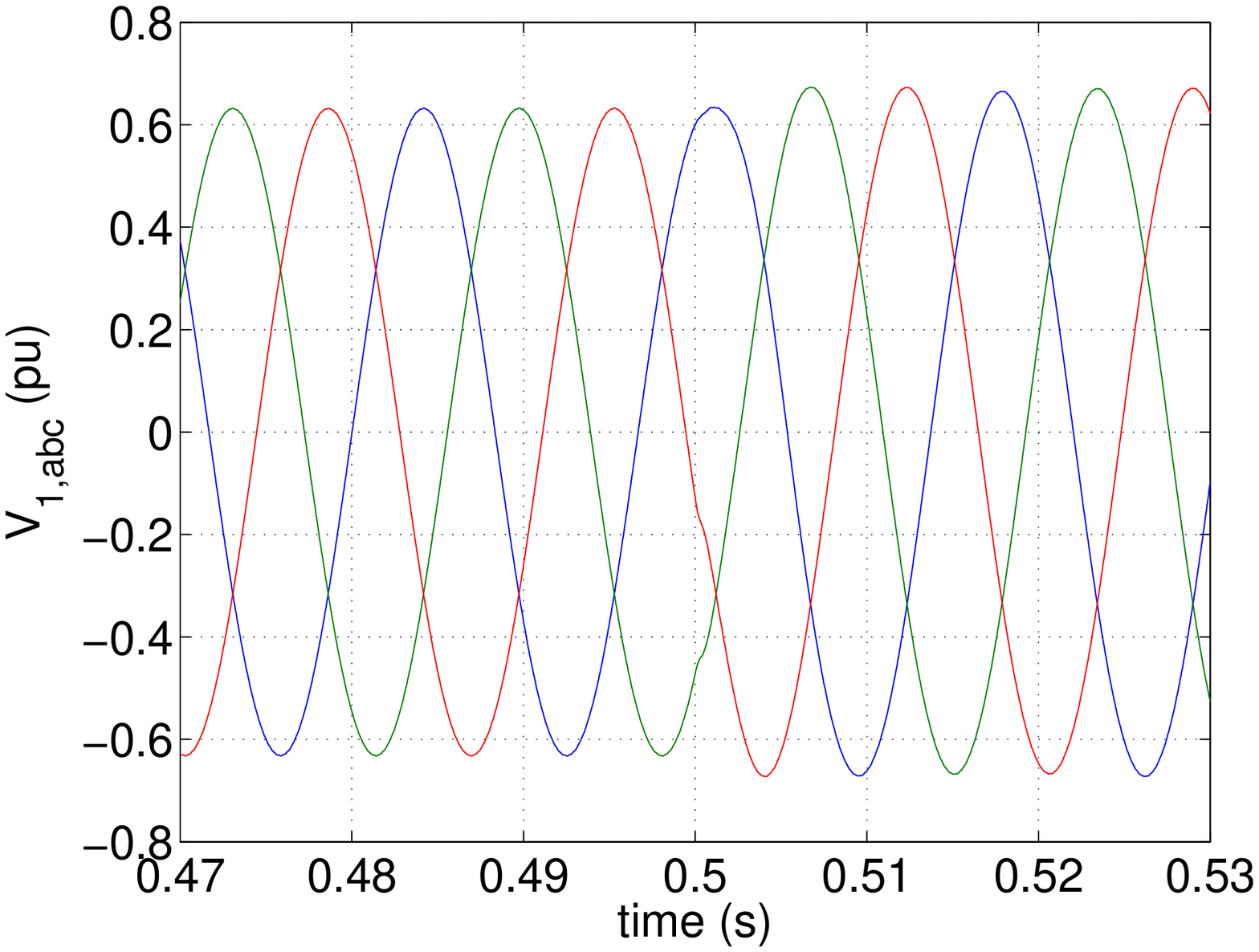}
                        \caption{Three-phase voltage at $PCC_1$ right before and after the first step.}
                        \label{fig:trackdqDGU1_DGU1_changes1}
                      \end{subfigure}
                      \begin{subfigure}[!htb]{0.48\textwidth}
                        \includegraphics[width=1\textwidth, height=120pt]{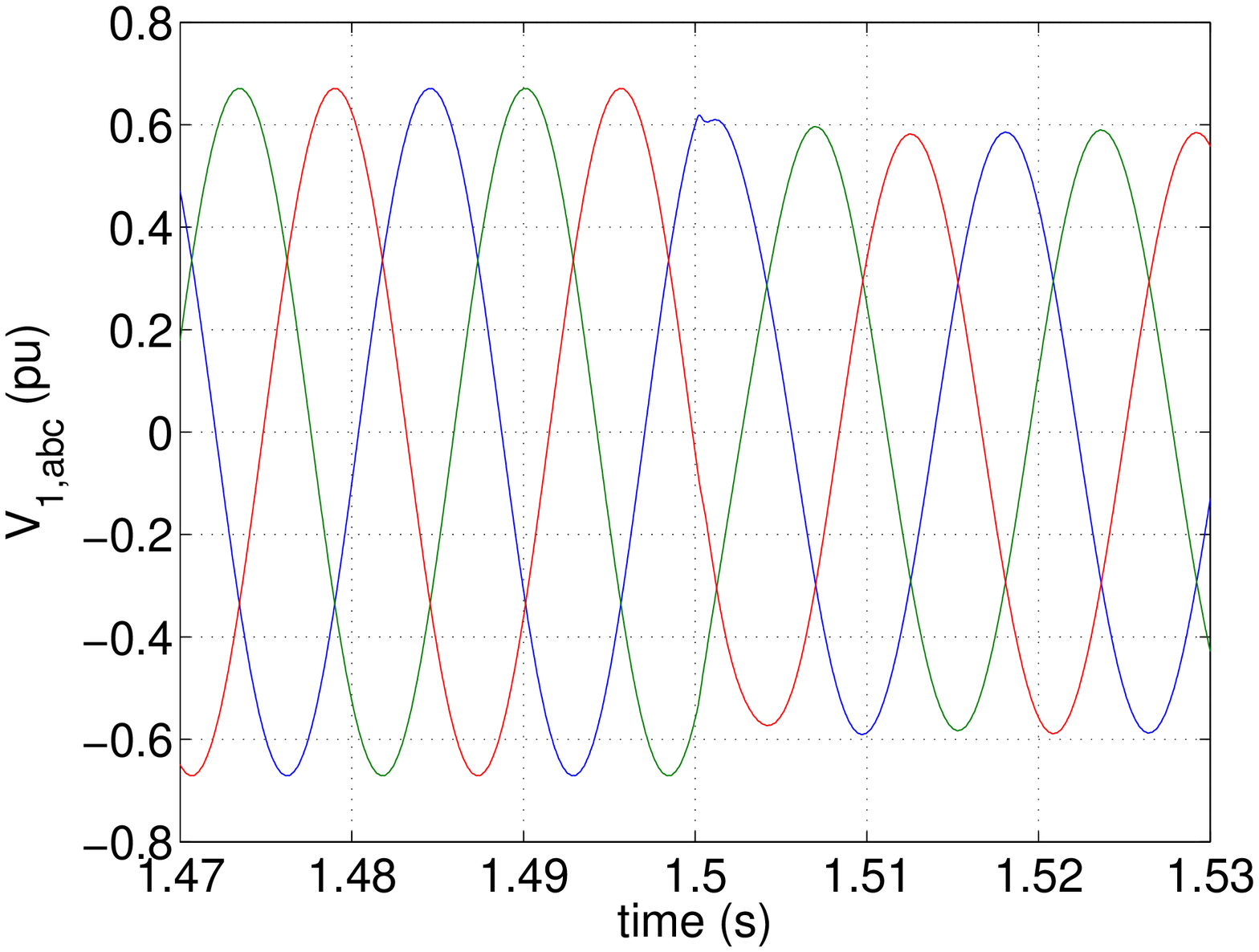}
                        \caption{Three-phase voltage at $PCC_1$ right before and after the second step.\\}
                        \label{fig:trackdqDGU1_DGU1_changes2}
                      \end{subfigure}
                      \caption{Performance of PnP decentralized voltage control in terms of set-point tracking for DGU 1.}
                      \label{fig:trackdqDGU1}
                    \end{figure}

               \subsubsection{Voltage tracking for DGU 2}
                    In this second test, voltage tracking properties at $PCC_2$ are evaluated. The loads used, in the two DGUs, are the same as in Section \ref{sec:scenario1voltagetrack1}. The \emph{d} and \emph{q} components of the voltage references at $PCC_1$ are kept constant at 0.6 pu and 0.8 pu, respectively. The \emph{d} component of the reference at $PCC_2$ steps down from 0.8 pu to 0.6 pu at $t=0.5$ s, and the \emph{q} component steps up from 0.2 pu to 0.4 pu at $t=1.5$ s. The dynamic responses of the overall microgrid to these changes in the reference signals are shown in Figure \ref{fig:trackdqDGU2}. Also in this case, after a short transient, the controller successfully maintains load voltages at prescribed levels.
                    \begin{figure}[!htb]
                      \centering
                      \begin{subfigure}[!htb]{0.48\textwidth}
                        \centering
                        \includegraphics[width=1\textwidth, height=120pt]{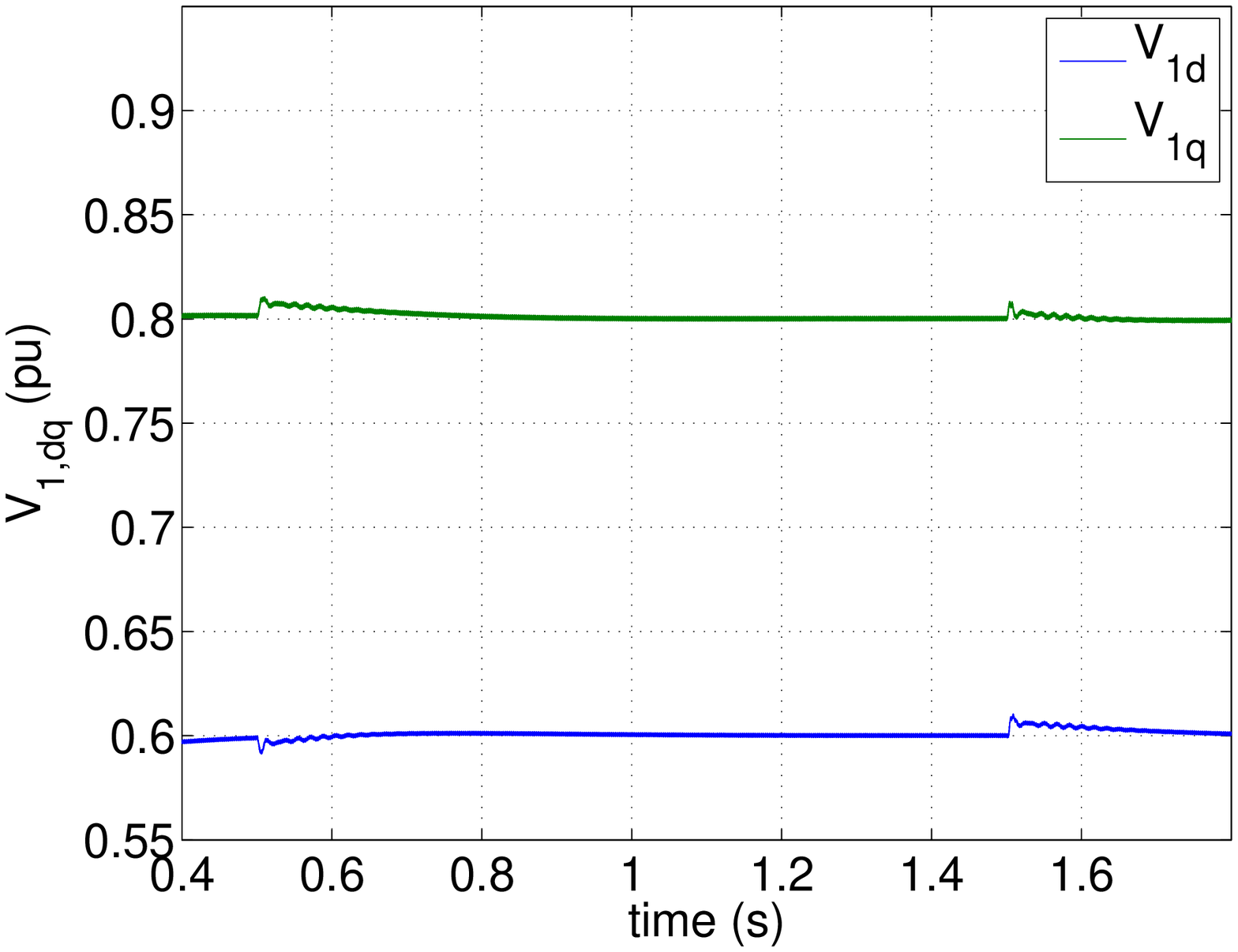}
                        \caption{\emph{d} and \emph{q} components of the load voltage at $PCC_1$.}
                        \label{fig:trackdqDGU2_DGU1}
                      \end{subfigure}
                      \begin{subfigure}[!htb]{0.48\textwidth}
                        \centering
                        \includegraphics[width=1\textwidth, height=120pt]{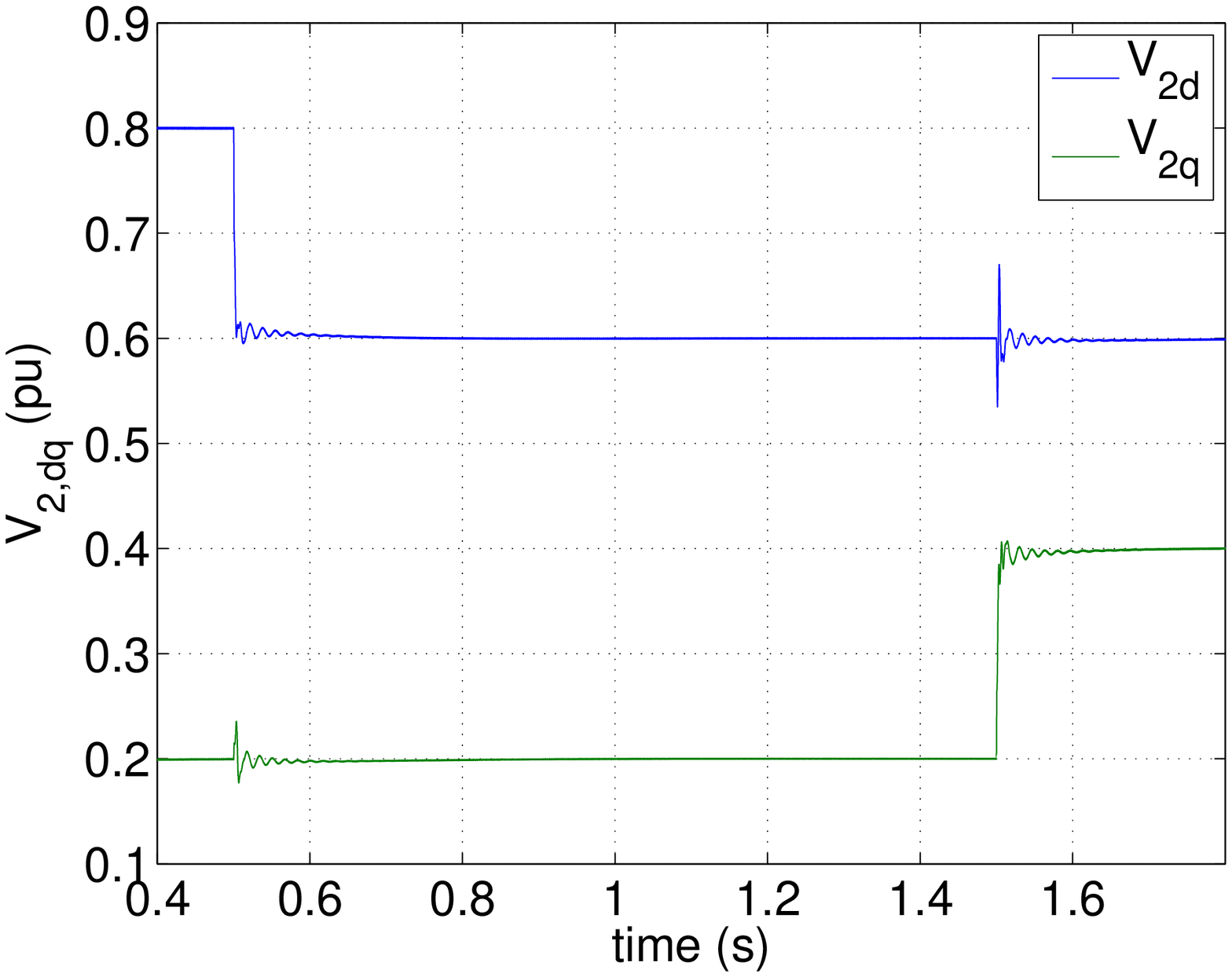}
                        \caption{\emph{d} and \emph{q} components of the load voltage at $PCC_2$.}
                        \label{fig:trackdqDGU2_DGU2}
                      \end{subfigure}
                      \begin{subfigure}[!htb]{0.48\textwidth}
                        \centering
                        \includegraphics[width=1\textwidth, height=120pt]{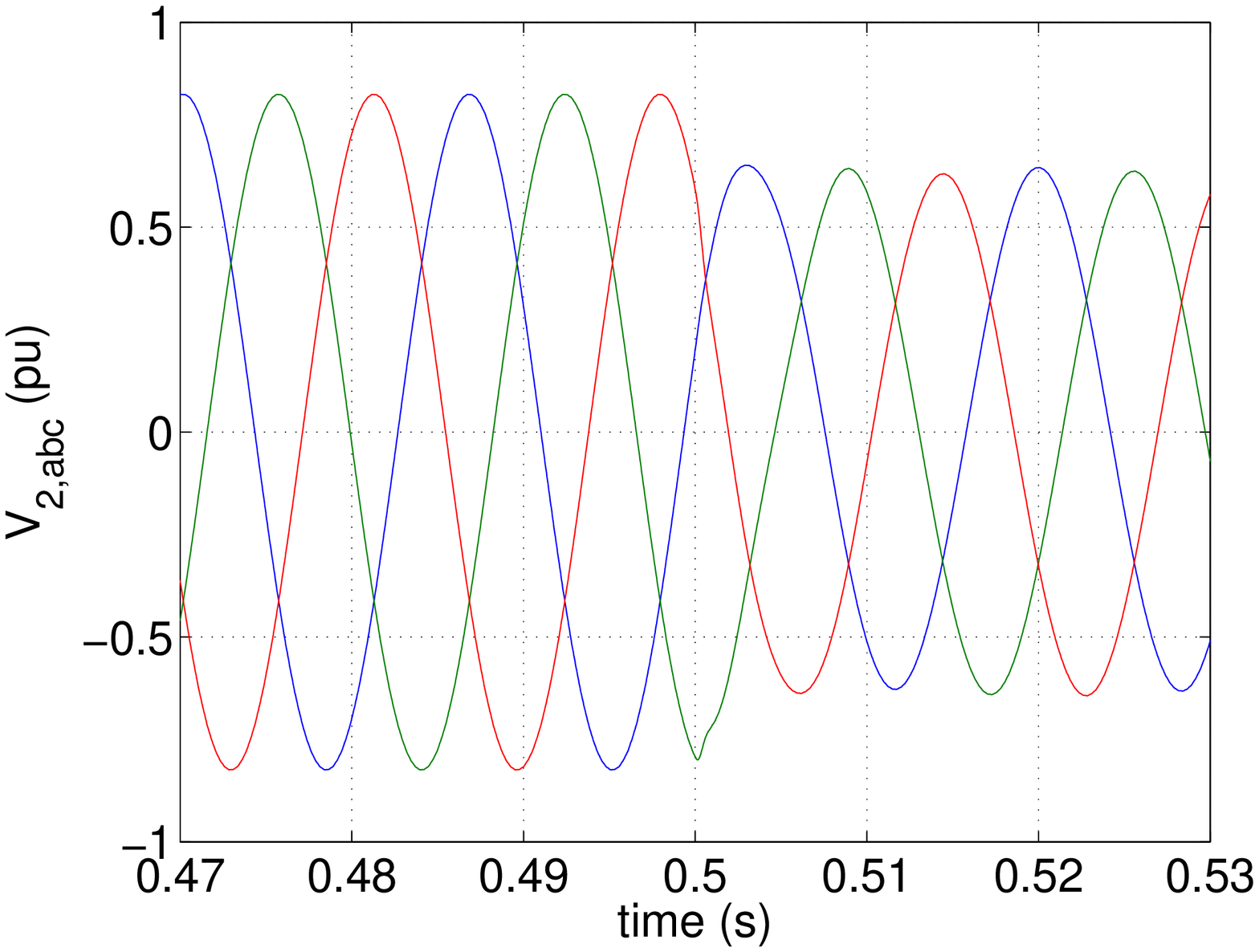}
                        \caption{Three-phase voltage at $PCC_2$ right before and after the first step.}
                        \label{fig:trackdqDGU2_DGU2_changes1}
                      \end{subfigure}
                      \begin{subfigure}[!htb]{0.48\textwidth}
                        \includegraphics[width=1\textwidth, height=120pt]{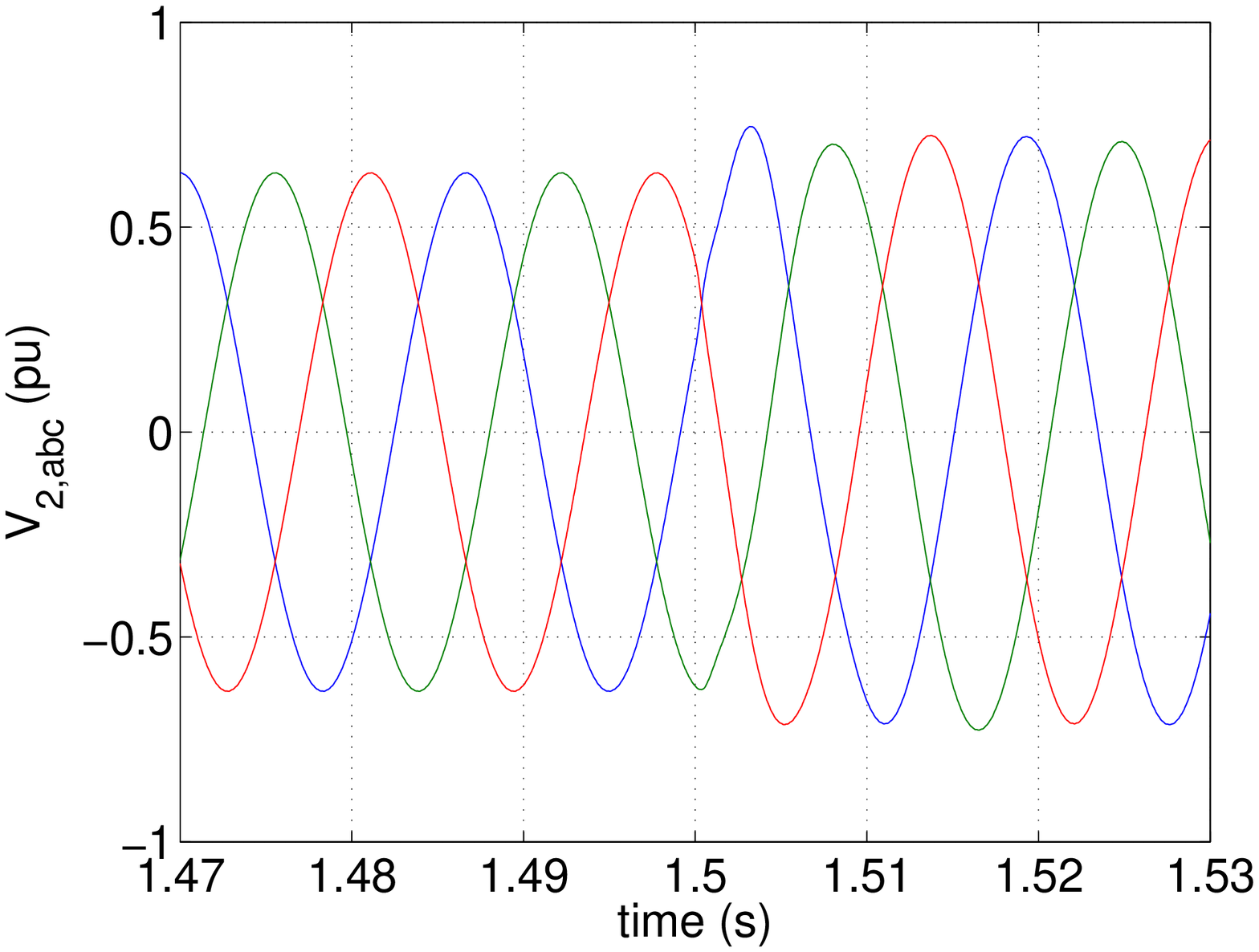}
                        \caption{Three-phase voltage at $PCC_2$ right before and after the second step.}
                        \label{fig:trackdqDGU2_DGU2_changes2}
                      \end{subfigure}
                      \caption{Performance of PnP decentralized voltage control in terms of set-point tracking for DGU 2.}
                      \label{fig:trackdqDGU2}
                    \end{figure}
                    
                    In the two tests presented above, the amplitude of the reference steps are the same used in simulations presented in \cite{Moradi2010}. Comparing our results with those in \cite{Moradi2010}, our controllers achieve better performances. In particular, using the proposed PnP decentralized controllers, the variation of the reference signals for a subsystem has almost no effect on the behaviour of the other DGU. We highlight that although these tests are useful for assessing the performance of different controllers, they are not realistic for an ImG. Indeed, the amplitude of step changes in the references is usually much smaller. In the following, we describe other case studies that are more meaningful to evaluate the electrical behaviour and the performance of the microgrid controlled by PnP regulators.

	       \subsubsection{Robustness to unknown load dynamics}
                    For assessing robustness to load dynamics, we consider RL parallel loads, initially set as described in Section \ref{sec:scenario1voltagetrack1} and apply sudden changes in the two load resistances. During the simulation, the \emph{d} and \emph{q} components of the load voltages at $PCC_1$ and $PCC_2$ are regulated at 0.8 pu and 0.3 pu, and at 0.5 pu and 0.9 pu, respectively. The load resistance at $PCC_1$ is increased from the starting value (76 $\Omega$) to double of it (152 $\Omega$) at $t=0.5$ s. Similarly, the resistance at $PCC_2$ is reduced from 76 $\Omega$ to 65 $\Omega$ at $t=1.5$ s. Figure \ref{fig:Rvar} shown the response of the ImG. From Figures \ref{fig:Rvardqpcc1} and \ref{fig:Rvardqpcc2}, it is apparent that, except for a very short transient, the \emph{d} and \emph{q} components of the load voltages remain unaffected by the load. Figures \ref{fig:Rvarabcpcc1} and \ref{fig:Rvarabcpcc2} show the load voltages, in abc frame, at $PCC_1$ and $PCC_2$, respectively. They reveal that the step changes in the loads are absorbed within a cycle. We recall that load currents are treated as measurable disturbances in our model. Varying the load resistance, is equivalent to changing disturbance values, hence the power required by loads. Furthermore, our control architecture is robust to the disturbance dynamics also due to the presence of compensators $\subss{N}{i}$, $i=\{1,2\}$. In fact, as shown in Figures \ref{fig:Rvarcompensator1} and \ref{fig:Rvarcompensator2}, the inputs $\subss{\tilde{u}}{1}(t)$ and $\subss{\tilde{u}}{2}(t)$, computed by the local compensators, react to the changes in load resistances, in order to zero the effect of these disturbances on voltages at PCCs. \\
                    A very similar case study (also in terms of parameter values) is presented also in \cite{Moradi2010}. A comparison of the results reveals that, in our case, transients due to resistances variations are shorter thanks to the input provided by the disturbance compensators. 
                    \begin{figure}[!htb]
                      \centering
                      \begin{subfigure}[!htb]{0.48\textwidth}
                        \centering
                        \includegraphics[width=1\textwidth, height=120pt]{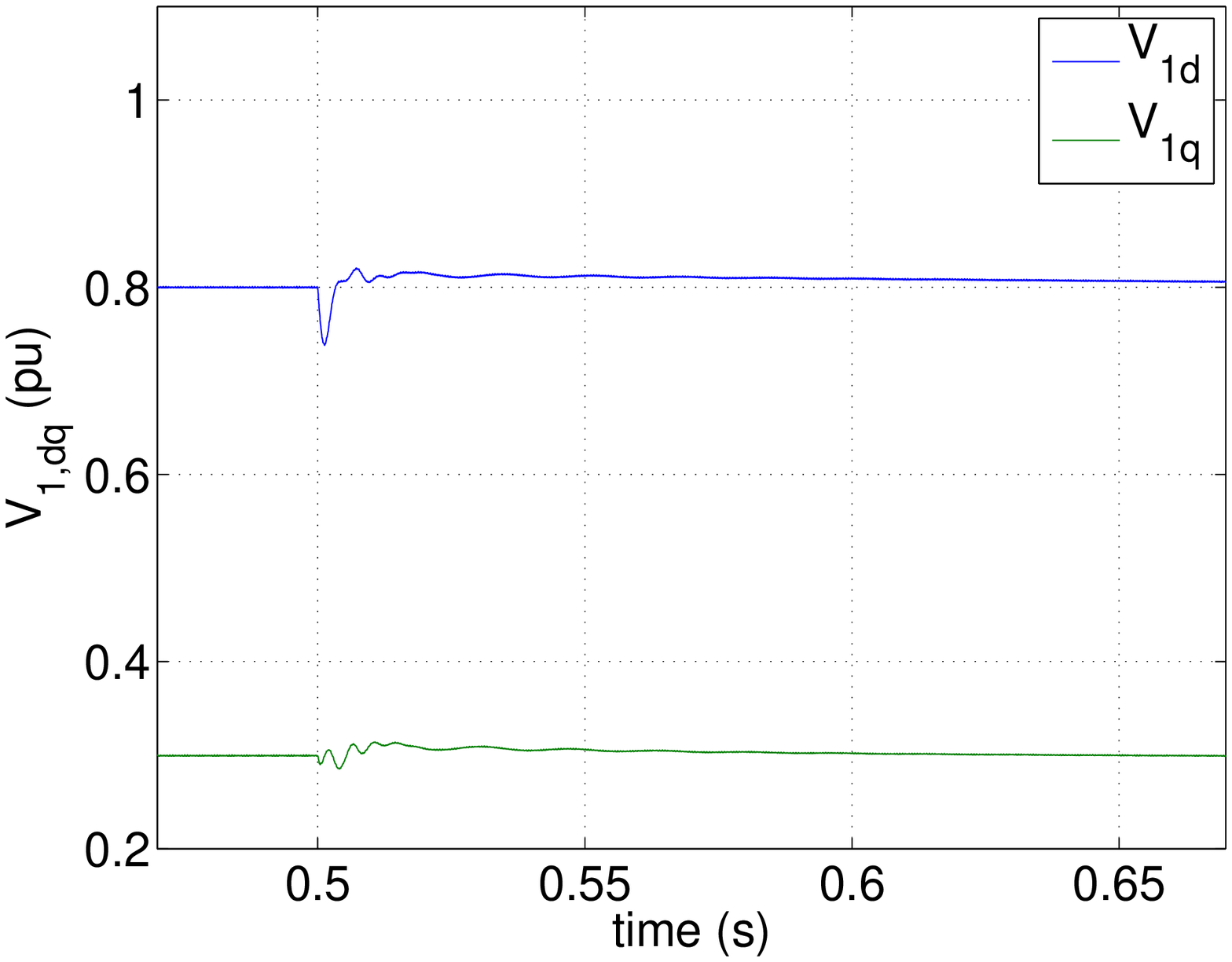}
                        \caption{\emph{d} and \emph{q} components of the load voltage at $PCC_1$.}
                        \label{fig:Rvardqpcc1}
                      \end{subfigure}
                      \begin{subfigure}[!htb]{0.48\textwidth}
                        \centering
                        \includegraphics[width=1\textwidth, height=120pt]{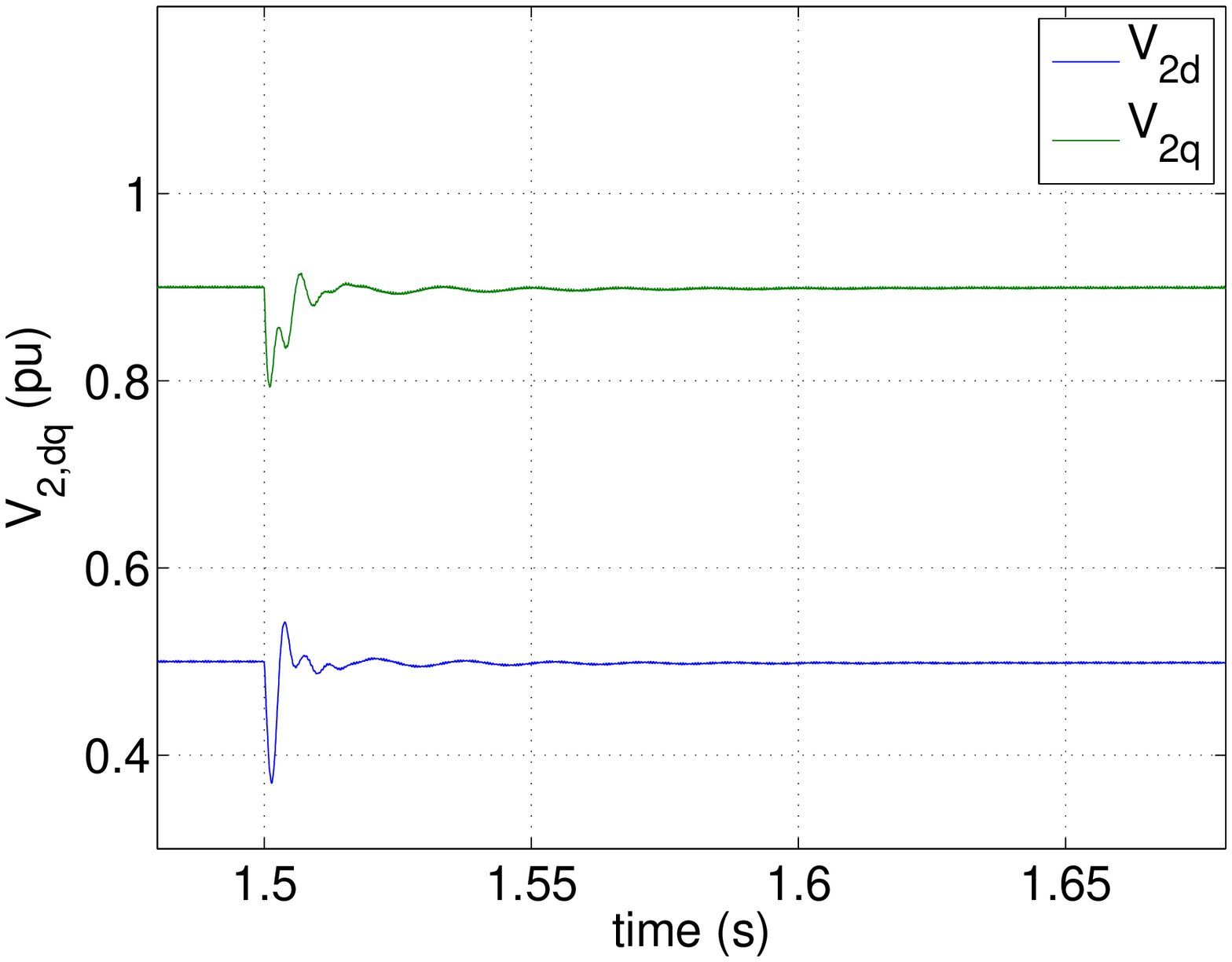}
                        \caption{\emph{d} and \emph{q} components of the load voltage at $PCC_2$.}
                        \label{fig:Rvardqpcc2}
                      \end{subfigure}
                      \begin{subfigure}[!htb]{0.48\textwidth}
                        \centering
                        \includegraphics[width=1\textwidth, height=120pt]{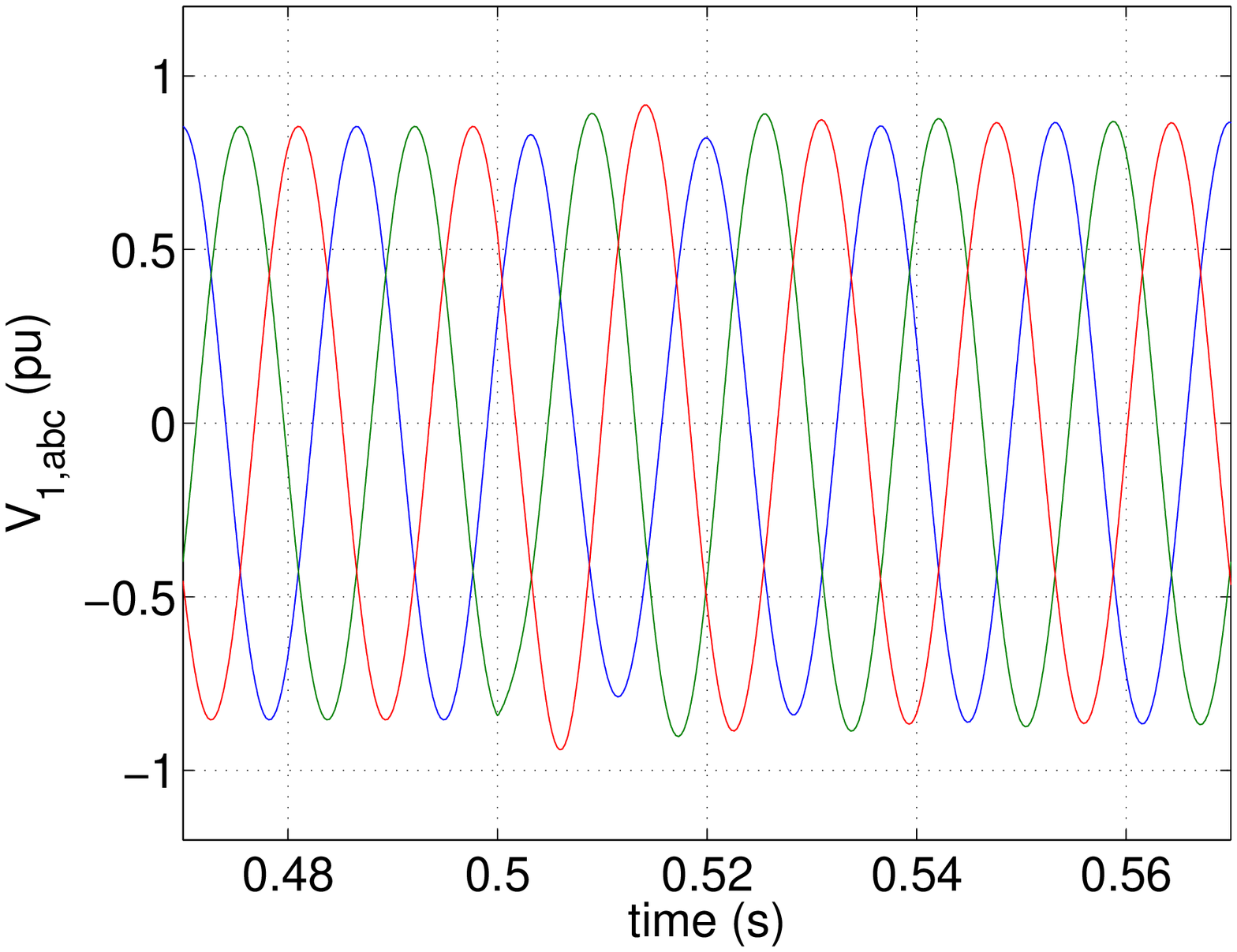}
                        \caption{Three-phase voltage at $PCC_1$.}
                        \label{fig:Rvarabcpcc1}
                      \end{subfigure}
                      \begin{subfigure}[!htb]{0.48\textwidth}
                        \centering
                        \includegraphics[width=1\textwidth, height=120pt]{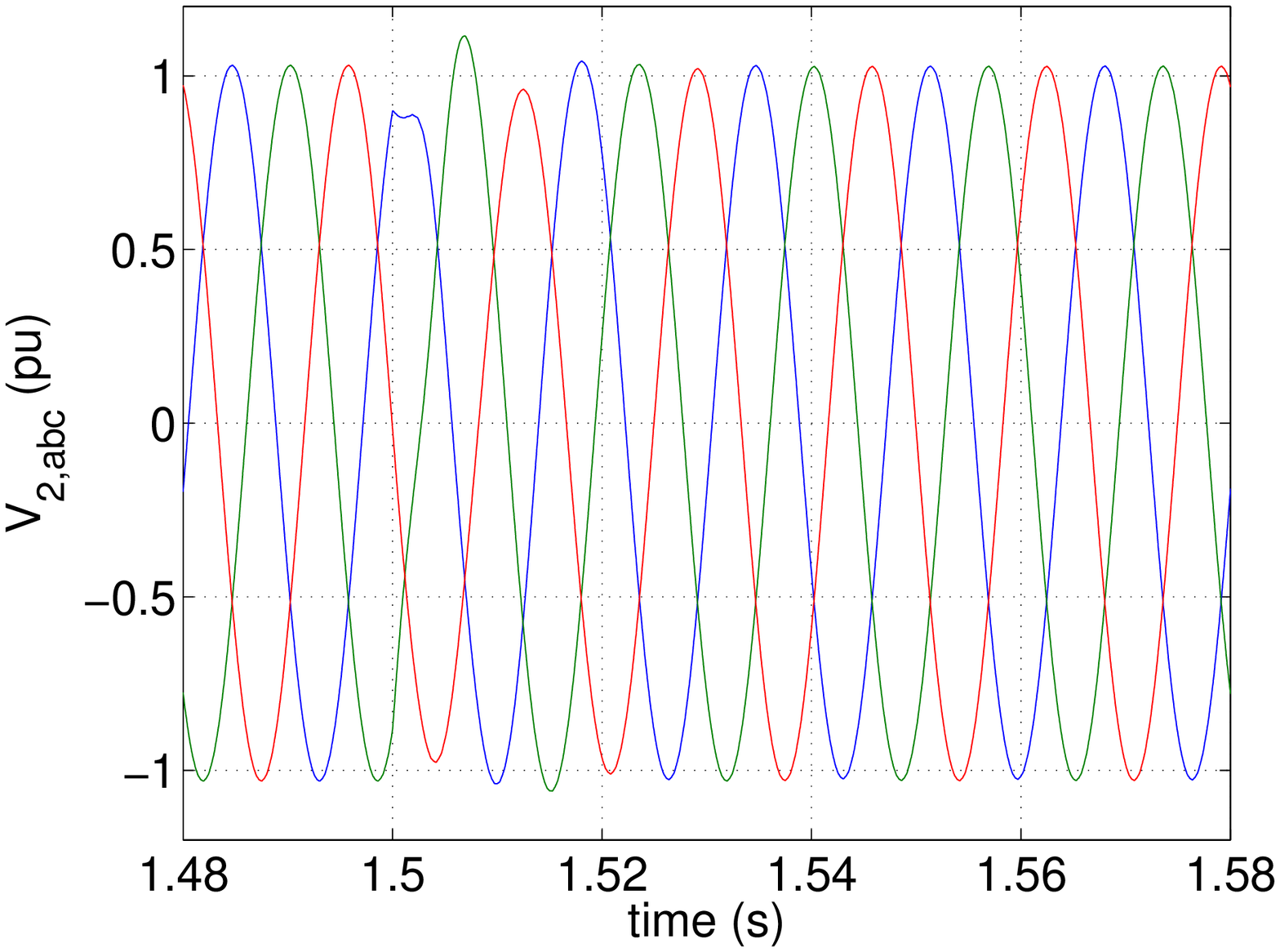}
                        \caption{Three-phase voltage at $PCC_2$.}
                        \label{fig:Rvarabcpcc2}
                      \end{subfigure}
                      \begin{subfigure}[!htb]{0.48\textwidth}
                        \centering
                        \includegraphics[width=1\textwidth, height=120pt]{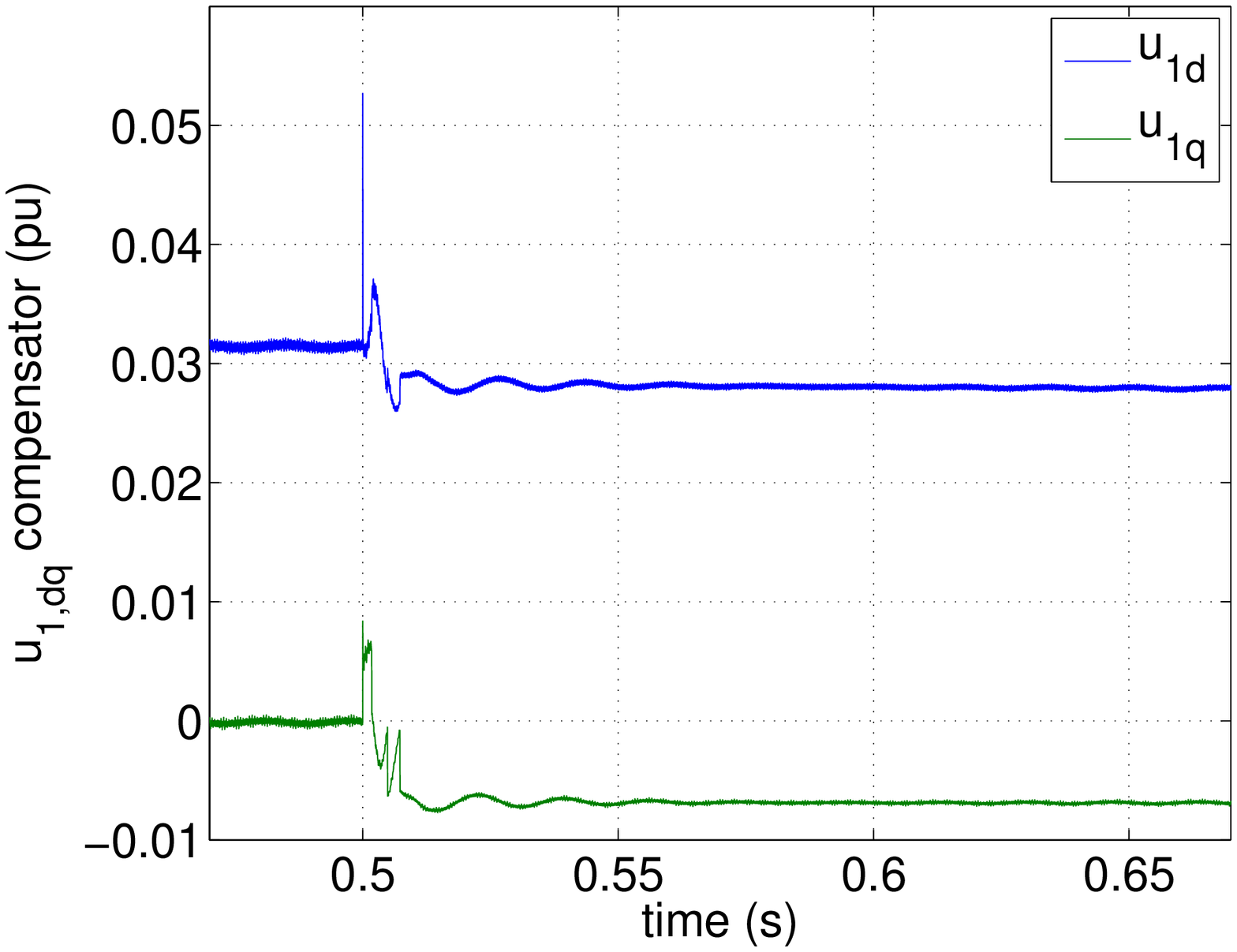}
                        \caption{Input signal $\subss{\tilde{u}}{1}(t)$ computed by compensator $\subss{N}{1}$.}
                        \label{fig:Rvarcompensator1}
                      \end{subfigure}
                      \begin{subfigure}[!htb]{0.48\textwidth}
                        \centering
                        \includegraphics[width=1\textwidth, height=120pt]{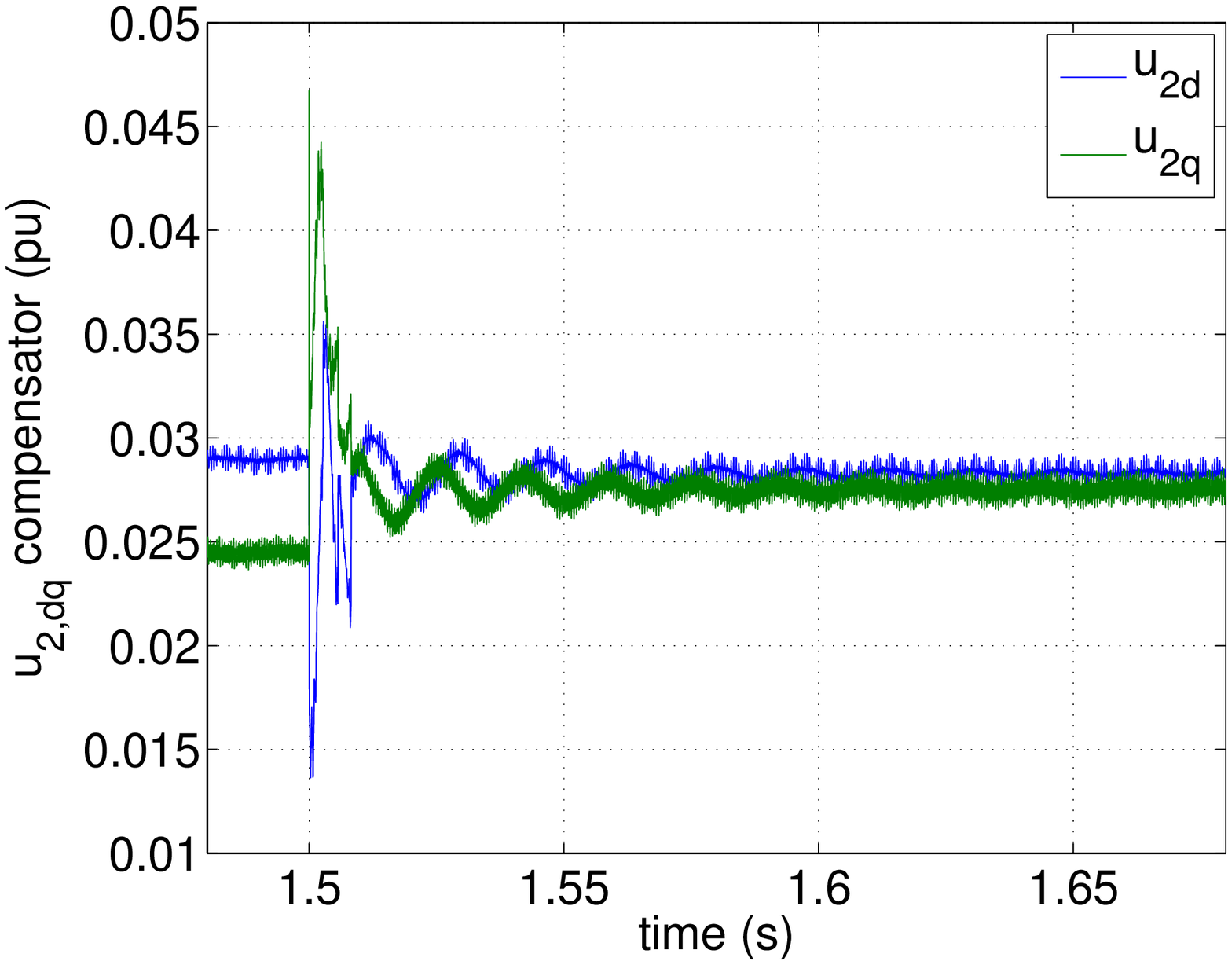}
                        \caption{Input signal $\subss{\tilde{u}}{2}(t)$ computed by compensator $\subss{N}{2}$.}
                        \label{fig:Rvarcompensator2}
                      \end{subfigure}
                      \caption{Robust performance evaluation of the microgrid against unknown load dynamics.}
                      \label{fig:Rvar}
                    \end{figure}

               \subsubsection{Impact of a nonlinear load}
                    In this test, we study the performance of our controllers in presence of a highly nonlinear load. Voltage references are initially set to $V_{d,ref}=0.8$ pu and $V_{q,ref}=0.6$ pu for both DGUs. At the beginning of the simulation, we connect at $PCC_1$ and $PCC_2$ the RL load described in Section \ref{sec:scenario1voltagetrack1}. At $t=0.5$ s, the load connected at $PCC_2$ is suddenly replaced by a three-phase six-pulse diode rectifier. The rectifier produces a DC output voltage that feeds a purely resistive load with $R=120\,\Omega$. We highlight that this is a standard test for assessing robustness of microgrid operations to nonlinearities (see, for example, Section VI.C in \cite{Babazadeh2013} and \cite{IEEE2009}).\\
                    Simulations are shown in Figure \ref{fig:NL}. In particular, Figure \ref{fig:NLdq} shows the \emph{dq} components of the load voltage at $PCC_2$ which confirm the good tracking performance in spite of the inclusion of the rectifier. From Figure \ref{fig:NLabc}, one can notice that, except for short transients, local controllers successfully stabilize the voltage. \\
                    Figure \ref{fig:NLTHD} provides a plot of the Total Harmonic Distortion (THD) (expressed in $\%$) of load phase-to-phase voltage at $PCC_2$. We note that, after the connection of the rectifier, the THD value grows. However, the average value of THD after $t=0.5$s is equal to $4\%$ which is below the maximum limit ($5\%$) recommended by IEEE standards in \cite{IEEE2009}. Considering that the rectifier input currents are highly distorted, as shown in Figure \ref{fig:NLabcI}, the control architecture ensures that the load is fed with high-quality voltages. 
                    \begin{figure}[!htb]
                      \centering
                      \begin{subfigure}[!htb]{0.45\textwidth}
                        \centering
                        \includegraphics[width=1\textwidth, height=120pt]{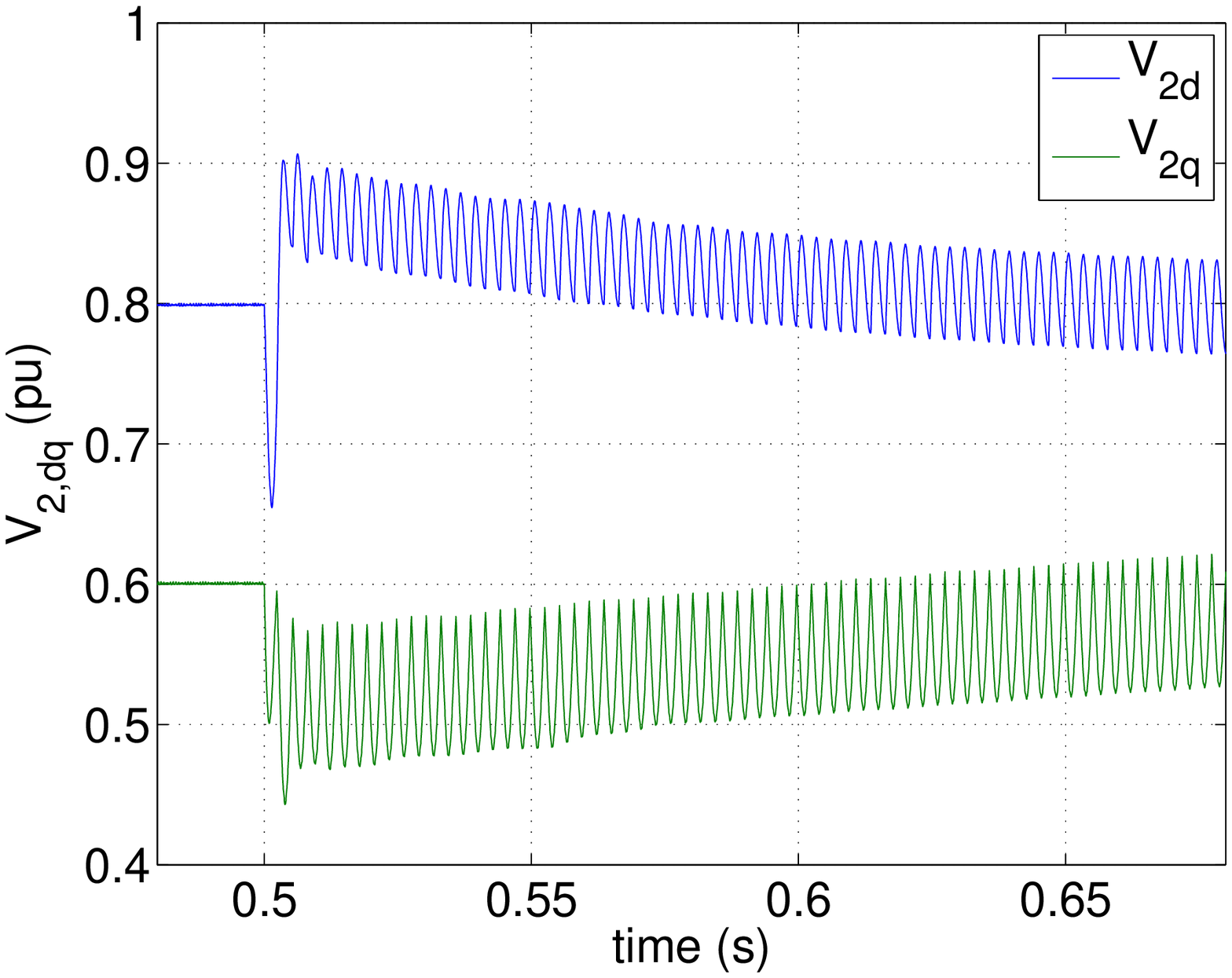}
                        \caption{\emph{d} and \emph{q} components of the voltage at $PCC_2$.}
                        \label{fig:NLdq}
                      \end{subfigure}
                      \begin{subfigure}[!htb]{0.45\textwidth}
                        \centering
                        \includegraphics[width=1\textwidth, height=120pt]{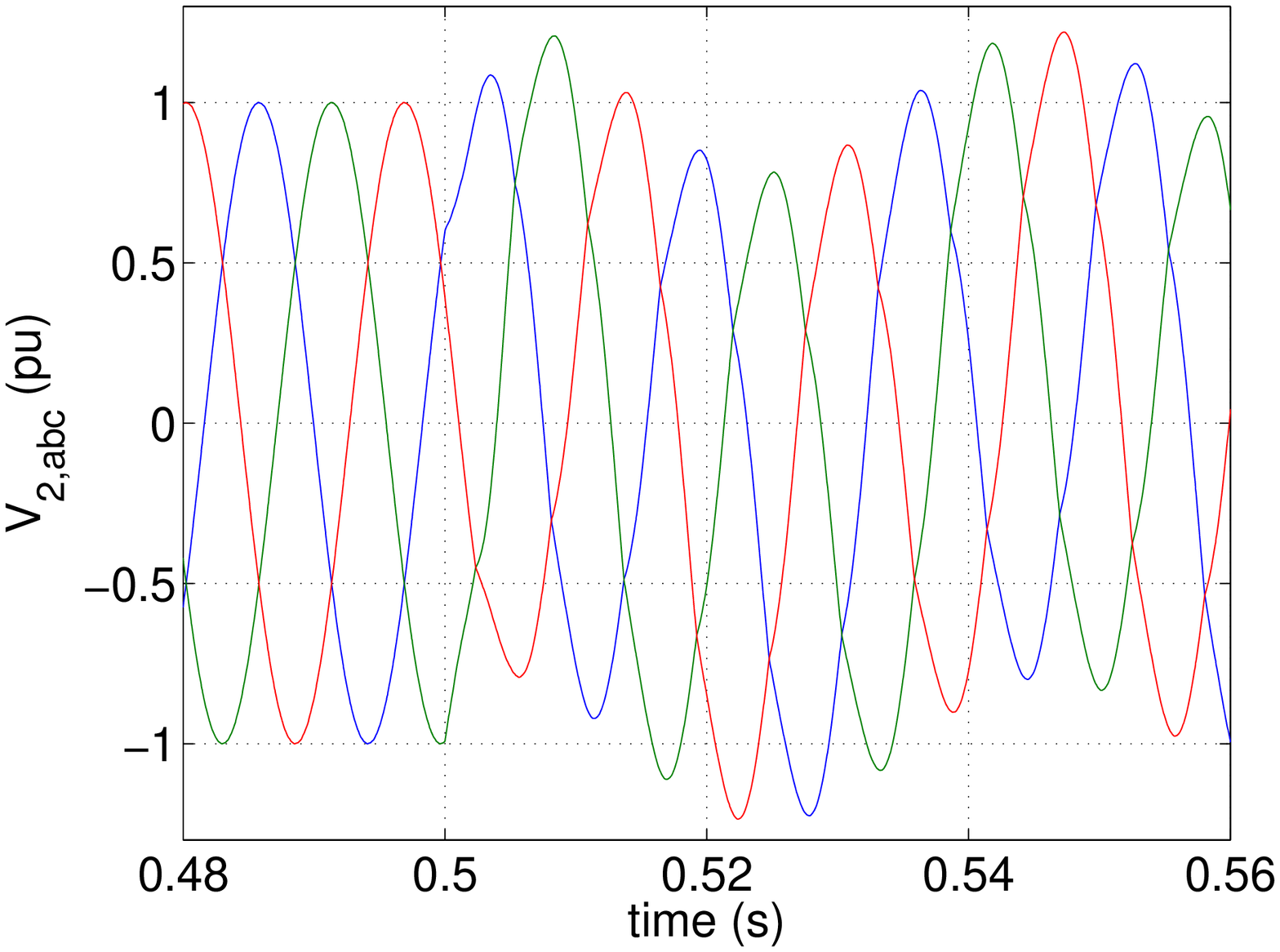}
                        \caption{Three-phase voltage at $PCC_2$.}
                        \label{fig:NLabc}
                      \end{subfigure}
                      \begin{subfigure}[!htb]{0.45\textwidth}
                        \centering
                        \includegraphics[width=1\textwidth, height=120pt]{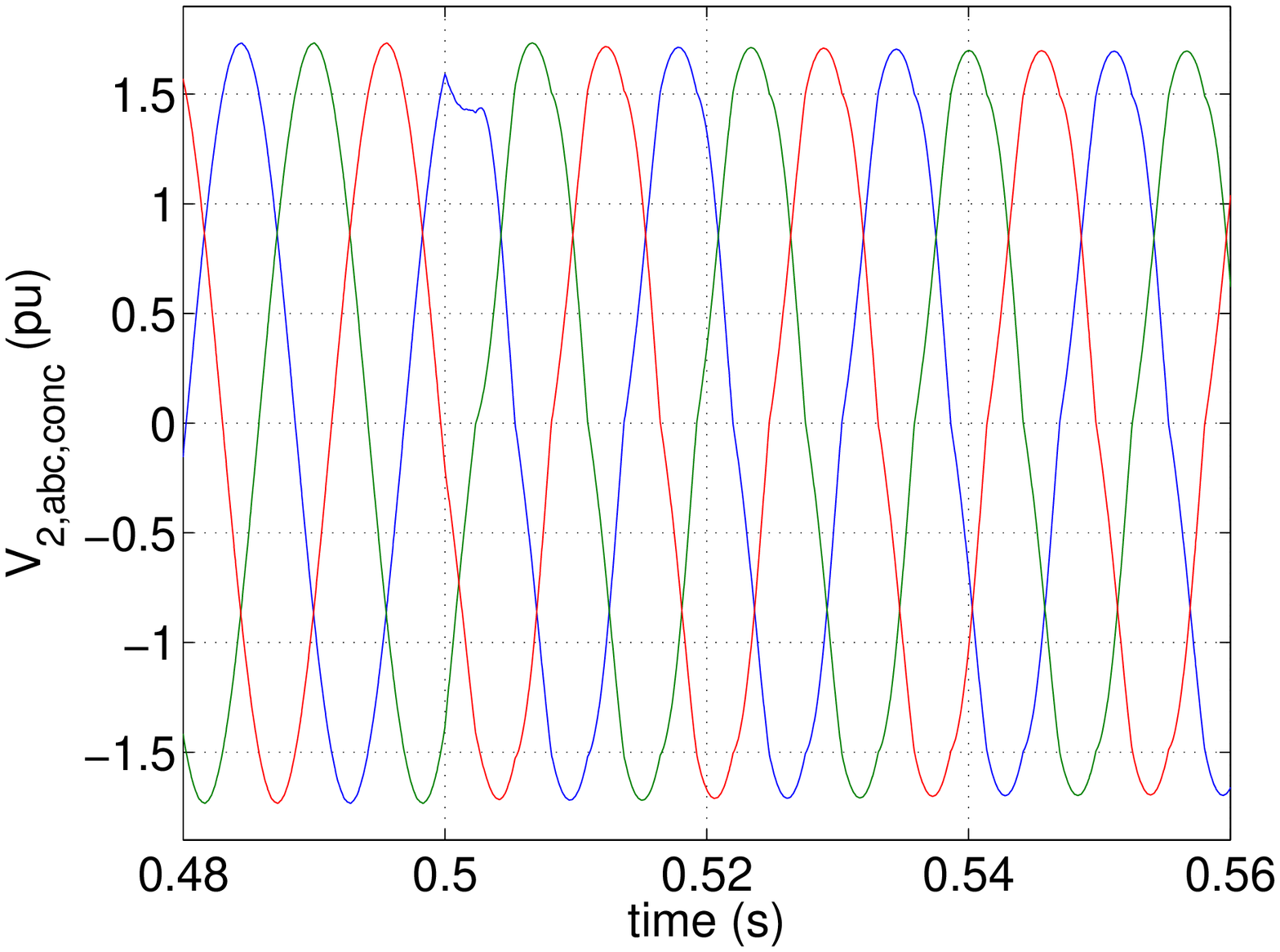}
                        \caption{Three-phase phase-to-phase voltage at $PCC_2$.}
                        \label{fig:NLabcconc}
                      \end{subfigure}
                      \begin{subfigure}[!htb]{0.45\textwidth}
                        \centering
                        \includegraphics[width=1\textwidth, height=120pt]{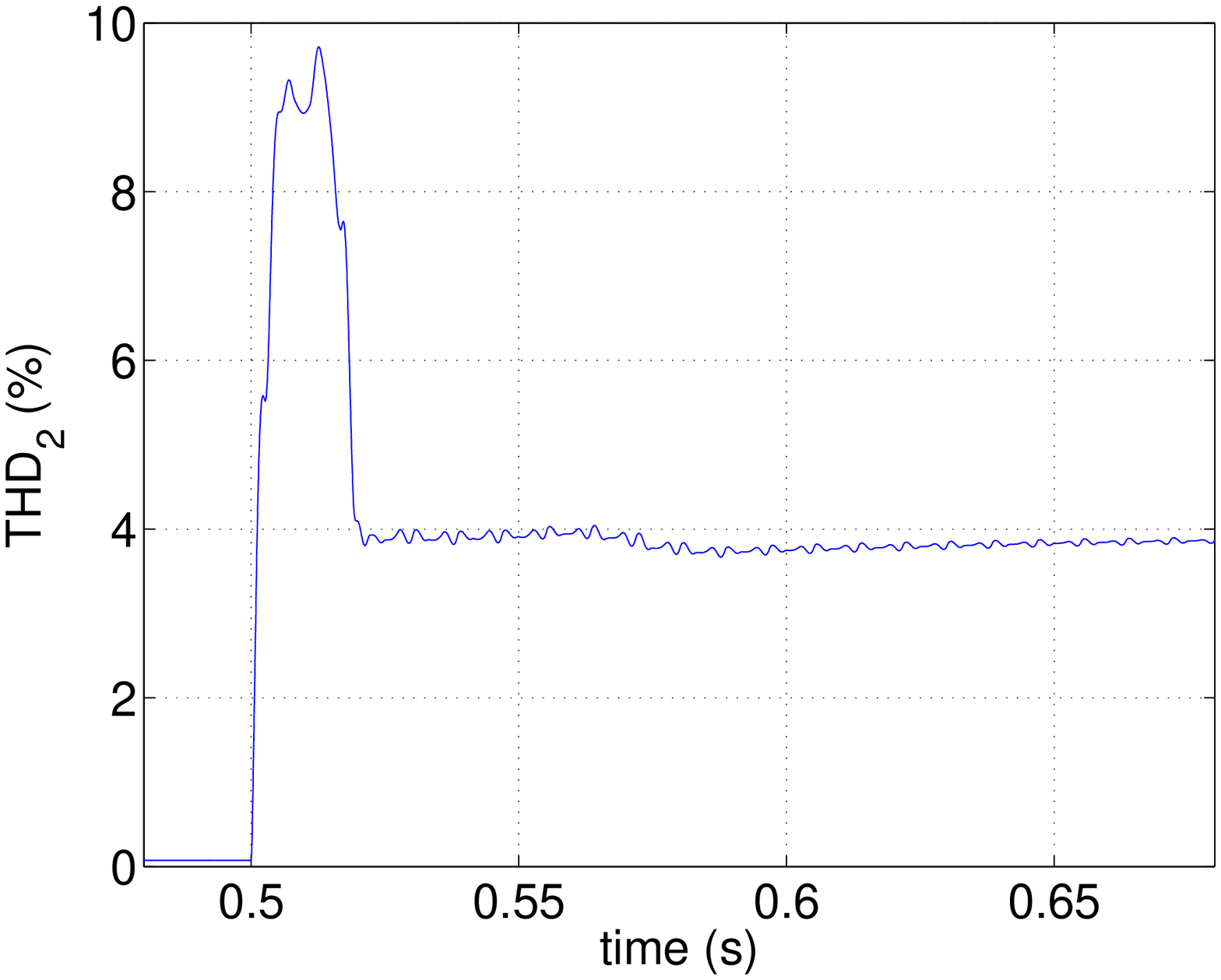}
                        \caption{THD of the voltage at $PCC_2$ in percentage.}
                        \label{fig:NLTHD}
                      \end{subfigure}
                      \begin{subfigure}[!htb]{0.45\textwidth}
                        \centering
                        \includegraphics[width=1\textwidth, height=120pt]{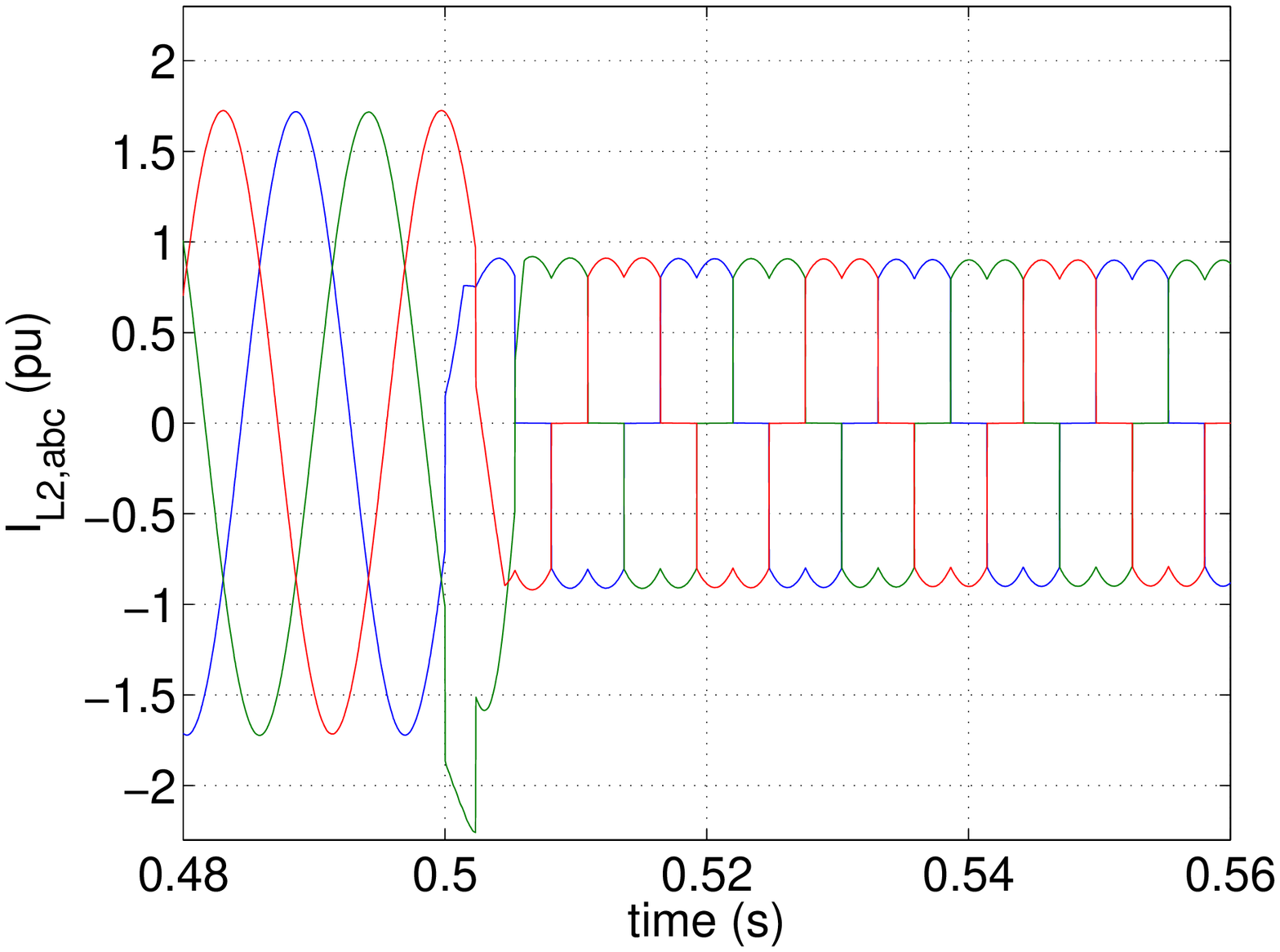}
                        \caption{Instantaneous load current $I_{L2,abc}$.}
                        \label{fig:NLabcI}
                      \end{subfigure}
                      \caption{Performance of PnP decentralized voltage control in presence of a highly nonlinear load.}
                      \label{fig:NL}
                    \end{figure}

	       \subsubsection{Performance under unbalanced load conditions}
                    In this test, we investigate the performance of PnP controller in presence of unbalanced loads. Voltage references are initially set to $V_{d,ref}=0.6$ pu and $V_{q,ref}=0.8$ pu for both DGUs. Moreover, the nominal RL load described in Section \ref{sec:scenario1voltagetrack1} is connected at $PCC_1$ and $PCC_2$. At $t=0.5$ s the RL load parameters at $PCC_1$ are changed to the values given in Table \ref{tbl:parunbload}, so that the load of DGU 1 becomes highly unbalanced.
                    \begin{table}[!htb]                 
                      \centering
                      \begin{tabular}{*{4}{c}}
                        \toprule
                        & Phase a & Phase b & Phase c \\
                        \midrule
                        R ($\Omega$) & 76 & 228  & 456 \\
                        L (mH) & 111.9 & 123.9 & 111.9\\
                        \bottomrule
                      \end{tabular}
                      \caption{Unbalanced load parameters}	
                      \label{tbl:parunbload}
                    \end{table}
                    
                    Figure \ref{fig:UnbaldqV} shows the \emph{d} and \emph{q} components of the load voltage at $PCC_1$ before and after unbalancing. We note that tracking of the reference signals is still guaranteed in spite of load changes. Moreover, instantaneous load voltages, shown in Figure \ref{fig:UnbalabcV}, confirm successful regulation of the output waveforms. Figure \ref{fig:UnbalILabc} shows the load current $I_{L1}$ provided by DGU 1 in the abc frame. One can notice how the controller induces major changes in the VSC behaviour in order to avoid spoiling the balance of load voltage at $PCC_1$.\\
                    To evaluate the voltage imbalance at $PCC_1$, we calculate the ratio $V_{N}/V_{P}$ (expressed in $\%$), where $V_{N}$ and $V_{P}$ are the magnitudes of the negative- and positive-sequence components of the phase-to-phase voltage. The time evolution of this ratio is represented in Figure \ref{fig:UnbalVnVp}. We notice that it is always below 0.5$\%$ which is less than the maximum permissible value (3$\%$) defined by IEEE in \cite{IEEE2009}. We highlight that a similar experiment has been executed in \cite{Babazadeh2013}. 
                    \begin{figure}[!htb]
                      \centering
                      \begin{subfigure}[!htb]{0.45\textwidth}
                        \centering
                        \includegraphics[width=1\textwidth, height=115pt]{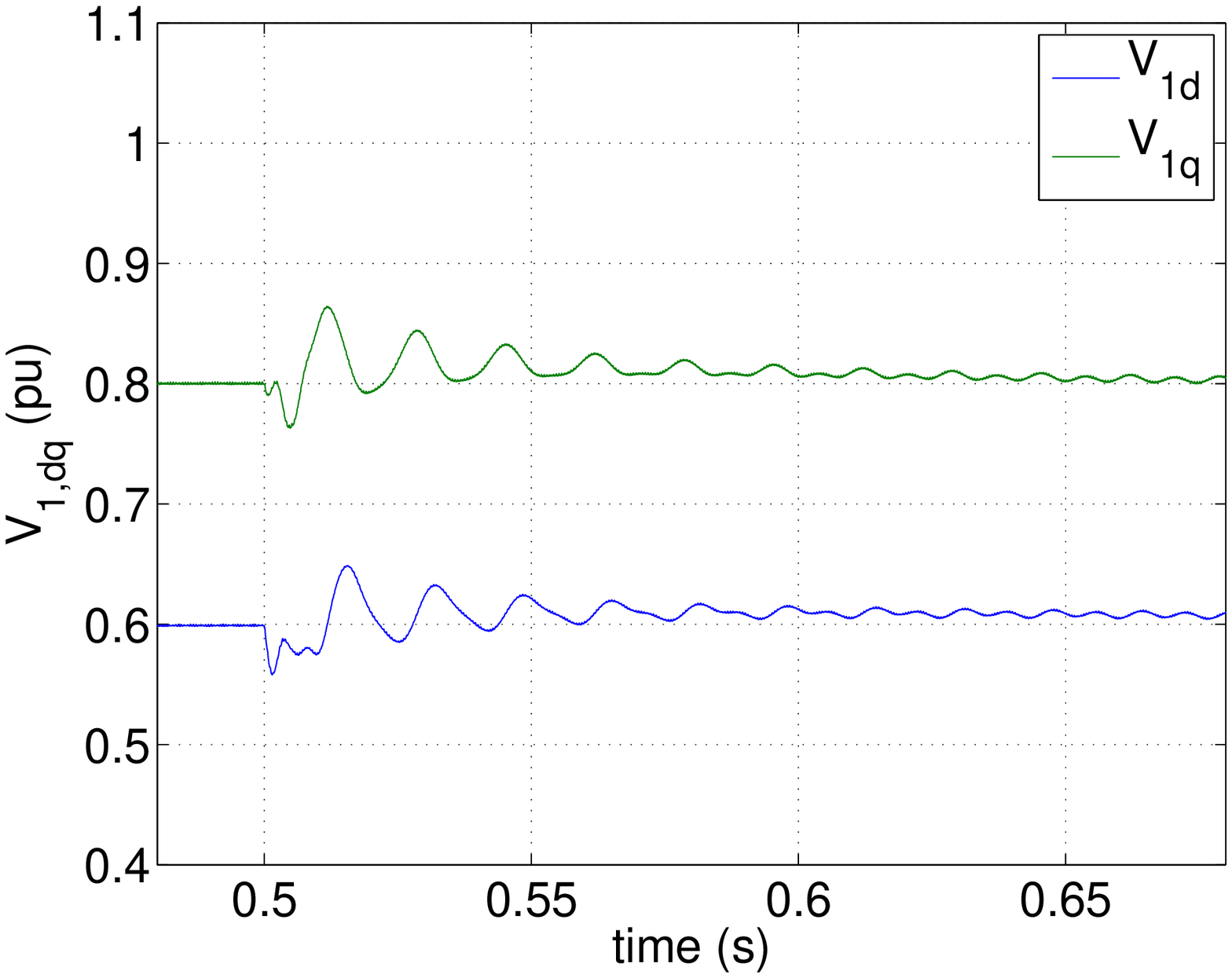}
                        \caption{\emph{d} and \emph{q} components of the voltage at $PCC_1$.}
                        \label{fig:UnbaldqV}
                      \end{subfigure}
                      \begin{subfigure}[!htb]{0.45\textwidth}
                        \centering
                        \includegraphics[width=1\textwidth, height=115pt]{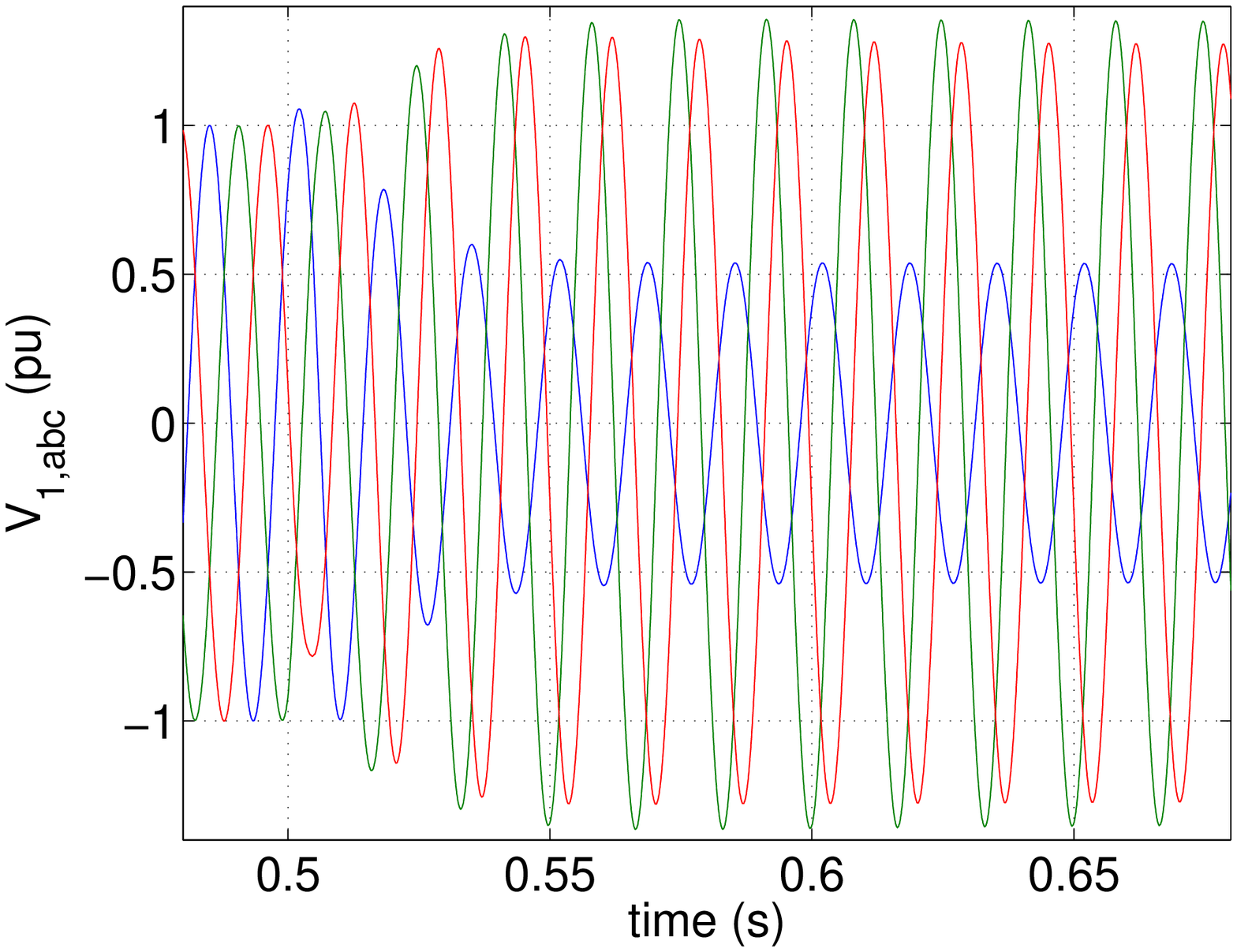}
                        \caption{Three-phase voltage at $PCC_1$.}
                        \label{fig:UnbalabcV}
                      \end{subfigure}
                      \begin{subfigure}[!htb]{0.45\textwidth}
                        \centering
                        \includegraphics[width=1\textwidth, height=115pt]{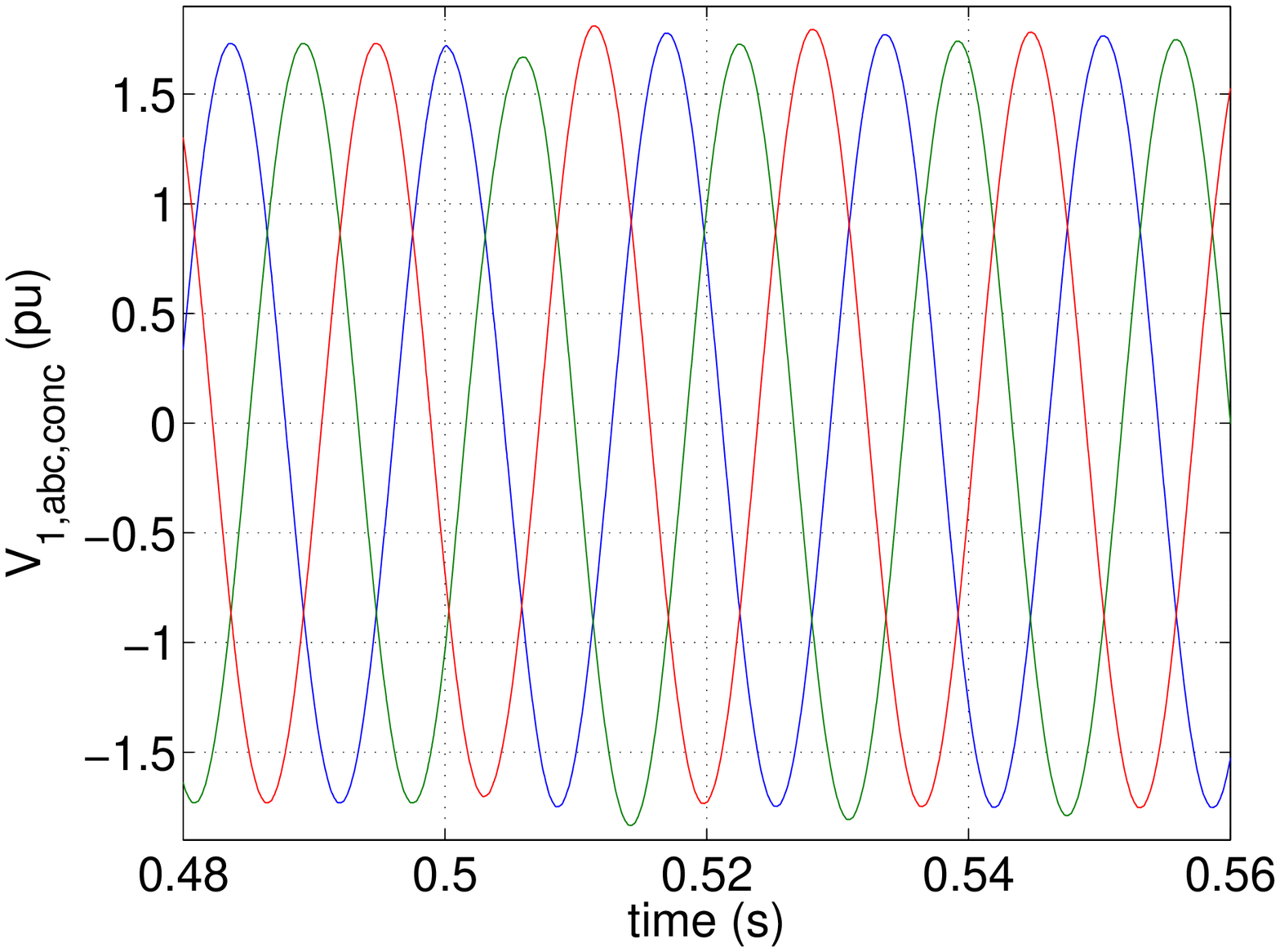}
                        \caption{Three-phase phase-to-phase voltage at $PCC_1$.}
                        \label{fig:UnbalabcVconc}
                      \end{subfigure}
                      \begin{subfigure}[!htb]{0.45\textwidth}
                        \centering
                        \includegraphics[width=1\textwidth, height=115pt]{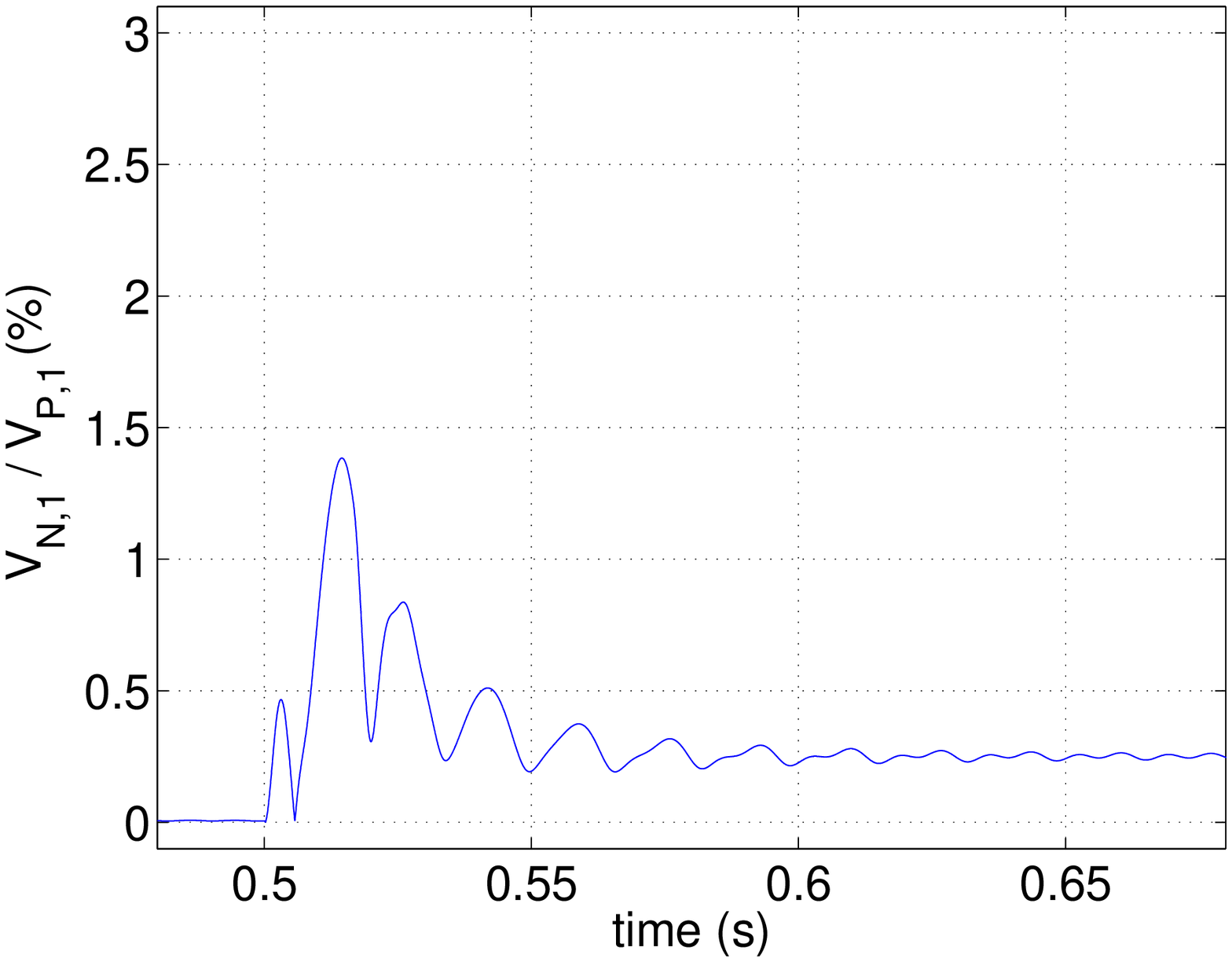}
                        \caption{Voltage imbalance ratio ($V_{N}/V_{P}$) in percentage.}
                        \label{fig:UnbalVnVp}
                      \end{subfigure}
                      \begin{subfigure}[!htb]{0.45\textwidth}
                        \centering
                        \includegraphics[width=1\textwidth, height=115pt]{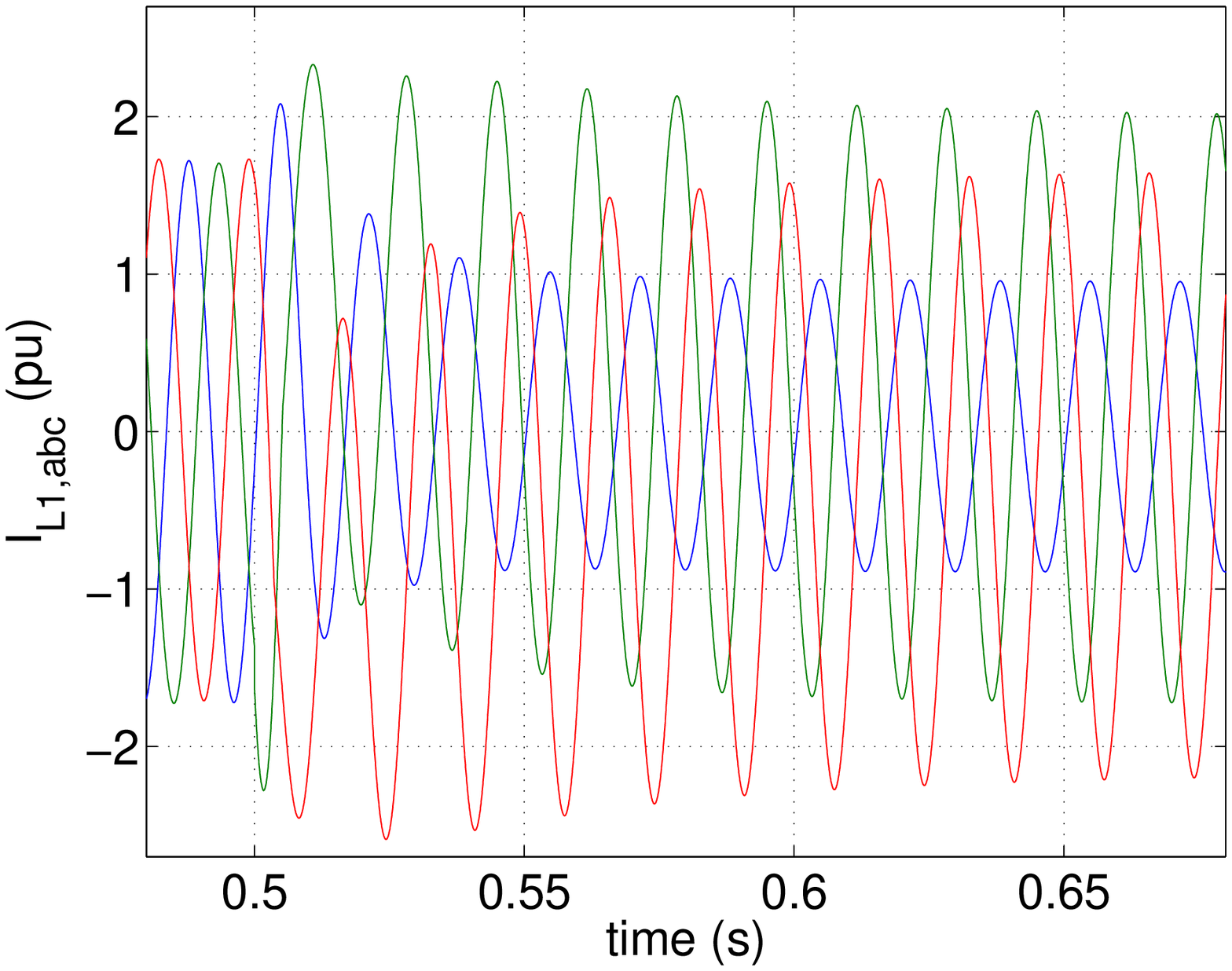}
                        \caption{Instantaneous load current $I_{L1}$.}
                        \label{fig:UnbalILabc}
                      \end{subfigure}                      
                      \caption{Performance of the PnP decentralized voltage control in presence of an unbalanced load.}
                      \label{fig:Unbal}                 
                    \end{figure}

               \subsubsection{Performance under unbalanced load conditions and robustness to different capacitance values}
                    \label{sec:differentC}
                    In this last test, we investigate the performance of PnP controller in presence of unbalanced loads and different capacitances at each PCCs. The capacitances are $C_{t1}=0.9C_t$ and $C_{t2}=1.1C_t$, but we designed controllers assuming the common capacitance value $C_t$ for each DGU. In other words, we aim at testing robustness of our control scheme against deviations from Assumption \ref{ass:ctrl}-(\ref{assum:line}). Voltage references are initially set to $V_{d,ref}=0.6$ pu and $V_{q,ref}=0.8$ pu for both DGUs. Moreover, the nominal RL load described in Section \ref{sec:scenario1voltagetrack1} is connected at $PCC_1$ and $PCC_2$. At $t=0.5$ s the RL load parameters at $PCC_1$ are changed to the values given in Table \ref{tbl:parunbload}, so that the load of DGU 1 becomes highly unbalanced. Figure \ref{fig:UnbaldqV_different_C} shows the \emph{d} and \emph{q} components of the load voltage at $PCC_1$ before and after unbalancing. We note that tracking of the reference signals is still guaranteed in spite of load changes. Moreover, instantaneous load voltages, shown in Figure \ref{fig:UnbalabcV_different_C}, confirm successful regulation of the output waveforms. Figure \ref{fig:UnbalILabc_different_C} shows the load current $I_{L1}$ provided by DGU 1 in the abc frame. One can notice how the controller induces major changes in the VSC behaviour in order to avoid spoiling the balance of load voltage at $PCC_1$. This test shows that, even if the controllers are designed considering capacitance $C_t$, thanks to the feedback, the integrators in the control loop guarantee robustness of stability with respect small deviations of capacitances $C_{ti}$ from $C_s$.\\ 
                    To evaluate the voltage imbalance at $PCC_1$, we calculate the ratio $V_{N}/V_{P}$ (expressed in $\%$), where $V_{N}$ and $V_{P}$ are the magnitudes of the negative- and positive-sequence components of the phase-to-phase voltage. The time evolution of this ratio is represented in Figure \ref{fig:UnbalVnVp_different_C}. We notice that it is always below 0.5$\%$ which is less than the maximum permissible value (3$\%$) defined by IEEE in \cite{IEEE2009}. 
                    \begin{figure}[!htb]
                      \centering
                      \begin{subfigure}[!htb]{0.45\textwidth}
                        \centering
                        \includegraphics[width=1\textwidth, height=115pt]{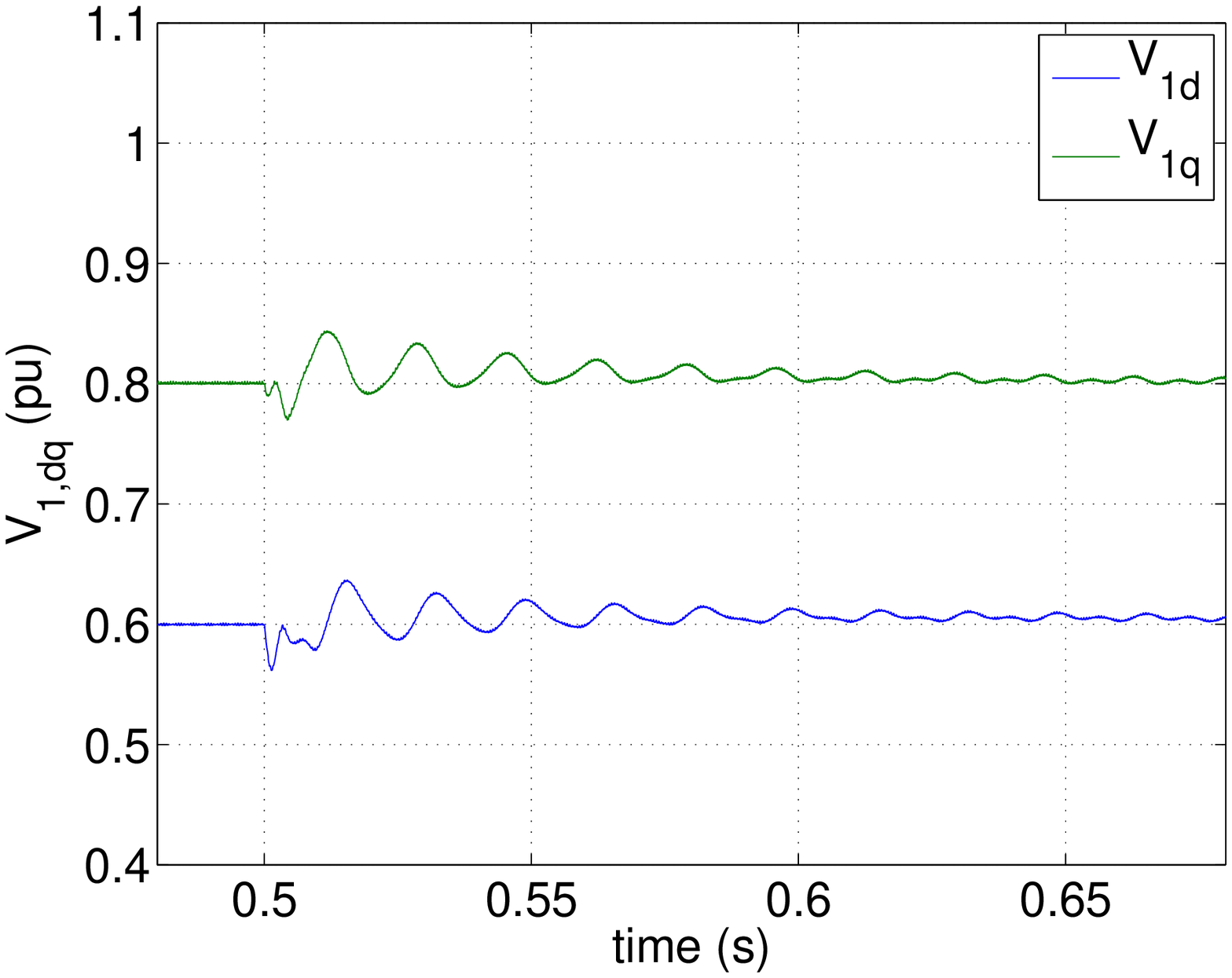}
                        \caption{\emph{d} and \emph{q} components of the voltage at $PCC_1$.}
                        \label{fig:UnbaldqV_different_C}
                      \end{subfigure}
                      \begin{subfigure}[!htb]{0.45\textwidth}
                        \centering
                        \includegraphics[width=1\textwidth, height=115pt]{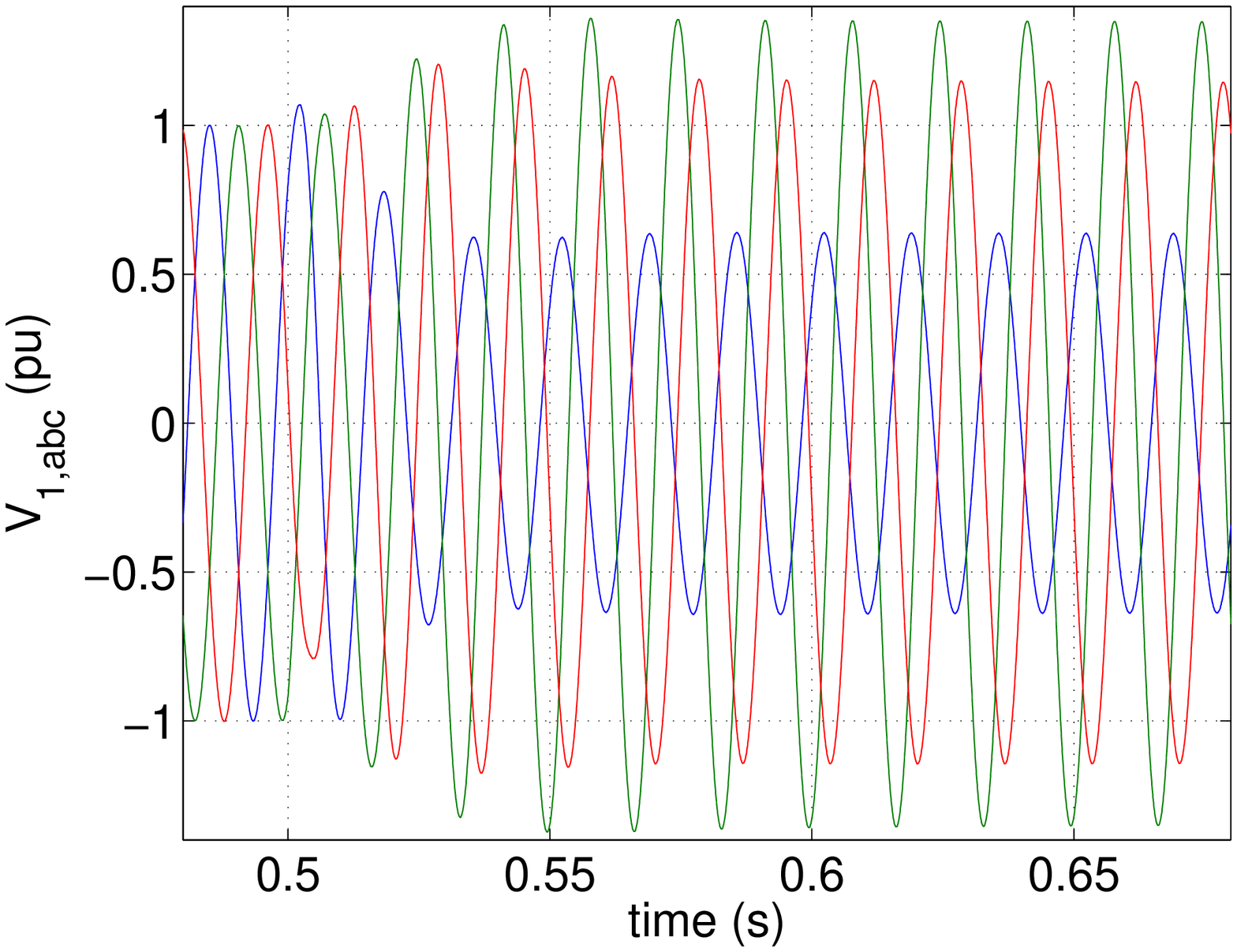}
                        \caption{Three-phase voltage at $PCC_1$.}
                        \label{fig:UnbalabcV_different_C}
                      \end{subfigure}
                      \begin{subfigure}[!htb]{0.45\textwidth}
                        \centering
                        \includegraphics[width=1\textwidth, height=115pt]{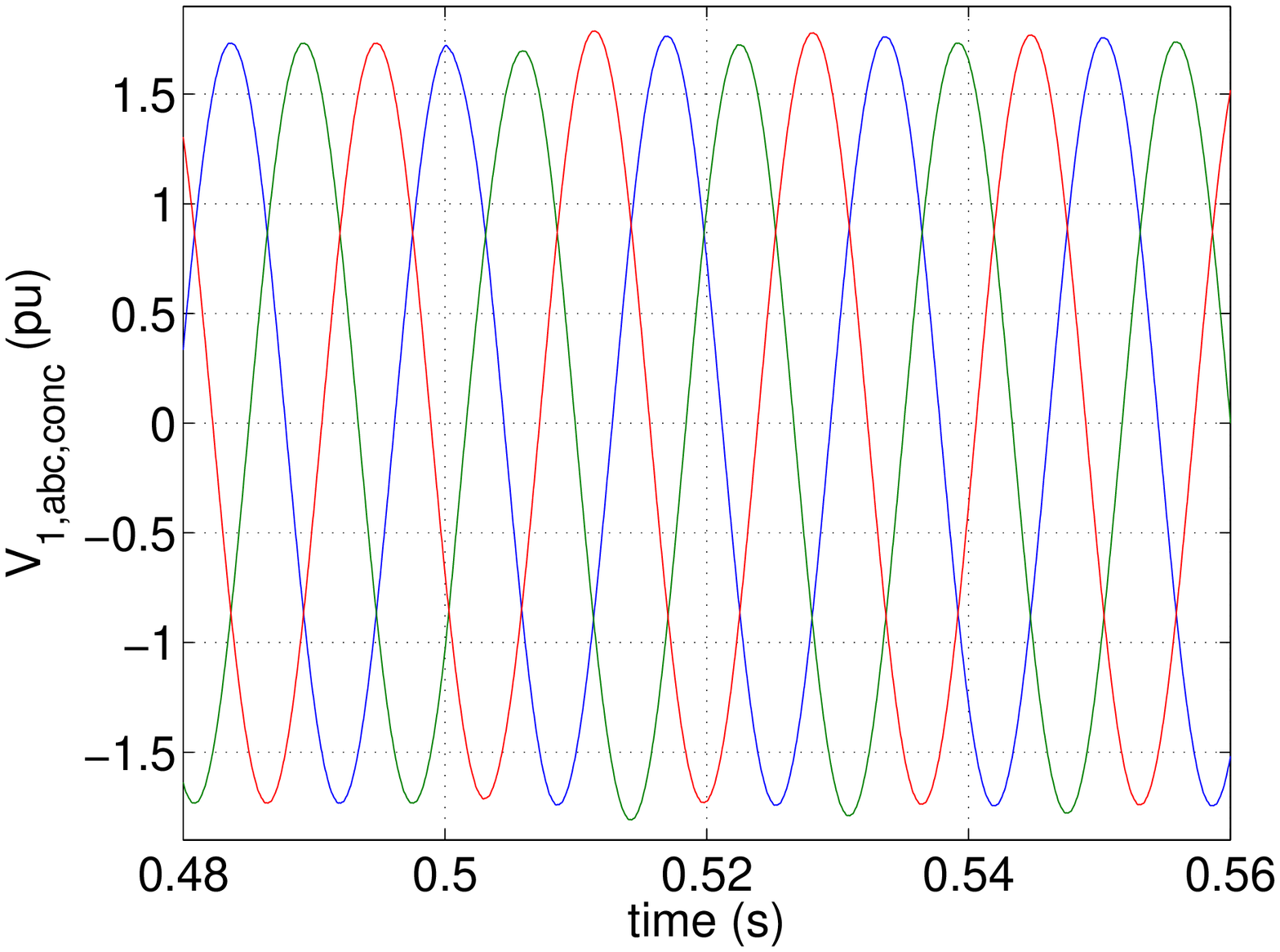}
                        \caption{Three-phase phase-to-phase voltage at $PCC_1$.}
                        \label{fig:UnbalabcVconc_different_C}
                      \end{subfigure}
                      \begin{subfigure}[!htb]{0.45\textwidth}
                        \centering
                        \includegraphics[width=1\textwidth, height=115pt]{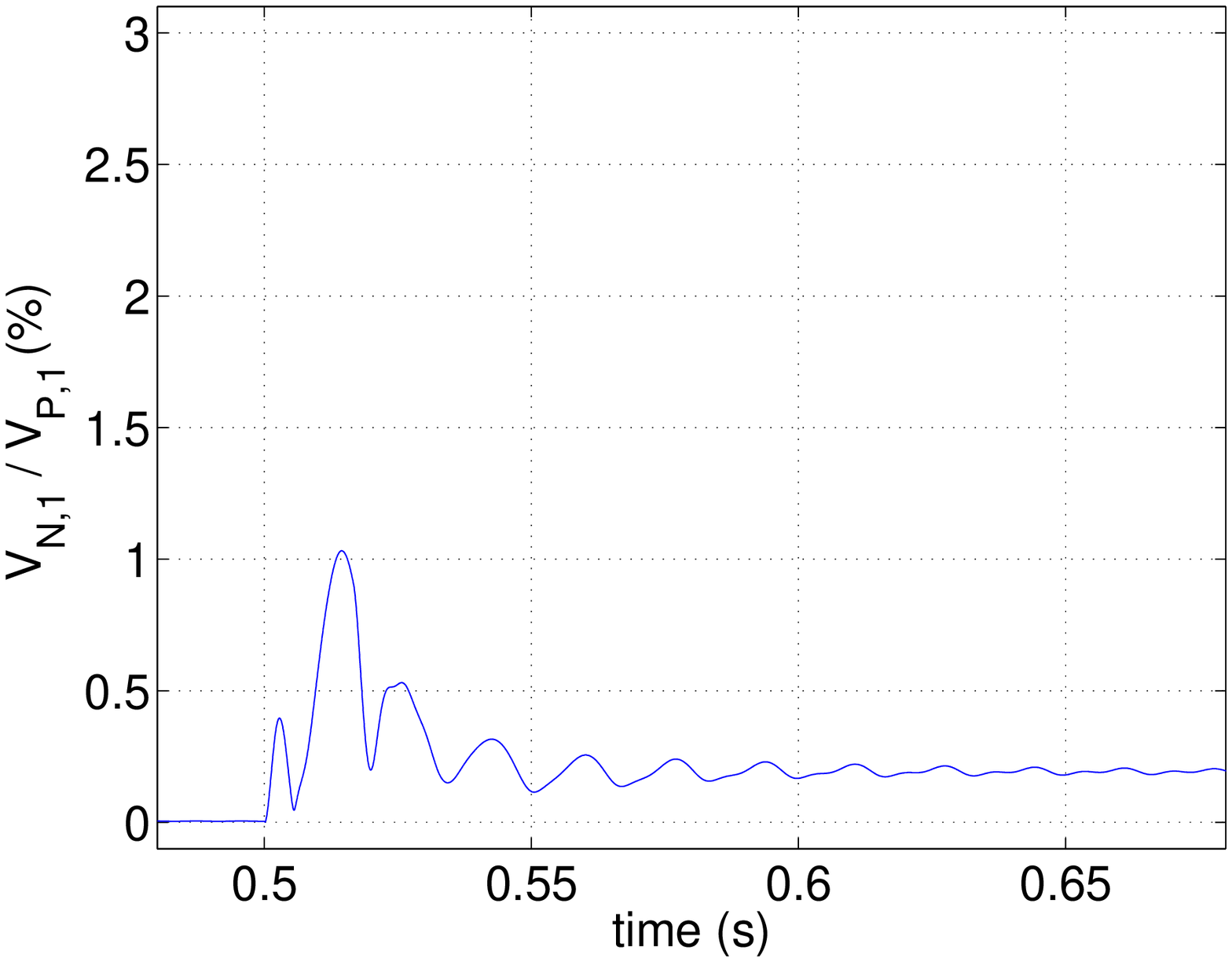}
                        \caption{Voltage imbalance ratio ($V_{N}/V_{P}$) in percentage.}
                        \label{fig:UnbalVnVp_different_C}
                      \end{subfigure}
                      \begin{subfigure}[!htb]{0.45\textwidth}
                        \centering
                        \includegraphics[width=1\textwidth, height=115pt]{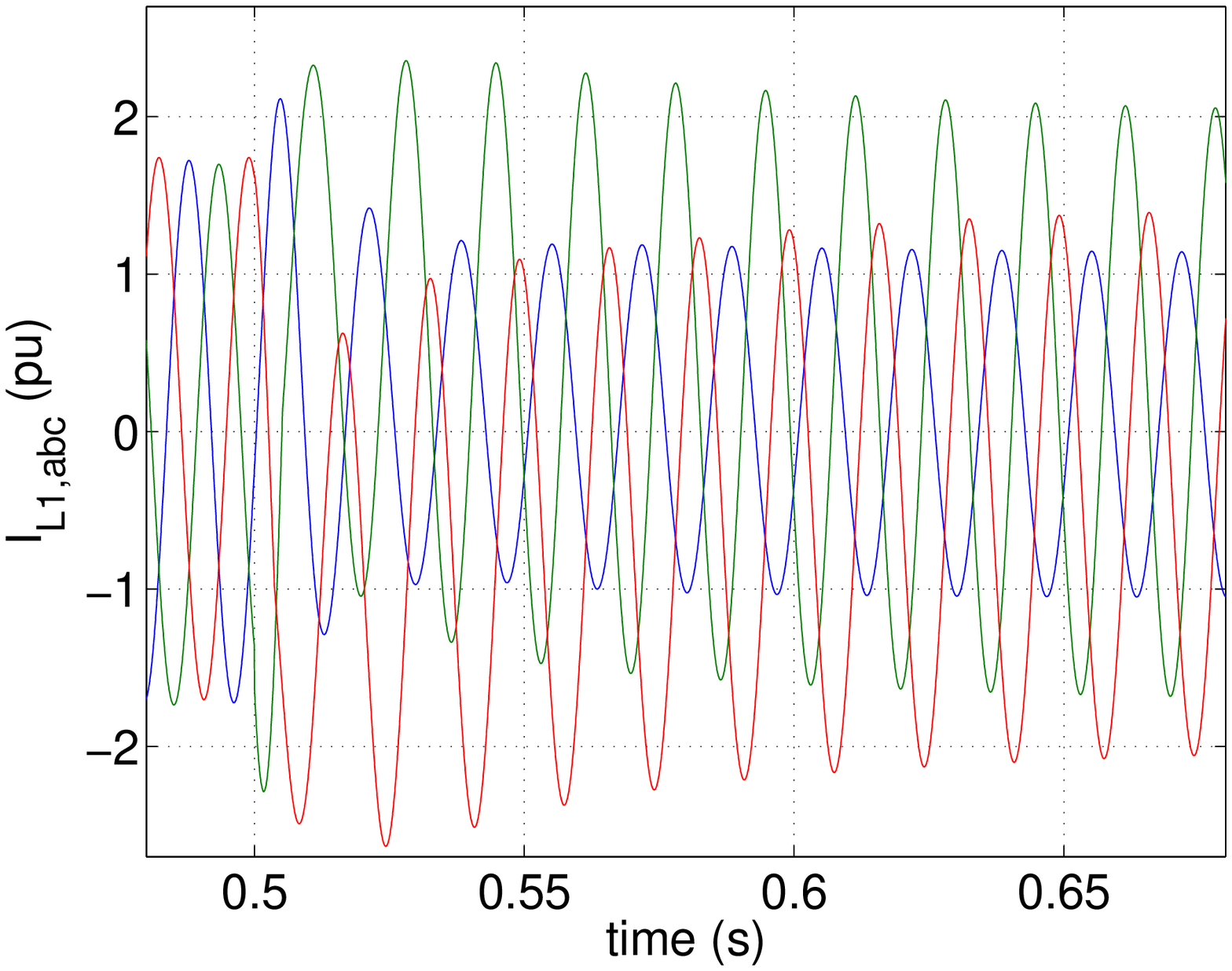}
                        \caption{Instantaneous load current $I_{L1}$.}
                        \label{fig:UnbalILabc_different_C}
                      \end{subfigure}                      
                      \caption{Performance of the PnP decentralized voltage control in presence of an unbalanced load.}
                      \label{fig:Unbal_different_C}                 
                    \end{figure}

          \clearpage

          \subsection{Scenario 2}
               \label{sec:scenario2}
               In this second scenario, we consider a meshed ImG composed by ten DGUs. The topology of the network is shown in Figure \ref{fig:10areas}. Differently from Scenario 1, some DGUs have more than one neighbour and hence the disturbances that will influence their dynamics will be greater. Furthermore, it is also present a loop that further complicates voltage regulation. We highlight that, to our knowledge, the control of loop-connected DGUs in a meshed-topology has been never attempted before in the literature.
               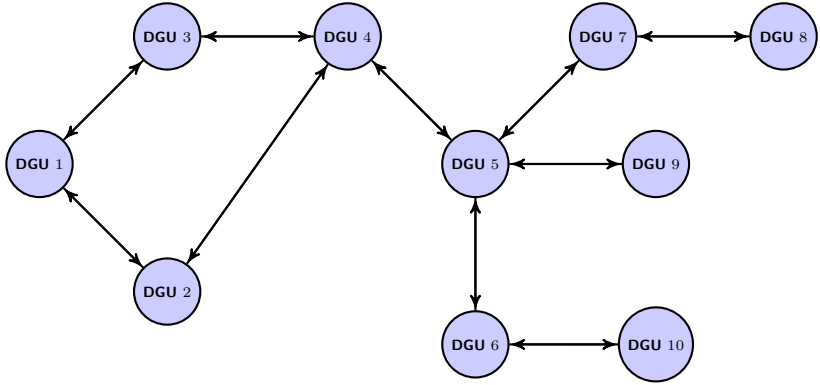
\begin{figure}[!htb]
                 \centering
                 \begin{tikzpicture}[scale=0.8,transform shape,->,>=stealth',shorten >=1pt,auto,node distance=3cm, thick,main node/.style={circle,fill=blue!20,draw,font=\sffamily\bfseries}]
					  
  \node[main node] (1) {\scriptsize{DGU $1$}};
  \node[main node] (2) [below right of=1] {\scriptsize{DGU $2$}};
  \node[main node] (3) [above right of=1] {\scriptsize{DGU $3$}};
  \node[main node] (4) [right of=3] {\scriptsize{DGU $4$}};
  \node[main node] (5) [below right of=4] {\scriptsize{DGU $5$}};
  \node[main node] (6) [below of=5] {\scriptsize{DGU $6$}};
  \node[main node] (7) [above right of=5] {\scriptsize{DGU $7$}};
  \node[main node] (8) [right of=7] {\scriptsize{DGU $8$}};
  \node[main node] (9) [right of=5] {\scriptsize{DGU $9$}};
  \node[main node] (10) [right of=6] {\scriptsize{DGU $10$}};
  
  \path[every node/.style={font=\sffamily\small}]
  (1) edge node [left] {} (3)
  (3) edge node [right] {} (1)

  (1) edge node [left] {} (2)
  (2) edge node [right] {} (1)
  
  (3) edge node [left] {} (4)
  (4) edge node [right] {} (3)
  
  (2) edge node [left] {} (4)
  (4) edge node [right] {} (2)
  
  (4) edge node [left] {} (5)
  (5) edge node [right] {} (4)
  
  (5) edge node [left] {} (6)
  (6) edge node [right] {} (5)
  
  (6) edge node [left] {} (10)
  (10) edge node [right] {} (6)
  
  (5) edge node [left] {} (9)
  (9) edge node [right] {} (5)
  
  (5) edge node [left] {} (7)
  (7) edge node [right] {} (5)
  
  (7) edge node [left] {} (8)
  (8) edge node [right] {} (7)
  ;
\end{tikzpicture}
                 \caption{Scheme of the microgrid composed of 10 DGUs. Arrows correspond to transmission lines.}
                 \label{fig:10areas}
               \end{figure}
               
               For all the subsystems $\subss{\hat{\Sigma}}{i}^{DGU}$, $i\in\DD=\{1,\dots,10\}$, we execute Algorithm \ref{alg:ctrl_design} in order to design controllers $\subss{\CC}{i}$ and compensators $\subss{\tilde{C}}{i}$ and $\subss{N}{i}$. We have chosen identical closed-loop transfer function $\subss{\tilde{F}}{i}(s)$, $i\in\DD$ equal to low-pass with 0 dB  DC gain and 1kHz bandwidth. Figure \ref{fig:closedLoop10Areas} shows that the closed-loop eigenvalues of the QSL microgrid (in red) are asymptotically stable. Figure \ref{fig:singularValuePreFilter10Areas} shows the singular value plots of the pre-filters $\subss{\tilde{C}}{i}$, $i\in\DD$ obtained solving Step \ref{enu:stepBalgCtrl} of Algorithm \ref{alg:ctrl_design}. With the addition of the pre-filters in the control loop of each DGUs, the bandwidth of the closed-loop transfer function $F(s)$ of the overall microgrid (in blue) is wider as shown in Figure \ref{fig:singularValueClosedLoop10Areas} (in green). Note that the closed-loop transfer function, with the addition of the pre-filters, coincides with the desired one $\subss{\tilde{F}}{i}(s)$. Finally, Figure \ref{fig:singularValueCompensator10Areas} shows the singular value plots of the disturbance compensators $\subss{N}{i}$ synthesized in Step \ref{enu:stepCalgCtrl} of Algorithm \ref{alg:ctrl_design}. Parameters of the microgrid are listed in Tables \ref{tbl:diffpar10}, \ref{tbl:linespar10} and \ref{tbl:commpar10} in Appendix \ref{sec:AppElectrPar}.
               \begin{figure}[!htb]
                 \centering
                 \begin{subfigure}[!htb]{0.48\textwidth}
                   \centering
                   \includegraphics[width=1\textwidth, height=130pt]{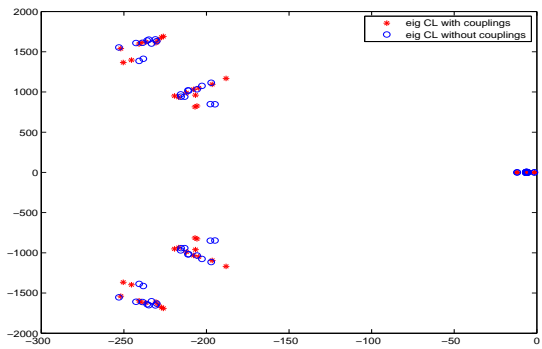}
                   \caption{Eigenvalues of the closed-loop (CL) QSL microgrid with (red) and without (blue) couplings.}
                   \label{fig:closedLoop10Areas}
                 \end{subfigure}
                 \begin{subfigure}[!htb]{0.48\textwidth}
                   \centering
                   \includegraphics[width=1\textwidth, height=130pt]{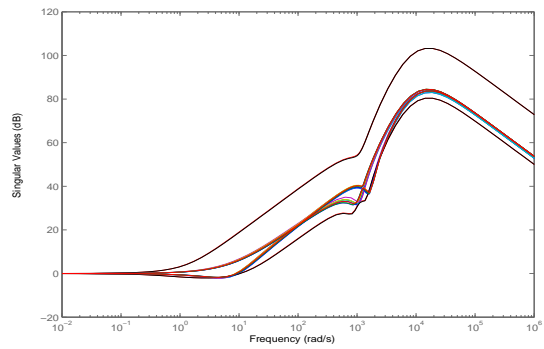}
                   \caption{Singular values of the pre-filters.}
                   \label{fig:singularValuePreFilter10Areas}
                 \end{subfigure}
                 \begin{subfigure}[!htb]{0.48\textwidth}
                   \centering
                   \includegraphics[width=1\textwidth, height=130pt]{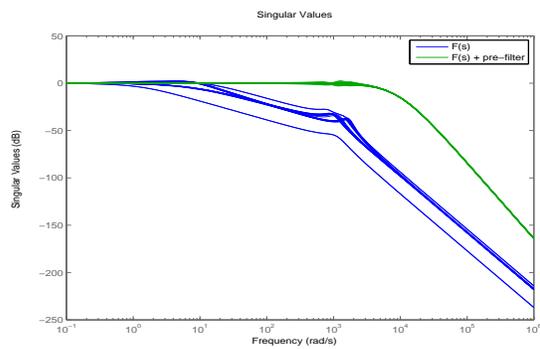}
                   \caption{Singular values of $F(s)$ with (green) and without (blue) pre-filters.}
                   \label{fig:singularValueClosedLoop10Areas}
                 \end{subfigure}
                 \begin{subfigure}[!htb]{0.48\textwidth}
                   \includegraphics[width=1\textwidth, height=130pt]{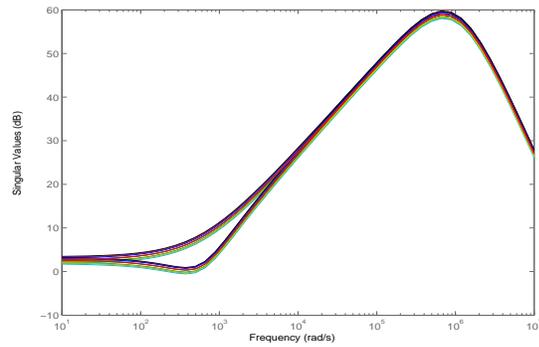}
                   \caption{Singular values of disturbances compensators.}
                   \label{fig:singularValueCompensator10Areas}
                 \end{subfigure}
                 \caption{Features of PnP controllers for Scenario 2.}
                 \label{fig:closedLoop10areas}
               \end{figure}

               \subsubsection{Voltage tracking properties}
                    \label{sec:voltagetrackScenario2}
                    In this test, we evaluate the tracking capabilities of PnP voltage controllers applied to the ImG of Scenario 2. At the $PCC$ of each DGU, we connect a different RL test load. Parameter values are given in Table \ref{tbl:parsim10} in Appendix \ref{sec:AppElectrPar} and they are kept constant during the simulation. The \emph{d} and \emph{q} components of the voltage at $PCC_{i}$, $i\in\DD$ are initially set at 1 pu and 0 pu, respectively. From $t=0.8$ s, the reference signals of each DGUs in a step-wise fashion. The final values and the time of these step changes are listed in Table \ref{tbl:parsim10} in Appendix \ref{sec:AppElectrPar}. Figure \ref{fig:trackdq10areas} shows the responses of all DGUs. We highlight that, in spite of a rather complex microgrid topology, PnP decentralized control ensures the tracking of the voltage references for all DGUs.                    
                    \begin{figure}[!htb]
                      \centering
                      \begin{subfigure}[!htb]{0.32\textwidth}
                        \centering
                        \includegraphics[width=1\textwidth]{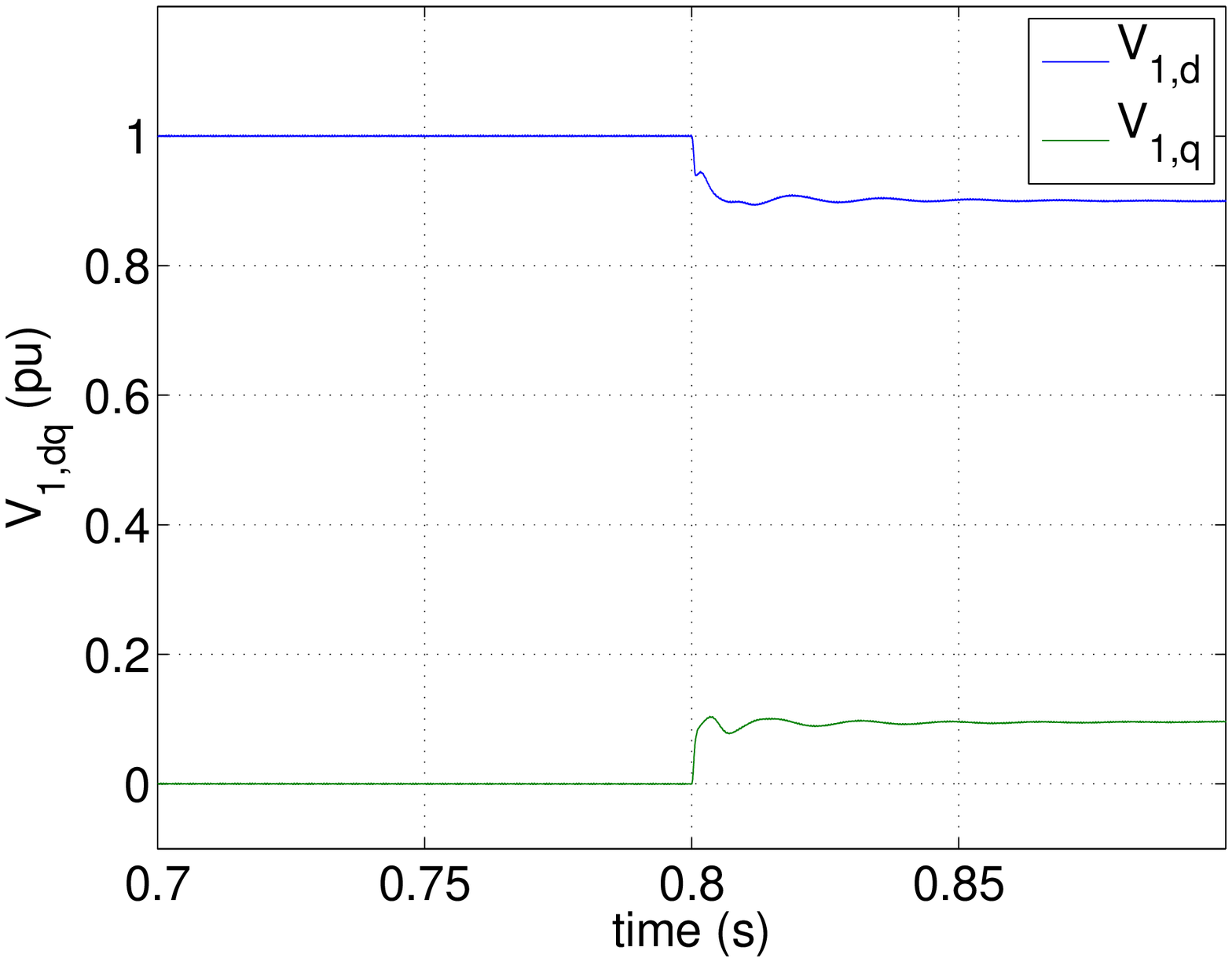}
                      \end{subfigure}
                      \begin{subfigure}[!htb]{0.32\textwidth}
                        \centering
                        \includegraphics[width=1\textwidth]{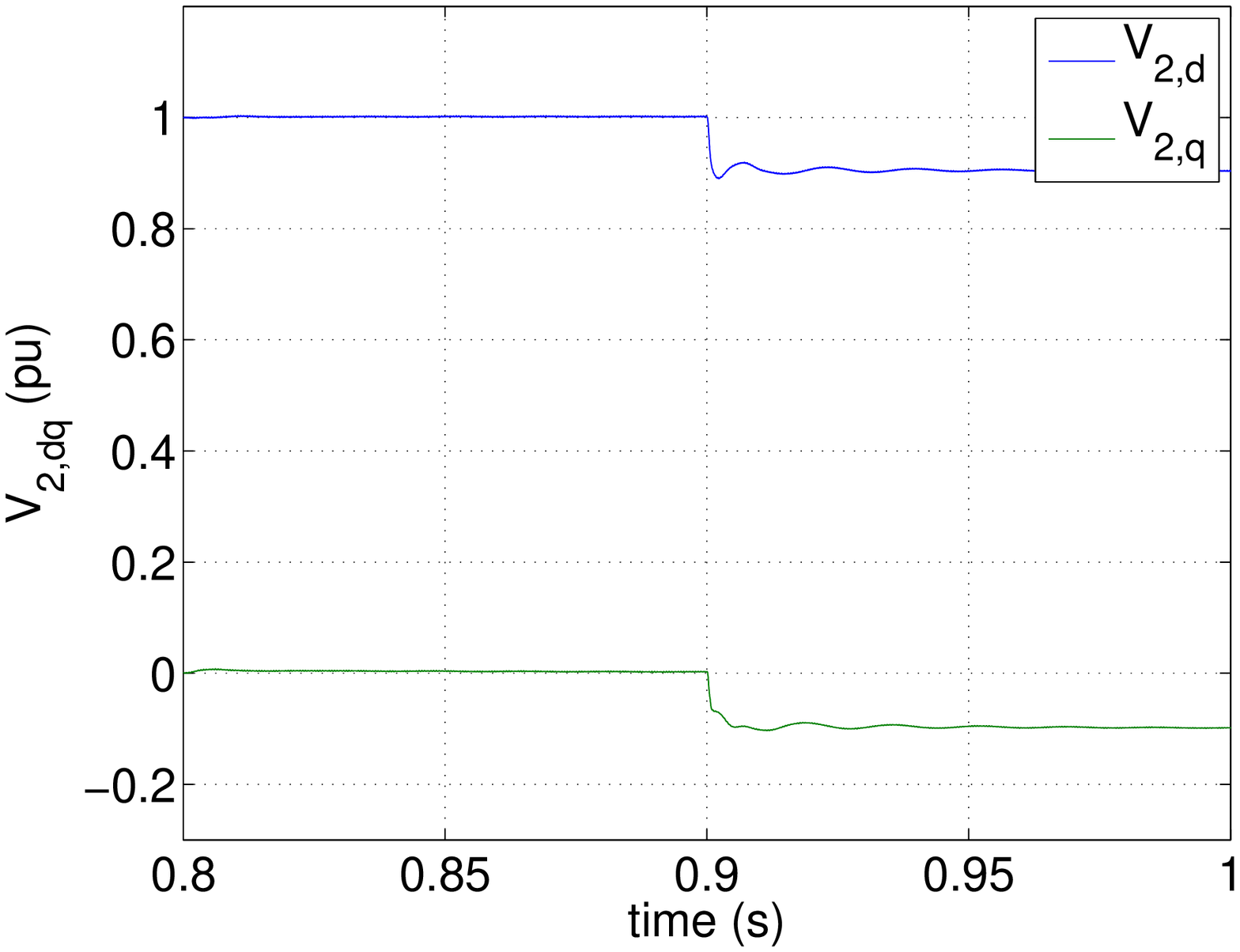}
                      \end{subfigure}
                      \begin{subfigure}[!htb]{0.32\textwidth}
                        \centering
                        \includegraphics[width=1\textwidth]{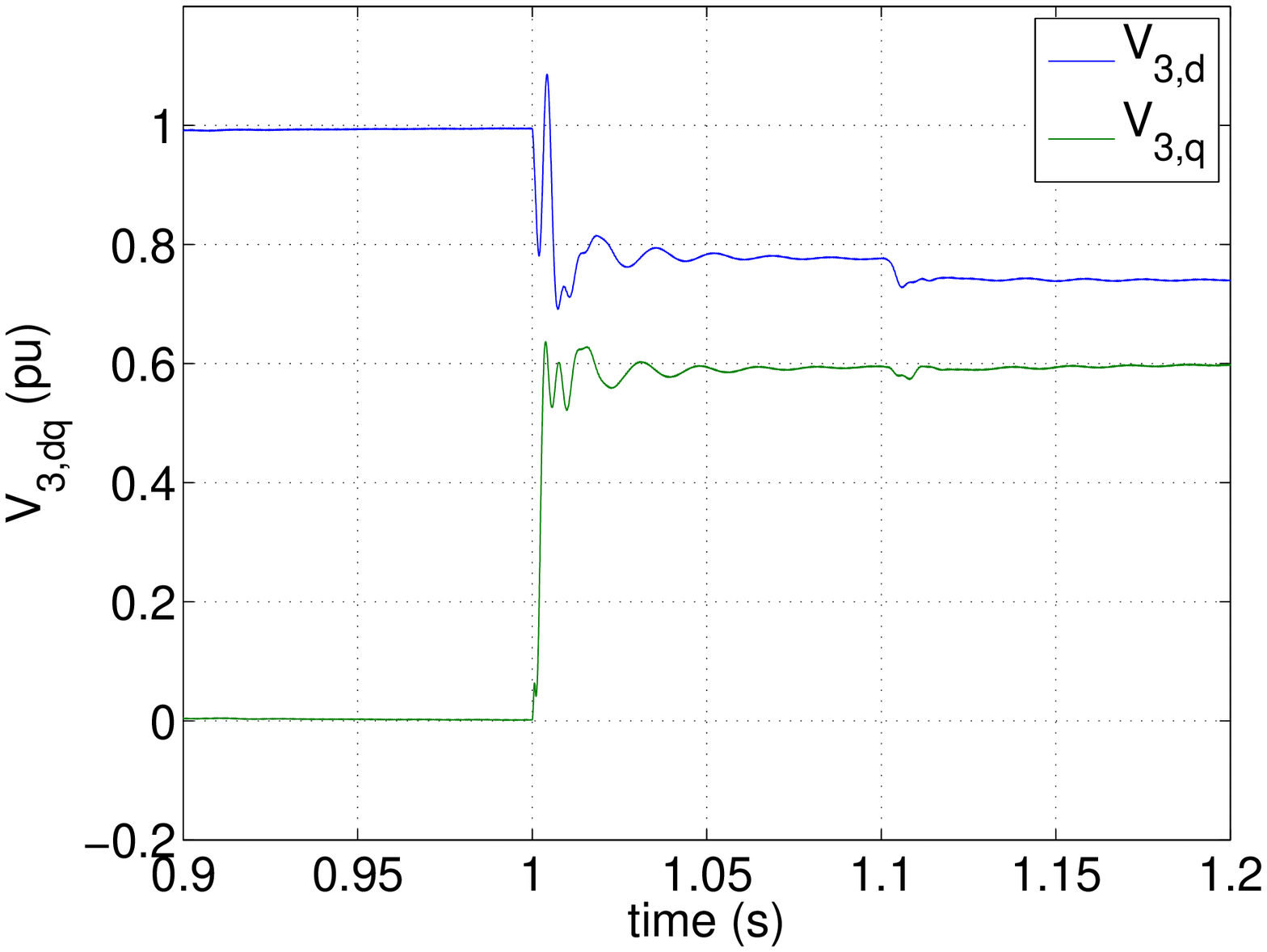}
                      \end{subfigure}
                      \begin{subfigure}[!htb]{0.32\textwidth}
                        \centering
                        \includegraphics[width=1\textwidth]{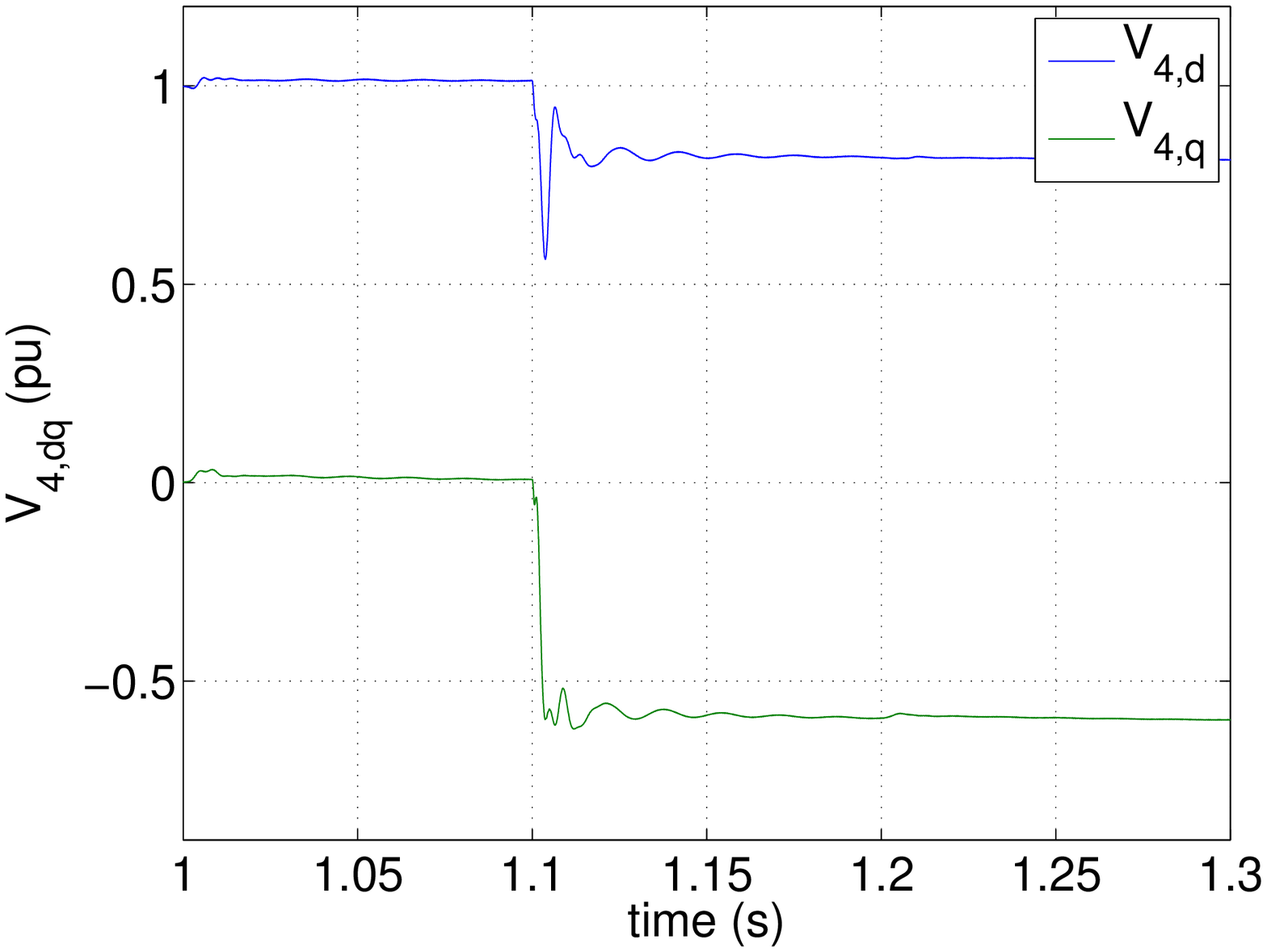}
                      \end{subfigure}
                      \begin{subfigure}[!htb]{0.32\textwidth}
                        \centering
                        \includegraphics[width=1\textwidth]{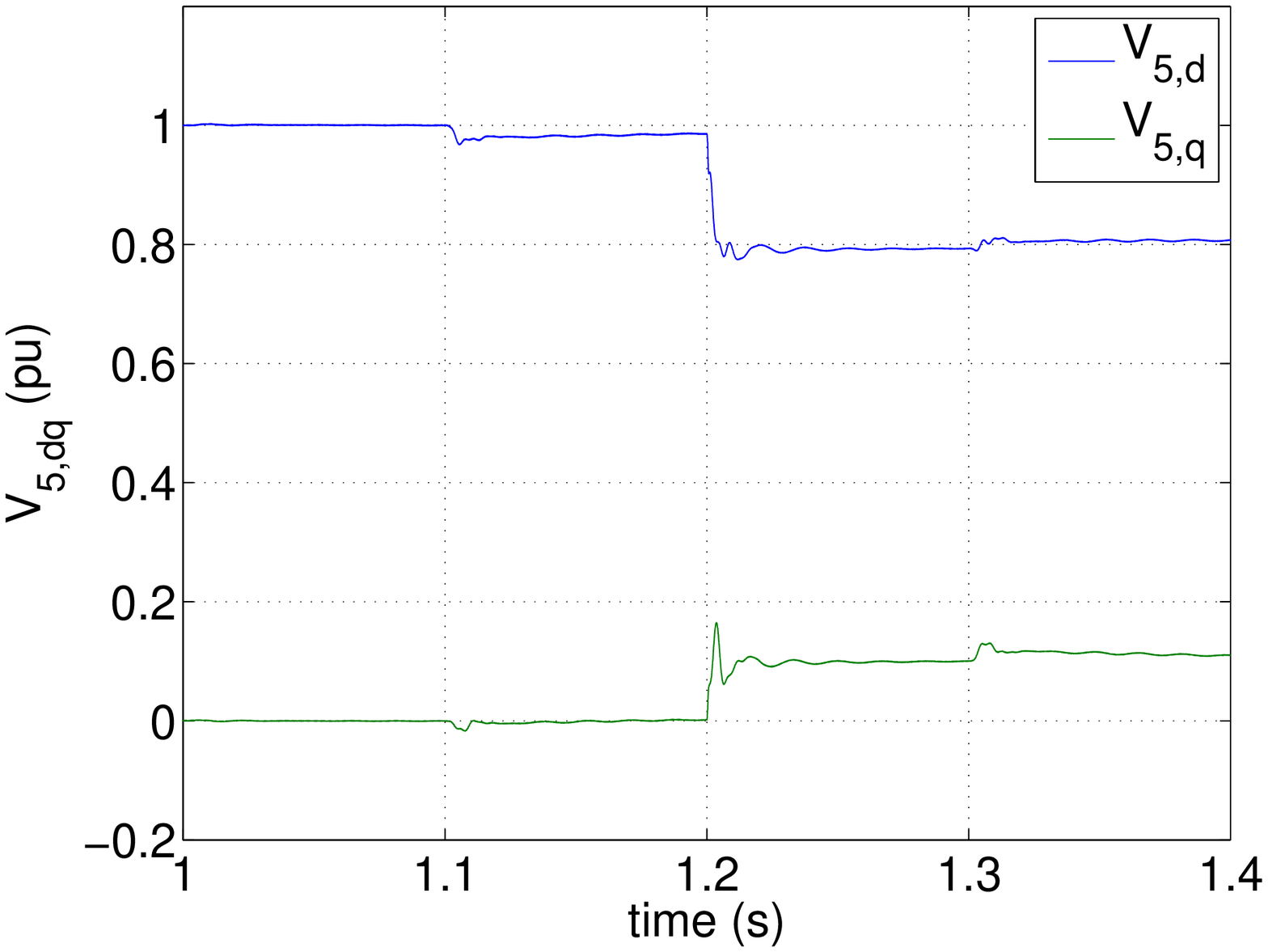}
                      \end{subfigure}
                      \begin{subfigure}[!htb]{0.32\textwidth}
                        \centering
                        \includegraphics[width=1\textwidth]{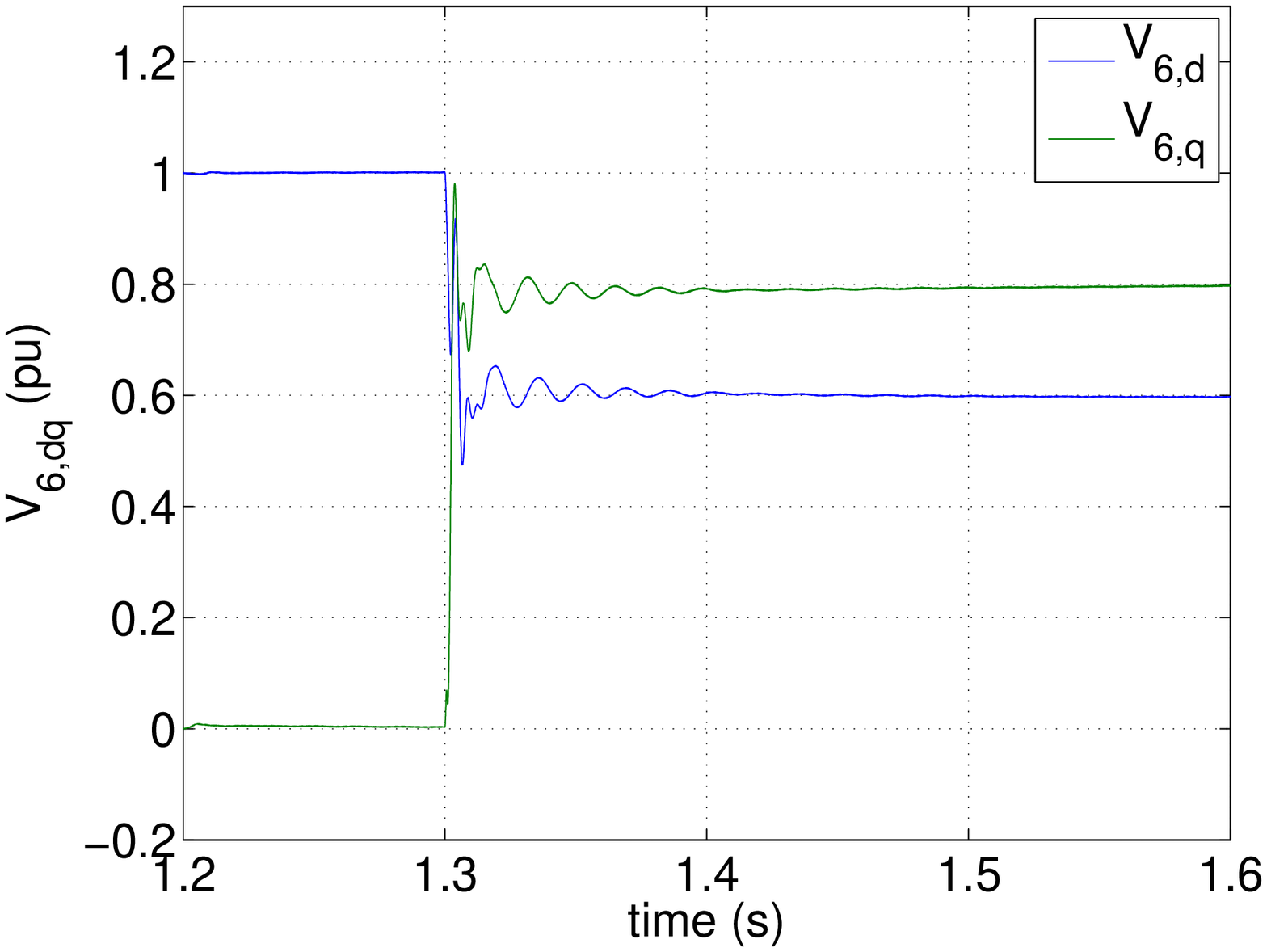}
                      \end{subfigure}
                      \begin{subfigure}[!htb]{0.32\textwidth}
                        \centering
                        \includegraphics[width=1\textwidth]{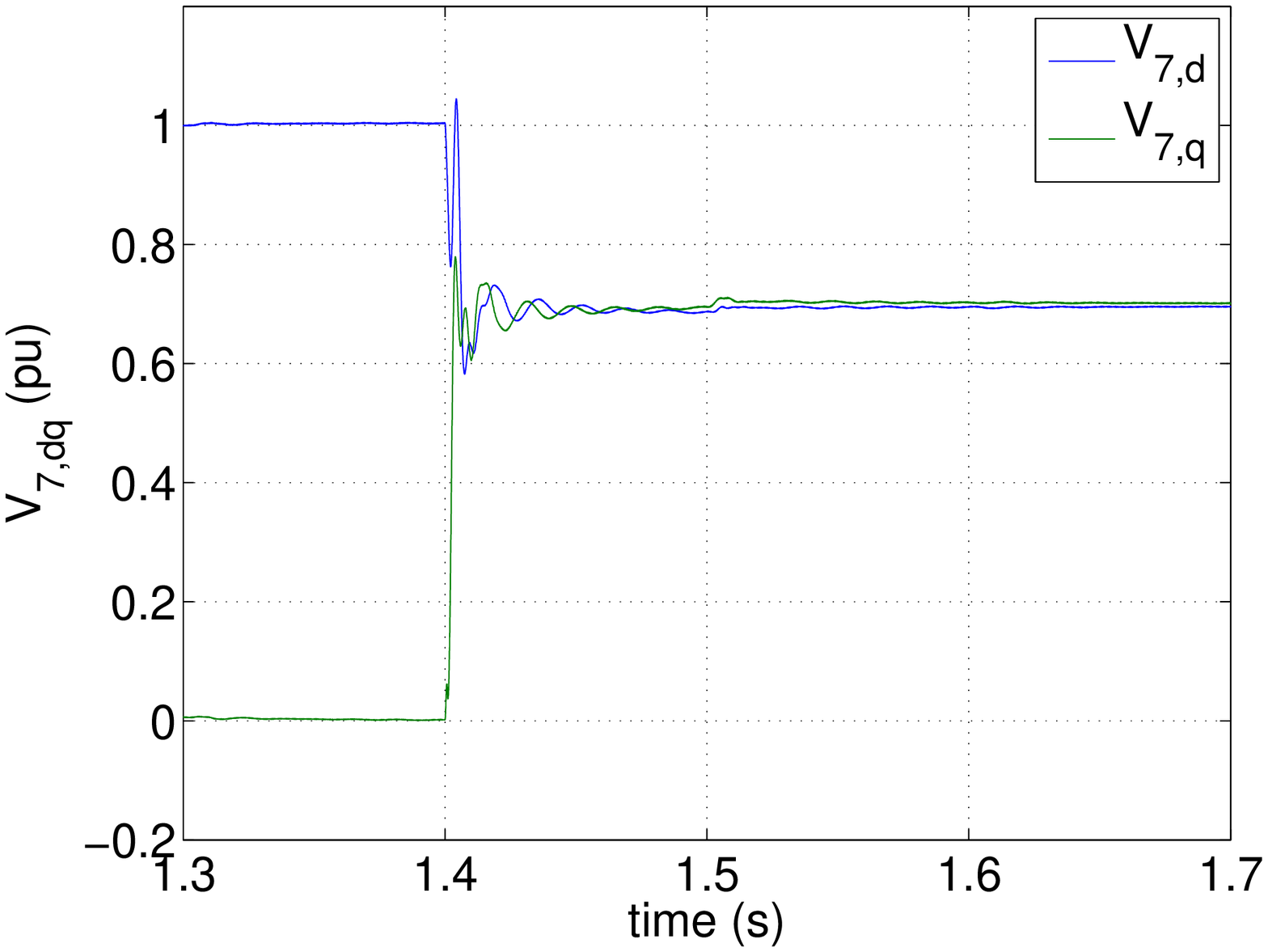}
                      \end{subfigure}
                      \begin{subfigure}[!htb]{0.32\textwidth}
                        \centering
                        \includegraphics[width=1\textwidth]{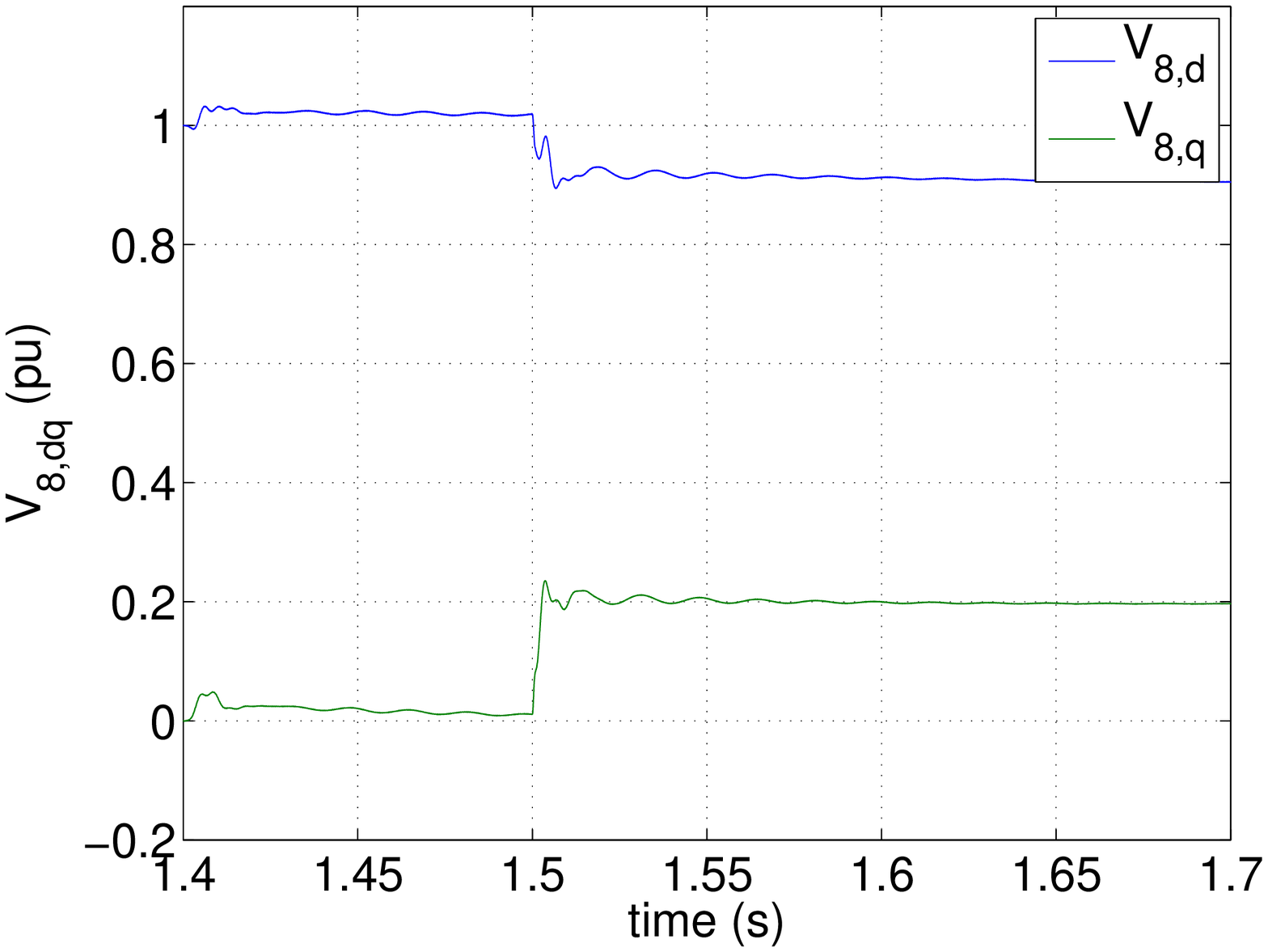}
                      \end{subfigure}
                      \begin{subfigure}[!htb]{0.32\textwidth}
                        \centering
                        \includegraphics[width=1\textwidth]{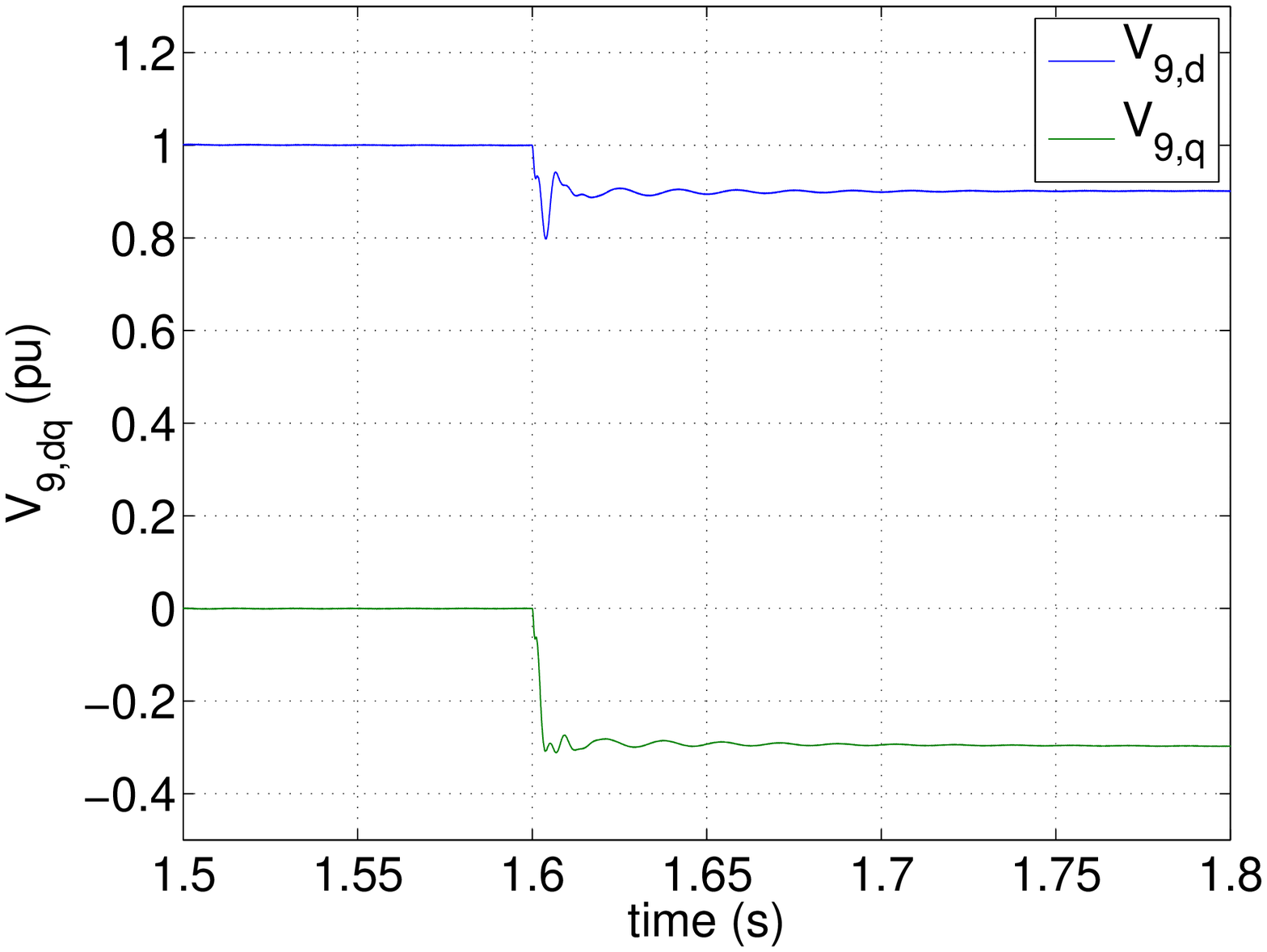}
                      \end{subfigure}
                      \begin{subfigure}[!htb]{0.32\textwidth}
                        \centering
                        \includegraphics[width=1\textwidth]{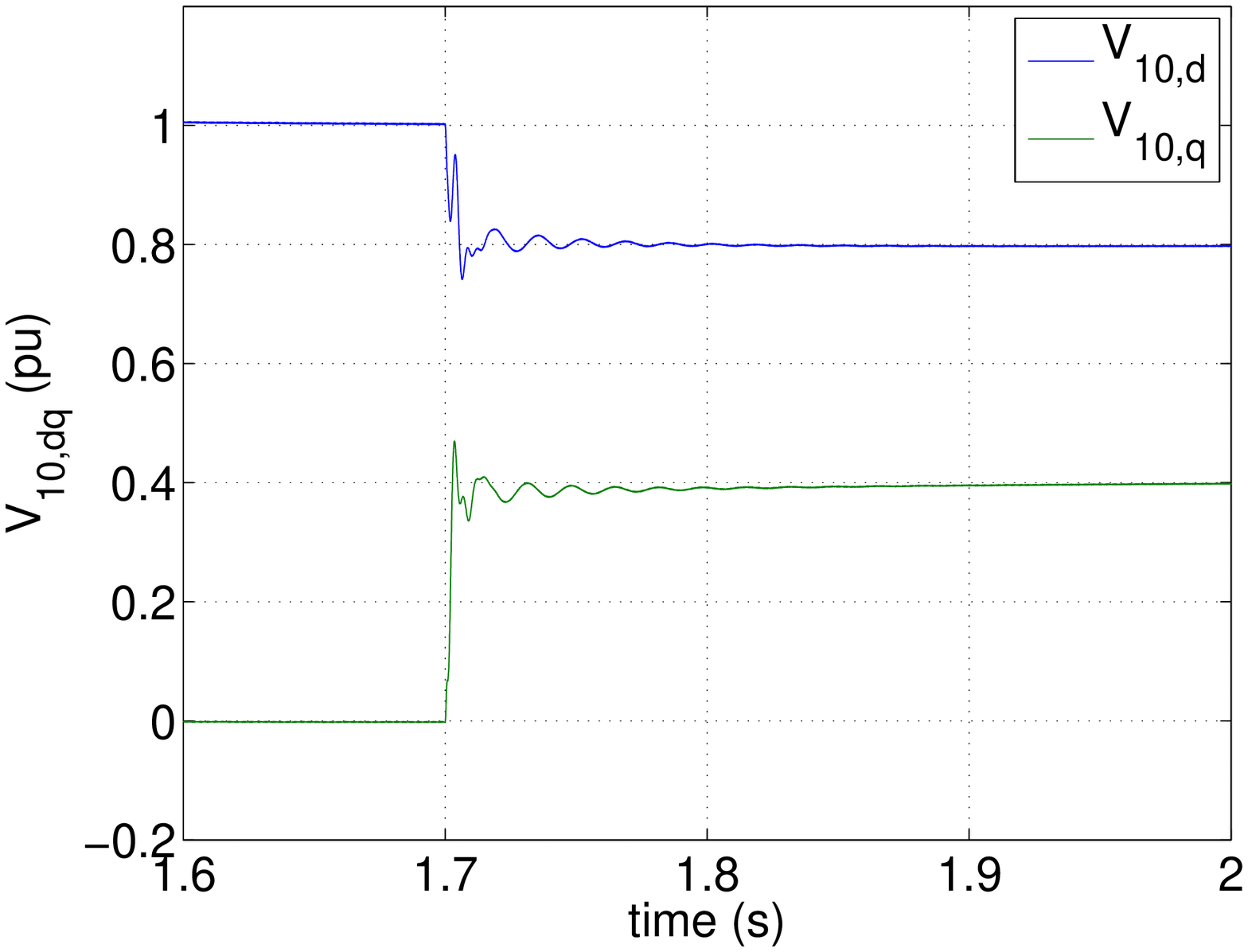}
                      \end{subfigure}
                      \caption{Performance of PnP decentralized voltage control in terms of set-point tracking.}
                      \label{fig:trackdq10areas}
                    \end{figure}

               \clearpage

               \subsubsection{Plug-in of a new DGU}
                    \label{sec:plugintest}
                    For evaluating the PnP capabilities of our control approach, we simulate the connection of a new DGU at the ImG in Figure \ref{fig:10areasplug}. In particular, we connect $\subss{\hat{\Sigma}}{11}^{DGU}$ as shown in Figure \ref{fig:10areasplug}. Therefore, the set of neighbours of DGU 11 is $\NN_{11}=\{1,6\}$. Since local dynamics of $\subss{\hat{\Sigma}}{j}^{DGU}$, $j\in\NN_{11}$ now depend on new parameters, a retuning of their controllers is needed. As described in Section \ref{sec:PnP}, only subsystems $\subss{\hat{\Sigma}}{j}^{DGU}$, $j\in\NN_{11}$ must update their controllers $\subss{\CC}{j}$ and compensators $\subss{\tilde{C}}{j}$ and $\subss{N}{j}$. More precisely, first matrices $\hat{A}_{jj}$, $j\in\NN_{11}$ are updated. Then, we proceed by re-executing Algorithm \ref{alg:ctrl_design} for each DGU $\subss{\hat{\Sigma}}{j}^{DGU}$, $j\in\NN_{11}$. Finally, we execute Algorithm \ref{alg:ctrl_design} for synthesizing controller $\subss{\CC}{11}$ and compensators $\subss{\tilde{C}}{11}$ and $\subss{N}{11}$ for the new DGU. Since Algorithm \ref{alg:ctrl_design} never stops in Step \ref{enu:stepAalgCtrl}, the plug-in of $\subss{\hat{\Sigma}}{11}^{DGU}$ is allowed and local controllers get replaced by the new ones.
                    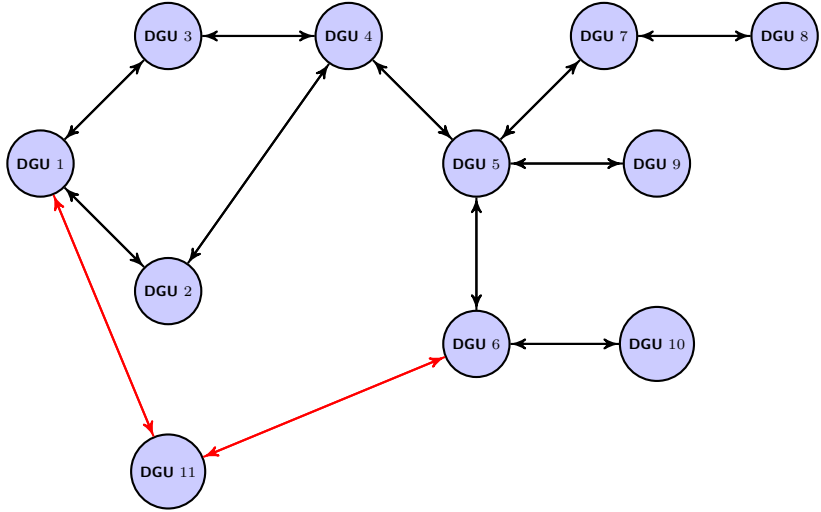
\begin{figure}[!htb]
                      \centering
                      \begin{tikzpicture}[scale=0.8,transform shape,->,>=stealth',shorten >=1pt,auto,node distance=3cm, thick,main node/.style={circle,fill=blue!20,draw,font=\sffamily\bfseries}]
					  
  \node[main node] (1) {\scriptsize{DGU $1$}};
  \node[main node] (2) [below right of=1] {\scriptsize{DGU $2$}};
  \node[main node] (3) [above right of=1] {\scriptsize{DGU $3$}};
  \node[main node] (4) [right of=3] {\scriptsize{DGU $4$}};
  \node[main node] (5) [below right of=4] {\scriptsize{DGU $5$}};
  \node[main node] (6) [below of=5] {\scriptsize{DGU $6$}};
  \node[main node] (7) [above right of=5] {\scriptsize{DGU $7$}};
  \node[main node] (8) [right of=7] {\scriptsize{DGU $8$}};
  \node[main node] (9) [right of=5] {\scriptsize{DGU $9$}};
  \node[main node] (10) [right of=6] {\scriptsize{DGU $10$}};
  \node[main node] (11) [below of=2] {\scriptsize{DGU $11$}};
  
  \path[every node/.style={font=\sffamily\small}]
  (1) edge node [left] {} (3)
  (3) edge node [right] {} (1)

  (1) edge node [left] {} (2)
  (2) edge node [right] {} (1)
  
  (3) edge node [left] {} (4)
  (4) edge node [right] {} (3)
  
  (2) edge node [left] {} (4)
  (4) edge node [right] {} (2)
  
  (4) edge node [left] {} (5)
  (5) edge node [right] {} (4)
  
  (5) edge node [left] {} (6)
  (6) edge node [right] {} (5)
  
  (6) edge node [left] {} (10)
  (10) edge node [right] {} (6)
  
  (5) edge node [left] {} (9)
  (9) edge node [right] {} (5)
  
  (5) edge node [left] {} (7)
  (7) edge node [right] {} (5)
  
  (7) edge node [left] {} (8)
  (8) edge node [right] {} (7);

  \draw[red] (1) to (11);
  \draw[red] (11) to (1);
  \draw[red] (6) to (11);
  \draw[red] (11) to (6);

\end{tikzpicture}
                      \caption{Scheme of the microgrid after the addition of $\subss{\hat{\Sigma}}{11}^{DGU}$.}
                      \label{fig:10areasplug}
                    \end{figure}

                    Figure \ref{fig:ctrlcomp} shows a comparison between the compensators $\subss{\tilde{C}}{j}$ and $\subss{N}{j}$, $j\in\NN_{11}$ synthesized before and after the addition of $\subss{\hat{\Sigma}}{11}^{DG}$. We can note that the addition of the new DGU substantially influences the synthesis of these compensators.
                    \begin{figure}[!htb]
                      \centering
                      \begin{subfigure}[!htb]{0.48\textwidth}
                        \centering
                        \includegraphics[width=1\textwidth]{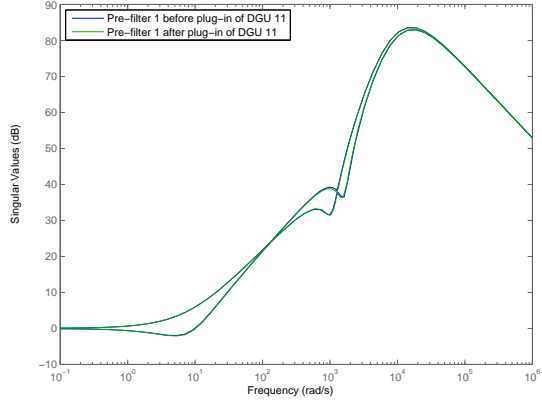}
                        \caption{Pre-filter $\subss{\tilde{C}}{1}$ before (blue) and after (green) the plugging-in of $\subss{\hat{\Sigma}}{11}^{DGU}$.}
                      \end{subfigure}
                      \quad
                      \begin{subfigure}[!htb]{0.48\textwidth}
                        \centering
                        \includegraphics[width=1\textwidth]{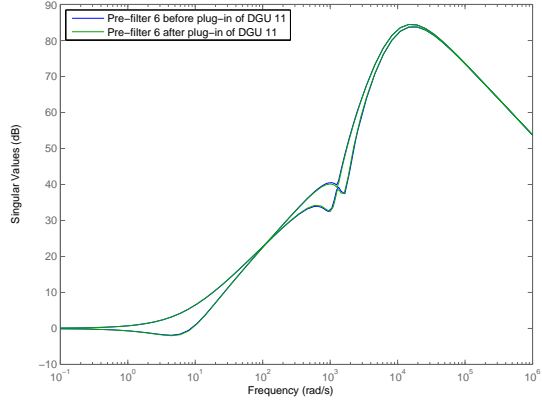}
                        \caption{Pre-filter $\subss{\tilde{C}}{6}$ before (blue) and after (green) the plugging-in of $\subss{\hat{\Sigma}}{11}^{DGU}$.}
                      \end{subfigure}
                      \begin{subfigure}[!htb]{0.48\textwidth}
                        \centering
                        \includegraphics[width=1\textwidth]{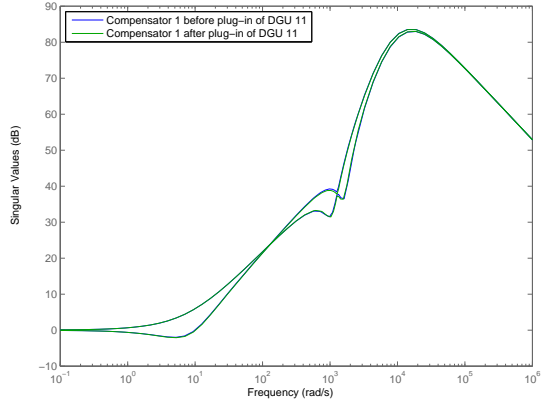}
                        \caption{Compensator $\subss{N}{1}$ before (blue) and after (green) the plugging-in of $\subss{\hat{\Sigma}}{11}^{DGU}$.}
                      \end{subfigure}
                      \quad
                      \begin{subfigure}[!htb]{0.48\textwidth}
                        \centering
                        \includegraphics[width=1\textwidth]{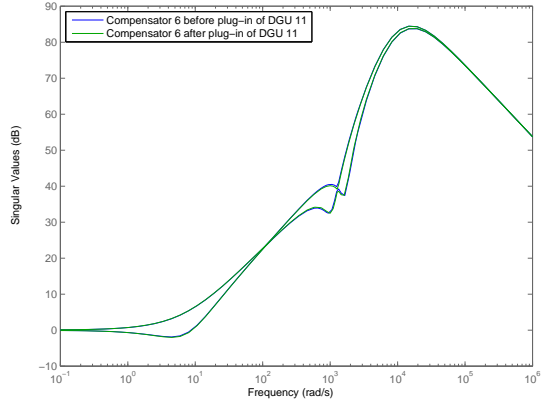}
                        \caption{Compensator $\subss{N}{6}$ before (blue) and after (green) the plugging-in of $\subss{\hat{\Sigma}}{11}^{DGU}$.}
                      \end{subfigure}
                      \caption{Comparison of PnP controllers before and after the plug-in of $\subss{\hat{\Sigma}}{11}^{DGU}$.}
                      \label{fig:ctrlcomp}
                    \end{figure}

                    After performing the plug-in operations just described, the closed-loop eigenvalues of the QSL microgrid are shown in Figure \ref{fig:CL11areaseigen}. Moreover, the closed-loop transfer function $F(s)$ coincides with the desired one, as shown in Figure \ref{fig:CL11areasCL}.
                    \begin{figure}[!htb]
                      \centering
                      \begin{subfigure}[!htb]{0.48\textwidth}
                        \centering
                        \includegraphics[width=1\textwidth]{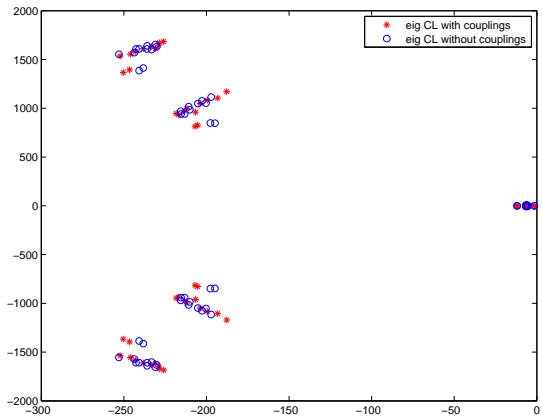}
                        \caption{Eigenvalues.}
                        \label{fig:CL11areaseigen}
                      \end{subfigure}
                      \begin{subfigure}[!htb]{0.48\textwidth}
                        \centering
                        \includegraphics[width=1\textwidth]{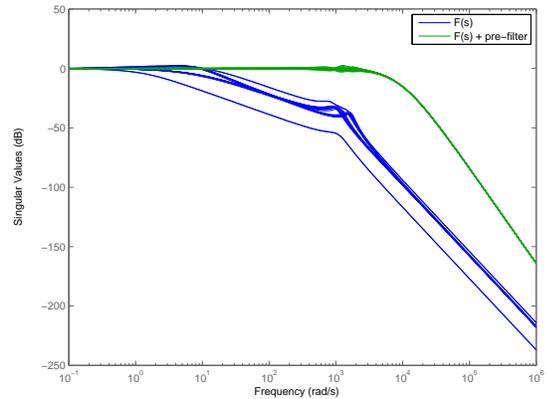}
                        \caption{Singular values of $F(s)$ with (green) and without (blue) pre-filters.}
                        \label{fig:CL11areasCL}
                      \end{subfigure}
                      \caption{Eigenvalues and singular values of the closed-loop QSL microgrid after plug-in operation.}
                      \label{fig:CL11areas}
                    \end{figure}

                    Next, we simulate the plug-in of $\subss{\hat{\Sigma}}{11}^{DGU}$. Until the plug-in of $\subss{\hat{\Sigma}}{11}^{DGU}$ references for DGUs 1-10 are the same described in Section \ref{sec:voltagetrackScenario2} and $\subss{\hat{\Sigma}}{11}^{DGU}$ is assumed to work isolated, connected to its nominal RL load (R $76~\Omega$ and L $111.9$ mH). The \emph{dq} components of the reference signal are set to $V_{d,ref}=1$ pu and $V_{q,ref}=0.4$ pu, respectively. At $t=2$ s, when the step changes of the previous case study are in steady state, we connect $\subss{\hat{\Sigma}}{11}^{DGU}$. Finally, at $t=2.3$ s, the \emph{d} component of the voltage reference for $\subss{\hat{\Sigma}}{11}^{DGU}$ steps down to 0.6 pu. Figure \ref{fig:plugsim} shows the \emph{dq} component of the load voltages for $\subss{\hat{\Sigma}}{11}^{DGU}$, and its neighbours $\subss{\hat{\Sigma}}{1}^{DGU}$ and $\subss{\hat{\Sigma}}{6}^{DGU}$. In particular, from Figures \ref{fig:plugsimdqpcc1} and \ref{fig:plugsimdqpcc6}, we note that at the plug-in time ($t=2$ s), the load voltages of $\subss{\hat{\Sigma}}{1}^{DGU}$ and $\subss{\hat{\Sigma}}{6}^{DGU}$ deviate from the respective reference signals. Thanks to the retuning of the controllers $\subss{\CC}{j}$, $\subss{\tilde{C}}{j}$ and $\subss{N}{j}$, $j\in\NN_{11}$, this deviation is immediately compensated and, after a short transient, the load voltages at $PCC_1$ and $PCC_6$ converge to the respective steady state values. Similar remarks can be done for the new DGU $\subss{\hat{\Sigma}}{11}^{DGU}$: as shown in Figure \ref{fig:plugsimdqpcc11}, there is a short transient at the time of the plug-in, that is effectively compensated by the control action. Moreover, the controller $\subss{\CC}{11}$ and compensators $\subss{\tilde{C}}{11}$ and $\subss{N}{11}$ ensure desired tracking when the the reference signal  $V_{d,ref}$ steps down at $t=2.3$.
                    
                    \begin{figure}[!htb]
                      \centering
                      \begin{subfigure}[!htb]{0.48\textwidth}
                        \centering
                        \includegraphics[width=1\textwidth]{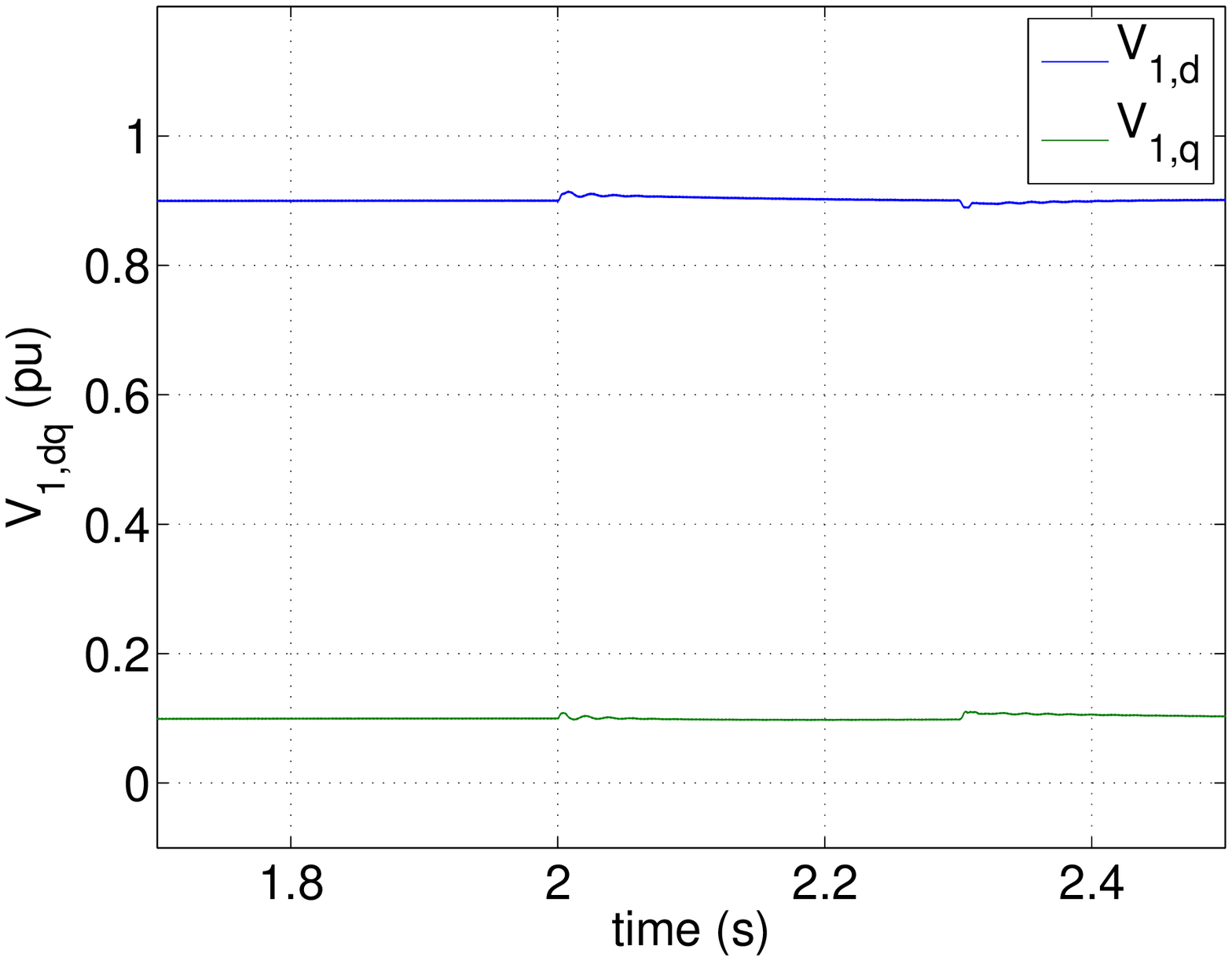}
                        \caption{\emph{d} and \emph{q} components of the voltage at $PCC_1$.}
                        \label{fig:plugsimdqpcc1}
                      \end{subfigure}
                      \begin{subfigure}[!htb]{0.48\textwidth}
                        \centering
                        \includegraphics[width=1\textwidth]{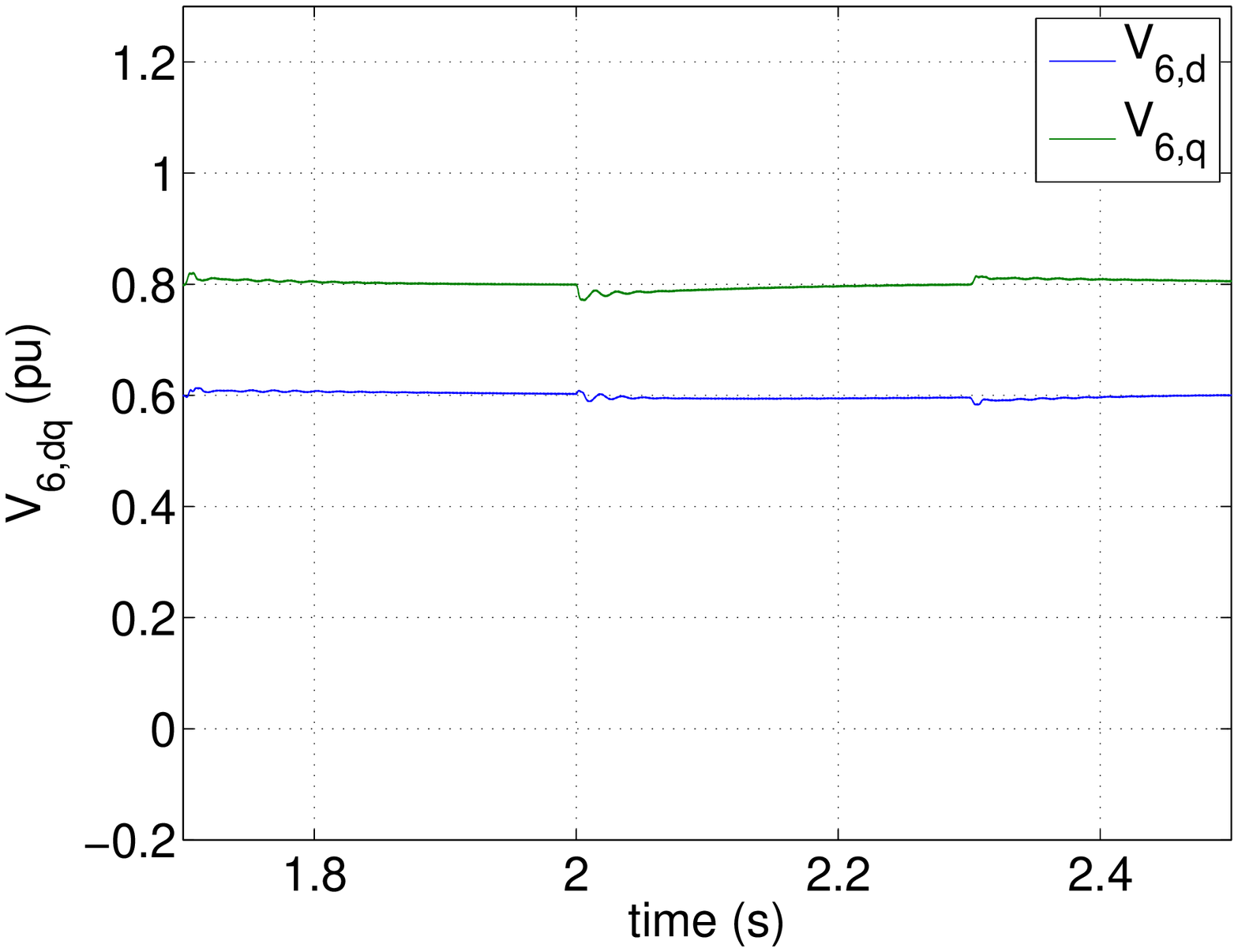}
                        \caption{\emph{d} and \emph{q} components of the voltage at $PCC_6$.}
                        \label{fig:plugsimdqpcc6}
                      \end{subfigure}
                      \begin{subfigure}[!htb]{0.48\textwidth}
                        \centering
                        \includegraphics[width=1\textwidth]{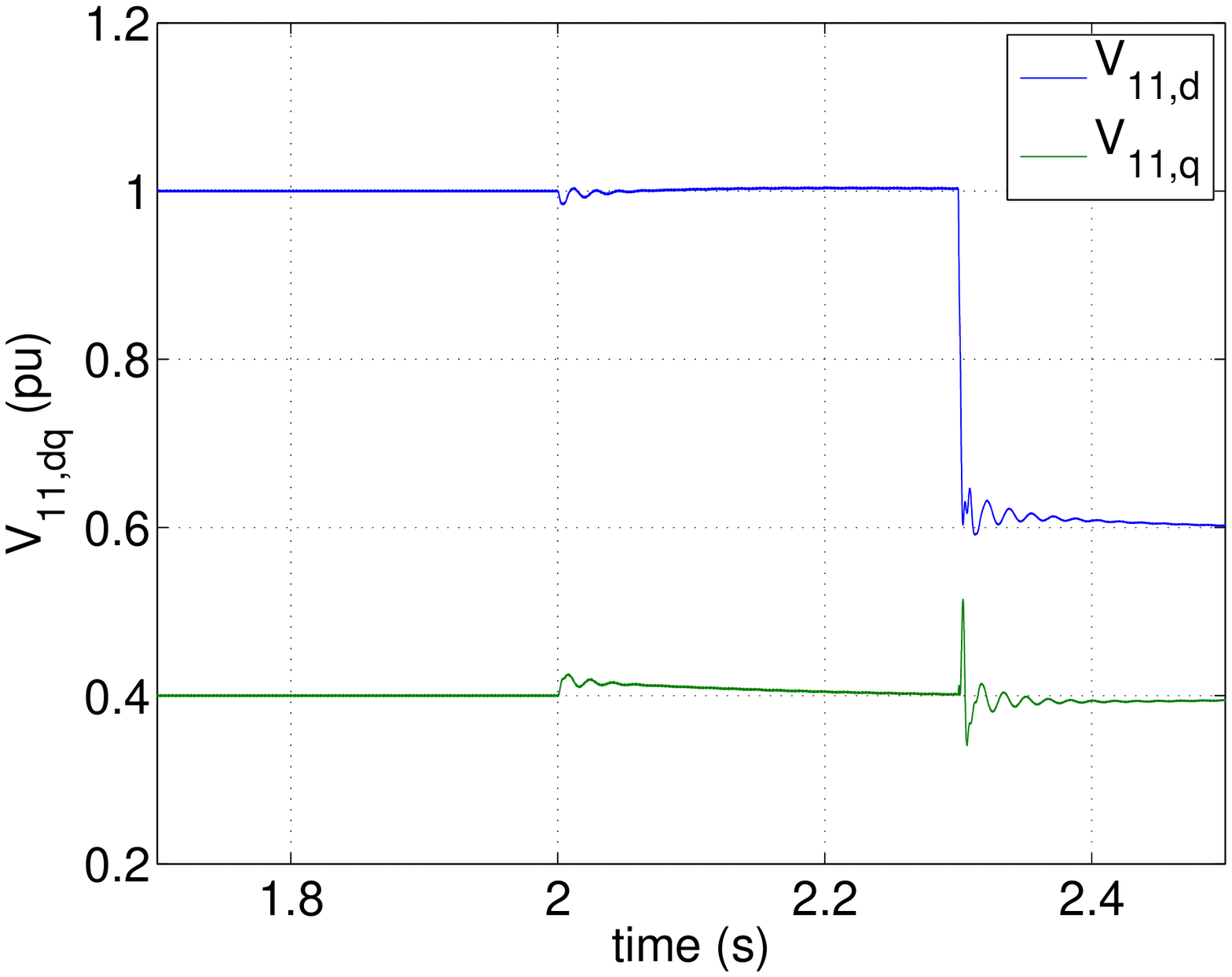}
                        \caption{\emph{d} and \emph{q} components of the voltage at $PCC_{11}$.}
                        \label{fig:plugsimdqpcc11}
                      \end{subfigure}
                      \caption{Performance of the decentralized voltage control during the plug-in operation.}
                      \label{fig:plugsim}
                    \end{figure}

               \clearpage

               \subsubsection{Unplugging of a DGU}
                    We consider the microgrid topology of Figure \ref{fig:10areasplug} and disconnect $\subss{\hat{\Sigma}}{2}^{DGU}$, hence obtaining the DGUs connected as in Figure \ref{fig:10areasunplug}. The set of neighbours of DGU 2 is $\NN_{2}=\{1,4\}$. Because of the disconnection, DGU $\subss{\hat{\Sigma}}{j}^{DGU}$, $j\in\NN_2$ change their neighbours and local dynamics $\hat{A}_{jj}$. Then, as described in Section \ref{sec:PnP}, each subsystem $\subss{\hat{\Sigma}}{j}^{DGU}$, $j\in\NN_2$ must redesign controller $\subss{\CC}{j}$ and compensators $\subss{\tilde{C}}{j}$ and $\subss{N}{j}$. Hence, matrices $\hat{A}_{jj}$, $j\in\NN_{2}$, are updated and then Algorithm \ref{alg:ctrl_design} is re-executed. Since Algorithm \ref{alg:ctrl_design} never stops in Step \ref{enu:stepAalgCtrl} the unplugging of $\subss{\hat{\Sigma}}{2}^{DGU}$ is allowed.
                    \begin{figure}[!htb]
                      \centering
                      \begin{tikzpicture}[scale=0.8,transform shape,->,>=stealth',shorten >=1pt,auto,node distance=3cm, thick,main node/.style={circle,fill=blue!20,draw,font=\sffamily\bfseries}]
					  
  \node[main node] (1) {\scriptsize{DGU $1$}};
  \node[main node] (2) [below right of=1] {\scriptsize{DGU $2$}};
  \node[main node] (3) [above right of=1] {\scriptsize{DGU $3$}};
  \node[main node] (4) [right of=3] {\scriptsize{DGU $4$}};
  \node[main node] (5) [below right of=4] {\scriptsize{DGU $5$}};
  \node[main node] (6) [below of=5] {\scriptsize{DGU $6$}};
  \node[main node] (7) [above right of=5] {\scriptsize{DGU $7$}};
  \node[main node] (8) [right of=7] {\scriptsize{DGU $8$}};
  \node[main node] (9) [right of=5] {\scriptsize{DGU $9$}};
  \node[main node] (10) [right of=6] {\scriptsize{DGU $10$}};
  \node[main node] (11) [below of=2] {\scriptsize{DGU $11$}};
  
  \path[every node/.style={font=\sffamily\small}]
  (1) edge node [left] {} (3)
  (3) edge node [right] {} (1)

  (3) edge node [left] {} (4)
  (4) edge node [right] {} (3)
  
  (4) edge node [left] {} (5)
  (5) edge node [right] {} (4)
  
  (5) edge node [left] {} (6)
  (6) edge node [right] {} (5)
  
  (6) edge node [left] {} (10)
  (10) edge node [right] {} (6)
  
  (5) edge node [left] {} (9)
  (9) edge node [right] {} (5)
  
  (5) edge node [left] {} (7)
  (7) edge node [right] {} (5)
  
  (7) edge node [left] {} (8)
  (8) edge node [right] {} (7)

  (1) edge node [left] {} (11)
  (11) edge node [right] {} (1)
  
  (6) edge node [left] {} (11)
  (11) edge node [right] {} (6);

  \draw[red,dashed] (1) to (2);
  \draw[red,dashed] (2) to (1);
  \draw[red,dashed] (2) to (4);
  \draw[red,dashed] (4) to (2);

\end{tikzpicture}
                      \caption{Scheme of the microgrid after the unplugging of $\subss{\hat{\Sigma}}{2}^{DGU}$.}
                      \label{fig:10areasunplug}
                    \end{figure}
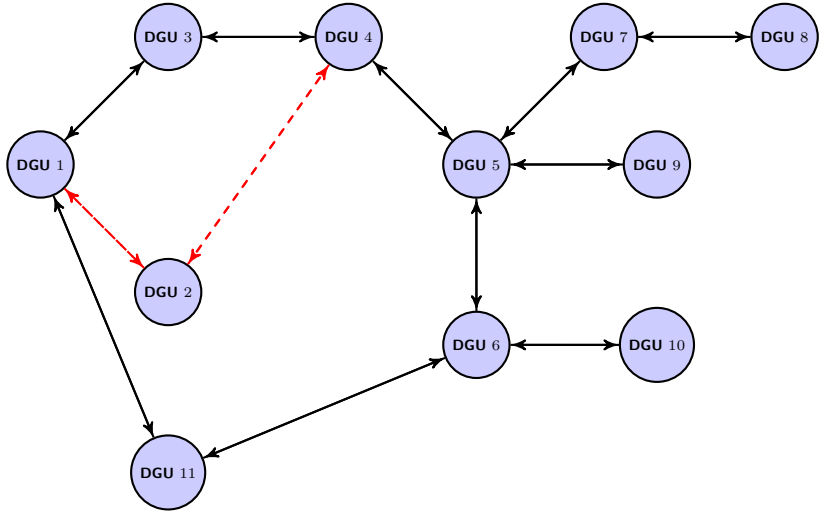

                    In spite of the unplugging operations, the closed-loop model of the QSL microgrid (now composed of 10 DGUs) is still asymptotically stable, as shown in Figure \ref{fig:CL10areasunplugeigen}. Moreover, the closed-loop transfer function $F(s)$ is the desired transfer function, as shown in Figure \ref{fig:CL10areasunplugCL}. 
                    \begin{figure}[!htb]
                      \centering
                      \begin{subfigure}[!htb]{0.48\textwidth}
                        \centering
                        \includegraphics[width=1\textwidth]{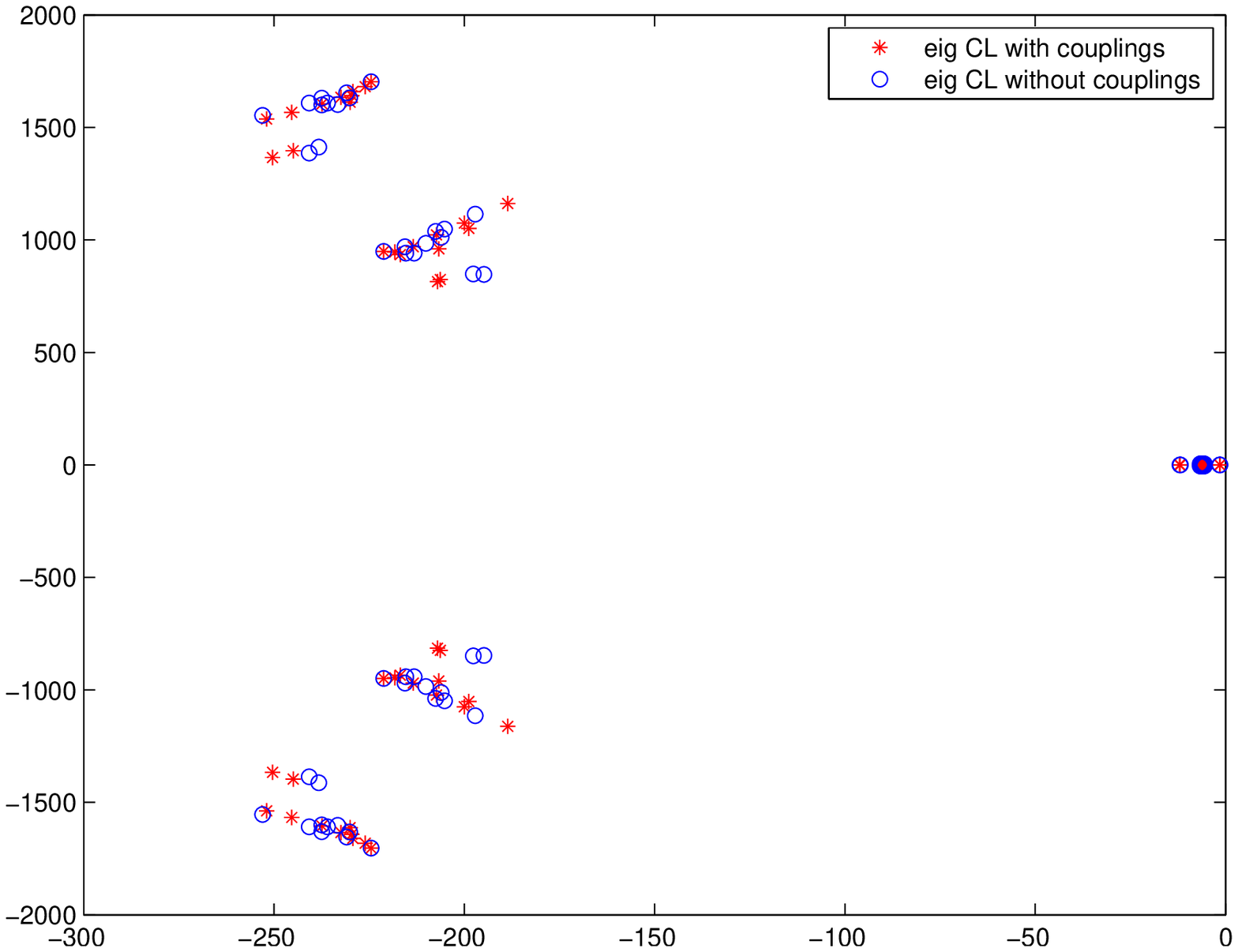}
                        \caption{Eigenvalues.}
                        \label{fig:CL10areasunplugeigen}
                      \end{subfigure}
                      \begin{subfigure}[!htb]{0.48\textwidth}
                        \centering
                        \includegraphics[width=1\textwidth]{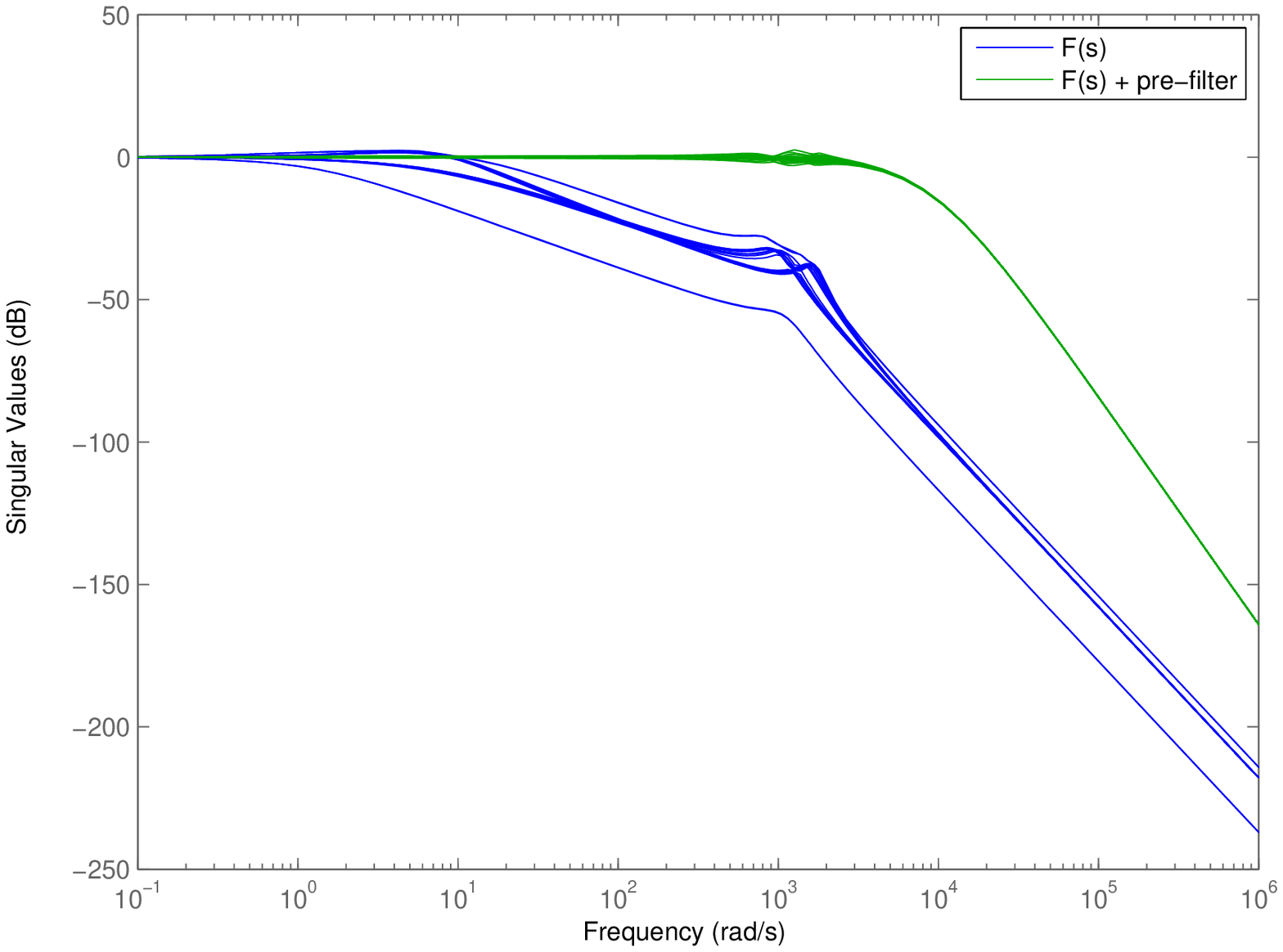}
                        \caption{Singular values of $F(s)$ with (green) and without (blue) pre-filters.}
                        \label{fig:CL10areasunplugCL}
                      \end{subfigure}
                      \caption{Eigenvalues and singular values of the closed-loop QSL microgrid after unplugging operation.}
                      \label{fig:CL11areasUnplug}
                    \end{figure}
                    
                    Finally, we simulate the unplugging operation. At $t=2.6$ s, after that the plug-in operation described in Section \ref{sec:plugintest} has been completed, DGU $\subss{\hat{\Sigma}}{2}^{DGU}$ is disconnected. As shown in Figure \ref{fig:unplugsim}, the \emph{dq} components of load voltages of DGU $\subss{\hat{\Sigma}}{j}^{DGU}$, $j\in\NN_2$ deviate from the respective reference signals. Thanks to the retuning of the controllers $\subss{\CC}{j}$ and compensators $\subss{\tilde{C}}{j}$ and $\subss{N}{j}$, $j\in\NN_{2}$, this deviation is immediately compensated and, after a short transient, the load voltages at $PCC_1$ and $PCC_4$ converge at the respective steady state values. Also in this case, stability of the microgrid is preserved despite the disconnection of $\subss{\hat{\Sigma}}{2}^{DGU}$.
                    \begin{figure}[!htb]
                      \centering
                      \begin{subfigure}[!htb]{0.48\textwidth}
                        \centering
                        \includegraphics[width=1\textwidth]{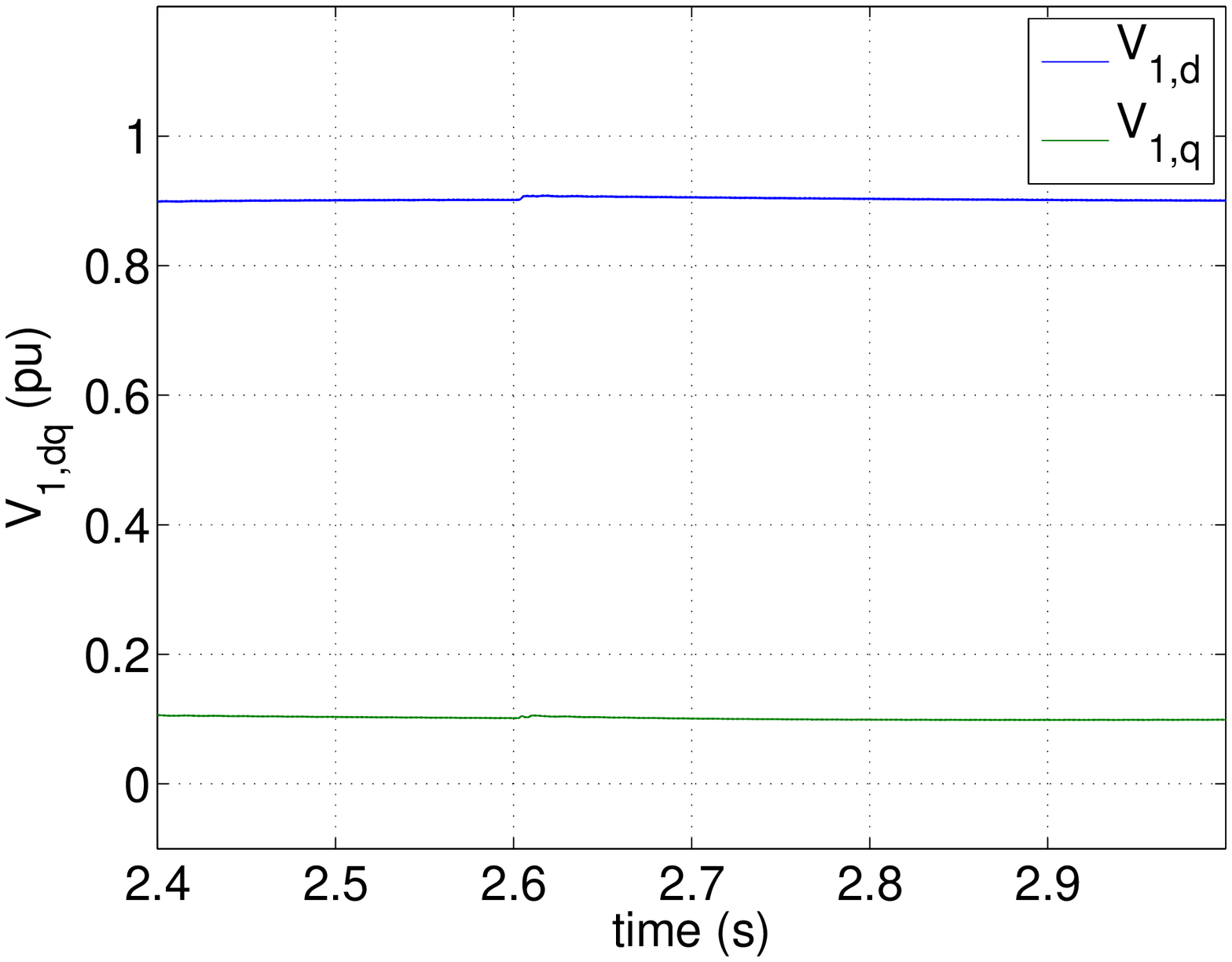}
                        \caption{\emph{d} and \emph{q} components of the voltage at $PCC_1$.}
                        \label{fig:unplugsimdqpcc1}
                      \end{subfigure}
                      \begin{subfigure}[!htb]{0.48\textwidth}
                        \centering
                        \includegraphics[width=1\textwidth]{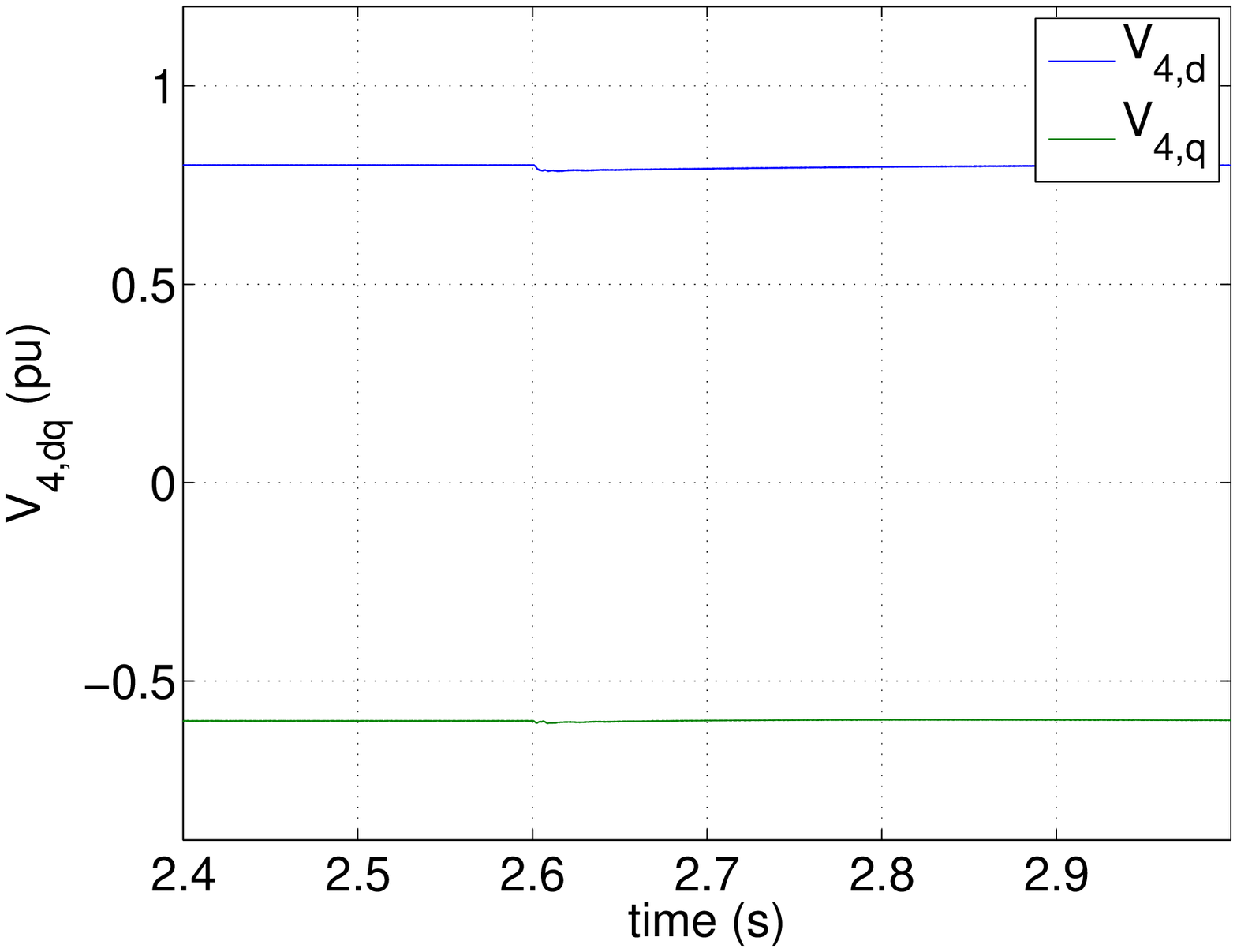}
                        \caption{\emph{d} and \emph{q} components of the voltage at $PCC_4$.}
                        \label{fig:unplugsimdqpcc4}
                      \end{subfigure}
                      \caption{Performance of decentralized voltage control during the unplugging operation.}
                      \label{fig:unplugsim}
                    \end{figure}

     \clearpage

     \section{Conclusions}
          \label{sec:conclusions}
          In this paper, we presented a decentralized control scheme for guaranteeing voltage and frequency stability in ImGs. Differently from other decentralized controllers available in the literature (e.g. \cite{Etemadi2012a,Etemadi2012,Guerrero2013}), a key feature of our approach is that plugging-in and -out of DGUs requires to update only a limited number of local controllers. Furthermore, a global model of the ImG is not required in any design step. Numerical results in Section~\ref{sec:Simresults} confirm effectiveness of PnP control even for ImGs with meshed topology and components accounting for nonlinearities commonly found in practice. Voltage and frequency control takes place at a very fast timescale where renewable sources (commonly equipped with storage devices) can be modeled as constant voltage generators. However, this approximation is no longer valid for describing the behaviour of ImGs over longer time horizons. In this case, dynamics and stochasticity of the sources plays an important role. This topic will be addressed in future research. Furthermore, local voltage controllers should be coupled with a higher control layer devoted to power flow regulation so as to orchestrate mutual help among DGUs. To this purpose, we will study if and how ideas from primary control of ImGs \cite{Guerrero2013} can be reappraised in our context.

     \clearpage

     \appendix
     \section{Matrices in microgrid models}
     \label{sec:AppMatrices}

     In this appendix, we provide all matrices appearing in Section \ref{sec:Model}.
     
     \subsection{Master-Slave model \eqref{eq:sysdistABCDM}}
          \label{sec:AppMasterSlave}
	\begin{equation*}
          \setlength{\arraycolsep}{1pt}
          \setcounter{MaxMatrixCols}{12}
          A=\matr{
            0 & \omega_0 & \frac{k_i}{C_{ti}} & 0 & \frac{1}{C_{ti}} & 0 & 0 & 0 & 0 & 0 & 0 & 0\\
            -\omega_0 & 0 & 0 & \frac{k_i}{C_{ti}} & 0 & \frac{1}{C_{ti}} & 0 & 0 & 0 & 0 & 0 & 0\\
            -\frac{k_i}{L_{ti}} & 0 & -\frac{R_{ti}}{L_{ti}} & \omega_0 & 0 & 0 & 0 & 0 & 0 & 0 & 0 & 0\\
            0 & -\frac{k_i}{L_{ti}} & -\omega_0 & -\frac{R_{ti}}{L_{ti}} & 0 & 0 & 0 & 0 & 0 & 0 & 0 & 0 \\
            -\frac{1}{L_{ij}} & 0 & 0 & 0 & -\frac{R_{ij}}{L_{ij}} & \omega_0 & 0 & 0 &\frac{1}{L_{ij}} & 0 & 0 & 0 \\
            0 & -\frac{1}{L_{ij}} & 0 & 0 & -\omega_0 & -\frac{R_{ij}}{L_{ij}} & 0 & 0 & 0 & \frac{1}{L_{ij}} & 0 & 0 \\
            \frac{1}{L_{ji}} & 0 & 0 & 0 & -\frac{R_{ji}}{L_{ji}} & \omega_0 & 0 & 0 &-\frac{1}{L_{ji}} & 0 & 0 & 0 \\
            0 & \frac{1}{L_{ji}} & 0 & 0 & -\omega_0 & -\frac{R_{ji}}{L_{ji}} & 0 & 0 & 0 & -\frac{1}{L_{ji}} & 0 & 0 \\  
            0 & 0 & 0 & 0 & 0 & 0 & \frac{1}{C_{tj}} & 0 & 0 & \omega_0 & \frac{k_j}{C_{tj}} & 0 \\
            0 & 0 & 0 & 0 & 0 & 0 & 0 & \frac{1}{C_{tj}} & -\omega_0 & 0 & 0 & \frac{k_j}{C_{tj}} \\
            0 & 0 & 0 & 0 & 0 & 0 & 0 & 0 & -\frac{k_j}{L_{tj}} & 0 & -\frac{R_{tj}}{L_{tj}} & \omega_0\\
            0 & 0 & 0 & 0 & 0 & 0 & 0 & 0 & 0 & -\frac{k_j}{L_{tj}} & -\omega_0 & -\frac{R_{tj}}{L_{tj}}
          }
        \end{equation*}
        
	\begin{equation*}
          \label{ABCDM}
          B=\matr{
            0 & 0 & 0 & 0 \\
            0 & 0 & 0 & 0 \\
            \frac{1}{L_{ti}} & 0 & 0 & 0 \\
            0 & \frac{1}{L_{ti}} & 0 & 0 \\
            0 & 0 & 0 & 0 \\
            0 & 0 & 0 & 0 \\
            0 & 0 & 0 & 0 \\
            0 & 0 & 0 & 0 \\
            0 & 0 & 0 & 0 \\
            0 & 0 & 0 & 0 \\
            0 & 0 & \frac{1}{L_{tj}} & 0 \\
            0 & 0 & 0 & \frac{1}{L_{tj}}
          }\qquad
          C^T =\matr{
            1 & 0 & 0 & 0 \\
            0 & 1 & 0 & 0 \\
            0 & 0 & 0 & 0 \\
            0 & 0 & 0 & 0 \\
            0 & 0 & 0 & 0 \\
            0 & 0 & 0 & 0 \\
            0 & 0 & 0 & 0 \\
            0 & 0 & 0 & 0 \\
            0 & 0 & 1 & 0 \\
            0 & 0 & 0 & 1 \\
            0 & 0 & 0 & 0 \\
            0 & 0 & 0 & 0
          }\qquad
          M=\matr{
            \frac{1}{C_{ti}} & 0 & 0 & 0 \\
            0 & \frac{1}{C_{ti}} & 0 & 0 \\
            0 & 0 & 0 & 0 \\
            0 & 0 & 0 & 0 \\
            0 & 0 & 0 & 0 \\
            0 & 0 & 0 & 0 \\
            0 & 0 & 0 & 0 \\
            0 & 0 & 0 & 0 \\
            0 & 0 & \frac{1}{C_{tj}} & 0 \\
            0 & 0 & 0 & \frac{1}{C_{tj}}\\
            0 & 0 & 0 & 0 \\
            0 & 0 & 0 & 0  
          }
	\end{equation*}
	
     \subsection{QSL model \eqref{eq:subsysDGUi} and \eqref{eq:subsysLine}}
          \label{sec:AppMasterMaster}
	
          \paragraph{DGU-$i$, $i\in\{1,2\}$}
          \begin{equation*}
            \renewcommand\arraystretch{1.5}
            A_{ii}=\begin{bmatrix}
              -\frac{1}{C_{ti}}\bigl(\frac{R_{ij}}{Z_{ij}^2}\bigr) & \omega_0-\frac{1}{C_{ti}}\bigl(\frac{X_{ij}}{Z_{ij}^2}\bigr) & \frac{k_i}{C_{ti}} & 0 \\
              -\omega_0+\frac{1}{C_{ti}}\bigl(\frac{X_{ij}}{Z_{ij}^2}\bigr) & -\frac{1}{C_{ti}}\bigl(\frac{R_{ij}}{Z_{ij}^2}\bigr) & 0 & \frac{k_i}{C_{ti}} \\
              -\frac{k_i}{L_{ti}} & 0 & -\frac{R_{ti}}{L_{ti}} & \omega_0 \\
              0 & -\frac{k_i}{L_{ti}} & -\omega_0 & -\frac{R_{ti}}{L_{ti}} 
            \end{bmatrix}, \mbox{with }X_{ij} = \omega_0L_{ij}\mbox{ and } Z_{ij} = \abs{R_{ij}+jX_{ij}}
          \end{equation*}
          \begin{equation*}
            \renewcommand\arraystretch{1.5}
            A_{ij}=\frac{1}{C_{ti}}
            \begin{bmatrix}
              \frac{R_{ij}}{Z_{ij}^2} & \frac{X_{ij}}{Z_{ij}^2} & 0 & 0 \\
              -\frac{X_{ij}}{Z_{ij}^2} & \frac{R_{ij}}{Z_{ij}^2} & 0 & 0 \\
              0 & 0 & 0 & 0\\
              0 & 0 & 0 & 0
            \end{bmatrix}
          \end{equation*}
          
          \begin{equation*}
            B_{i}=\begin{bmatrix}
              0 & 0\\
              0 & 0\\
              \frac{1}{L_{ti}} & 0\\
              0 & \frac{1}{L_{ti}}
            \end{bmatrix}
            \qquad
            M_{i}=\begin{bmatrix}
              -\frac{1}{C_{ti}} & 0 \\
              0 & -\frac{1}{C_{ti}} \\
              0 & 0 \\
              0 & 0 \\
            \end{bmatrix}
            \qquad
            C_{i}=\begin{bmatrix}
              1&0&0&0\\
              0&1&0&0\\
              0&0&1&0\\
              0&0&0&1
            \end{bmatrix} 
            \qquad
            H_{i}=\begin{bmatrix}
              1 & 0 & 0 & 0\\
              0 & 1 & 0 & 0
            \end{bmatrix}
          \end{equation*}
          
          \paragraph{Line $i\neq j$}
          \begin{equation}
            \renewcommand\arraystretch{1.2}
            \label{matrixss3}
            A_{li,ij}=\begin{bmatrix}
              -\frac{1}{L_{ij}} & 0 & 0 & 0\\
              0 & -\frac{1}{L_{ij}} & 0 & 0
            \end{bmatrix}\quad	
            A_{lj,ij}=\begin{bmatrix}
              \frac{1}{L_{ij}} & 0 & 0 & 0\\
              0 & \frac{1}{L_{ij}} & 0 & 0
            \end{bmatrix}\quad
            A_{ll,ij}=\begin{bmatrix}
              -\frac{R_{ij}}{L_{ij}} & \omega_0\\
              -\omega_0 & -\frac{R_{ij}}{L_{ij}}
            \end{bmatrix}
          \end{equation}
          
     \subsection{QSL model of microgrid composed by $N$ DGUs}
          \label{sec:AppNDGunit}          
          \paragraph{DGU-$i$, $i\in\DD$}
          \begin{equation}
            \label{eq:Aii}
            \renewcommand\arraystretch{2}
            A_{ii}=\begin{bmatrix}
              -\frac{1}{C_{ti}}\biggl(\sum_{j\in\NN_i}\frac{R_{ij}}{Z_{ij}^2}\biggr) & \omega_0-\frac{1}{C_{ti}}\biggl(\sum_{j\in\NN_i}\frac{X_{ij}}{Z_{ij}^2}\biggr) & \frac{k_{i}}{C_{ti}} & 0 \\
              -\omega_0+\frac{1}{C_{ti}}\biggl(\sum_{j\in\NN_i}\frac{X_{ij}}{Z_{ij}^2}\biggr) & -\frac{1}{C_{ti}}\biggl(\sum_{j\in\NN_i}\frac{R_{ij}}{Z_{ij}^2}\biggr) & 0 & \frac{k_i}{C_{ti}} \\
              -\frac{k_i}{L_{ti}} & 0 & -\frac{R_{ti}}{L_{ti}} & \omega_0 \\
              0 & -\frac{k_i}{L_{ti}} & -\omega_0 & -\frac{R_{ti}}{L_{ti}} 
            \end{bmatrix}
          \end{equation}          
          \begin{equation}
            \label{eq:Aij}
            \renewcommand\arraystretch{2}
            A_{ij}=\frac{1}{C_{ti}}\begin{bmatrix}
              \frac{R_{ij}}{Z_{ij}^2} & \frac{X_{ij}}{Z_{ij}^2} & 0 & 0 \\
              -\frac{X_{ij}}{Z_{ij}^2} & \frac{R_{ij}}{Z_{ij}^2} & 0 & 0 \\
              0 & 0 & 0 & 0\\
              0 & 0 & 0 & 0
            \end{bmatrix}
	\end{equation}
        where $R_{ij}$ and $L_{ij}$ are the resistence and the inductance of the line among DGU $i$ and DGU $j$. Moreover $X_{ij} = \omega_0L_{ij}$  and $Z_{ij} = \abs{R_{ij}+j X_{ij}}$. Matrices $B_i$, $C_i$, $M_i$ and $H_i$ are equal to those in Section \ref{sec:AppMasterMaster}.
	
        \paragraph{Overall model of the ImG with $N$ DGUs}
             \label{sec:AppOverallsys}

             \begin{equation}
               \label{TheSystem}
               \begin{aligned}
                 \begin{bmatrix}
                   \subss{\dx}{1} \\
                   \subss{\dx}{2} \\
                   \subss{\dx}{3} \\
                   \vdots \\
                   \subss{\dx}{N}
                 \end{bmatrix} 
                 &= 
                 \underbrace{\left[\begin{array}{ccccc}
                       A_{11} & A_{12} & A_{13} & \dots  & A_{1N} \\
                       A_{21} & A_{22} & A_{23} & \dots  & A_{2N} \\
                       A_{31} & A_{32} & A_{3l} & \dots  & A_{3N} \\
                       \vdots & \vdots & \vdots & \ddots & \vdots\\
                       A_{N1} & A_{N2} & A_{N3} & \dots  & A_{NN}
                     \end{array}
                   \right]}_{\mbf{A}}
                 \begin{bmatrix}
                   \subss{x}{1} \\
                   \subss{x}{2} \\
                   \subss{x}{3} \\
                   \vdots \\
                   \subss{x}{N}
                 \end{bmatrix} 
                 +\\
                 &+\,
                 \underbrace{\begin{bmatrix}
                     B_{1} & 0 & \dots & 0\\
                     0 & B_{2} & \ddots & \vdots\\
                     \vdots & \ddots & \ddots & 0\\
                     0& \dots & 0  & B_{N}
                   \end{bmatrix}}_{\mbf{B}}
                 \begin{bmatrix}
                   \subss{u}{1}\\
                   \subss{u}{2}\\
                   \vdots\\
                   \subss{u}{N}
                 \end{bmatrix}
                 + \underbrace{\begin{bmatrix}
                     M_{1} & 0 & \dots & 0\\
                     0 & M_{2} & \ddots & \vdots\\
                     \vdots & \ddots & \ddots & 0\\
                     0& \dots & 0  & M_{N}
                   \end{bmatrix}}_{\mbf{M}}
                 \begin{bmatrix}
                   \subss{d}{1}\\
                   \subss{d}{2}\\
                   \vdots\\
                   \subss{d}{N}
                 \end{bmatrix}\\    
                 \begin{bmatrix}
                   \subss{y}{1}\\
                   \subss{y}{2}\\
                   \subss{y}{3}\\
                   \vdots\\
                   \subss{y}{N}
                 \end{bmatrix}
                 &=
                 \underbrace{\left[\begin{array}{ccccc}
                       C_{1} & 0 & 0 & \dots & 0 \\
                       0 & C_{2} & 0 & \ddots & \vdots \\
                       0 & 0 & C_{3} & \ddots & 0 \\
                       \vdots & \ddots & \ddots &\ddots & 0\\
                       0 & \dots & 0 & 0  & C_{N}
                     \end{array}
                   \right]}_{\mbf{C}}
                 \begin{bmatrix}
                   \subss{x}{1} \\
                   \subss{x}{2} \\
                   \subss{x}{3} \\
                   \vdots \\
                   \subss{x}{N}
                 \end{bmatrix}\\
                 \begin{bmatrix}
                   \subss{z}{1}\\
                   \subss{z}{2}\\
                   \subss{z}{3}\\
                   \vdots\\
                   \subss{z}{N}
                 \end{bmatrix}
                 &=
                 \underbrace{\begin{bmatrix}
                     H_{1} & 0 & 0 & \dots & 0 \\
                     0 & H_{2} & 0 & \ddots & \vdots \\
                     0 & 0 & H_{3} & \ddots & 0 \\
                     \vdots & \ddots & \ddots &\ddots & 0\\
                     0& \dots & 0 & 0  & H_{N}
                   \end{bmatrix}}_{\mbf{H}}\begin{bmatrix}
                   \subss{y}{1}\\
                   \subss{y}{2}\\
                   \subss{y}{3}\\
                   \vdots\\
                   \subss{y}{N}
                 \end{bmatrix}.
               \end{aligned} 
             \end{equation}
     \clearpage
     \section{Relevance of coupling in the decentralized design of regulators for the Master-Slave model}
\label{sec:AppUnstable}
	In this appendix, we show why the decentralized design of
        stabilizing controllers must take into account how DGUs are
        coupled through matrices $A_{ij}$, $i,j\in\NN, i\neq j$. We consider the master-slave system \eqref{eq:sysdistdq} and, for simplicity, set $i=1$ and $j=2$. Augmenting each DGU model with integrators (as in Section \ref{sec:ctrlint}) one obtains
	\begin{equation}
          \label{eq:subsystemsdec}
          \begin{aligned}
            \subss{\hat{\Sigma}}{1} &: \left\lbrace \begin{aligned}
		\subss{\dot{\hat{x}}}{1}(t) &= \hat{A}_{11}\subss{\hat{x}}{1}(t)+\hat{A}_{12}\subss{\hat{x}}{2}(t)+\hat{B}_1 \subss{u}{1}(t)+\hat{M}_1 \subss{\hat{d}}{1}(t)\\
		\subss{y}{1}(t)       &= \hat{C}_1 \subss{\hat{x}}{1}(t)
              \end{aligned}
            \right.\\
            \subss{\hat{\Sigma}}{2} &: \left\lbrace \begin{aligned}
		\subss{\dot{\hat{x}}}{2}(t) &= \hat{A}_{22}\subss{\hat{x}}{2}(t)+\hat{A}_{21}\subss{\hat{x}}{1}(t)+\hat{B}_2 \subss{u}{2}(t)+\hat{M}_2 \subss{\hat{d}}{2}(t)\\
		\subss{y}{2}(t)       &= \hat{C}_2 \subss{\hat{x}}{2}(t)
              \end{aligned}
            \right.
          \end{aligned}
	\end{equation}
        where $\subss{\hat{x}}{1}=[V_{1,d},V_{1,q},I_{t1,d},I_{t1,q},I_{s,d},I_{s,q},v_{1,d},v_{1,q}]^T\in\Rset^8$, $\subss{u}{1}=[V_{t1,d},V_{t1,q}]^T\in\Rset^2$, $\subss{y}{1}=[V_{1,d},V_{1,q}]^T\in\Rset^2$, $\subss{\hat{d}}{1}=[I_{L1_d},I_{L1_q},V_{1,d\,ref},V_{1,q\,ref}]^T\in\Rset^4$ are, respectively, the state, the control input, the controlled variable and the exogenous input of the Master area. Similarly, for the Slave area we have $\subss{\hat{x}}{2}=[V_{2,d},V_{2,q},I_{t2,d},I_{t2,q},v_{2,d},v_{2,q}]^T\in\Rset^6$, $\subss{u}{2}=[V_{t2,d},V_{t2,q}]^T\in\Rset^2$, $\subss{y}{2}=[V_{2,d},V_{2,q}]^T\in\Rset^2$, $\subss{d}{2}=[I_{L2_d},I_{L2_q},V_{2,d\,ref},V_{2,q\,ref}]^T\in\Rset^4$. Matrices in \eqref{eq:subsystemsdec} are obtained from \eqref{eq:sysdistdq} as explained in Section \ref{sec:ctrlint}. Moreover, we use the same electrical parameters in \cite{Moradi2010}.\\
        Next, we design decentralized controllers for each DGU assuming that they are dynamically decoupled, hence $\hat{A}_{12}=\hat{A}_{21}=0$. From the definition of matrix $\hat{C}_i$ in \eqref{eq:augith}, the state $\subss\hx i$ is measured and therefore we design state-feedback decentralized controllers
	\begin{equation}
          \label{eq:ctrlLaws}
          \begin{aligned}
            \subss{u}{1}(t)&=K_{1}\subss{\hat{x}}{1}(t)\\
            \subss{u}{2}(t)&=K_{2}\subss{\hat{x}}{2}(t)
          \end{aligned}
	\end{equation}
        where $K_{1}$ and $K_{2}$ are designed as Linear Quadratic Regulators (LQRs) using the weights $Q_{1}=\diag(I_8,10I_2)$, $Q_{2}=\diag(I_6,10I_2)$, $R_{1}=I_2$ and $R_{2}=I_2$. Control laws \eqref{eq:ctrlLaws} guarantee that the closed-loops decoupled DGUs
	\begin{equation}
          \label{eq:sysdecex}
          \begin{bmatrix}
            \subss{\dot{\hat{x}}}{1}(t)\\
            \subss{\dot{\hat{x}}}{2}(t)
          \end{bmatrix}=\begin{bmatrix}
            \hat{A}_{11} & 0\\
            0 & \hat{A}_{22}\\
          \end{bmatrix}
          \begin{bmatrix}
            \subss{\hat{x}}{1}(t)\\
            \subss{\hat{x}}{2}(t)
          \end{bmatrix}
          +\begin{bmatrix}
            \hat{B}_{1}K_{1} & 0\\
            0 & \hat{B}_{2}K_{2}
          \end{bmatrix}
          \begin{bmatrix}
            \subss{\hat{x}}{1}(t)\\
            \subss{\hat{x}}{2}(t)
          \end{bmatrix}.
	\end{equation}
        are asymptotically stable. 

        Figure \ref{fig:eigOpenL} shows the eigenvalues of the open-loop system \eqref{eq:subsystemsdec}: they are asymptotically stable except for the eigenvalues associated with the integrators, i.e. eigenvalues in the origin. Using the stabilizing controllers \eqref{eq:ctrlLaws}, the eigenvalues of the decoupled closed-loop system \eqref{eq:sysdecex} are asymptotically stable (see black stars in Figure \ref{fig:eigCL}).

	\begin{figure}[!htb]
          \centering
          \includegraphics[width=0.5\textwidth]{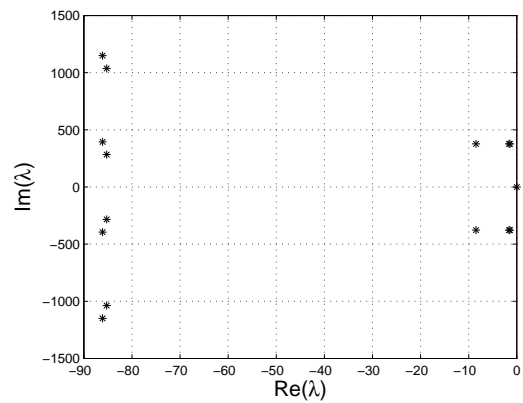}
          \caption{Eigenvalues of the overall open-loop system \eqref{eq:subsystemsdec}.}
          \label{fig:eigOpenL}
        \end{figure} 
	\begin{figure}[!htb]
          \centering
          \includegraphics[width=0.5\textwidth]{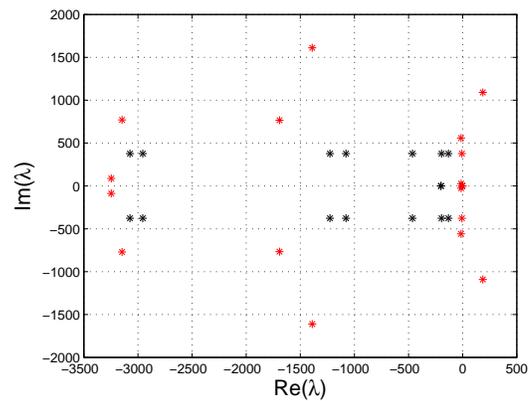}
          \caption{Eigenvalues of the overall closed-loop system \eqref{eq:sysdecex} (in black) and of the overall closed-loop system \eqref{eq:syscouplex} (in red).}
          \label{fig:eigCL}
	\end{figure}
        
        Considering coupling terms, the closed-loop system is  
	\begin{equation}
          \label{eq:syscouplex}
          \begin{bmatrix}
            \subss{\dot{\hat{x}}}{1}(t)\\
            \subss{\dot{\hat{x}}}{2}(t)
          \end{bmatrix}=\begin{bmatrix}
            \hat{A}_{11} & \hat{A}_{12} \\
            \hat{A}_{21}  & \hat{A}_{22}\\
          \end{bmatrix}
          \begin{bmatrix}
            \subss{\hat{x}}{1}(t)\\
            \subss{\hat{x}}{2}(t)
          \end{bmatrix}
          +\begin{bmatrix}
            \hat{B}_{1}K_{1} & 0\\
            0 & \hat{B}_{2}K_{2}
          \end{bmatrix}
          \begin{bmatrix}
            \subss{\hat{x}}{1}(t)\\
            \subss{\hat{x}}{2}(t)
          \end{bmatrix}.
	\end{equation}
        Since the controllers have been designed without taking into account the coupling terms, we can not ensure that system \eqref{eq:syscouplex} is asymptotically stable. Indeed, in this case, some eigenvalues of \eqref{eq:syscouplex} have positive real part (see red stars in Figure \ref{fig:eigCL}).

     \clearpage
     \section{Electrical and simulation parameters of Scenario 1 and 2}
     \label{sec:AppElectrPar}
     This appendix collects all the electrical and simulation parameters of Scenarios 1 and 2 described in Sections \ref{sec:scenario1} and \ref{sec:scenario2}, respectively.

     \begin{table}[!htb]
       \caption{Electrical parameters for microgrid with two DGUs in Scenario 1.}	
       \label{tbl:par2areas}
       \centering
       \begin{tabular}{*{2}{c}}
         \toprule
         Quantity & Values \\
         \midrule
         $R_{t1}$ (VSC filter resistance in DGU 1) & 1.5 m$\Omega$\\
         $L_{t1}$ (VSC filter inductance in DGU 1) & 100 $\mu$H\\
         $R_{t2}$ (VSC filter resistance in DGU 2) & 1.8 m$\Omega$\\
         $L_{t2}$ (VSC filter inductance in DGU 2) & 120 $\mu$H\\
         $C_{t}$ (Shunt capacitance) & 62.86 $\mu$F\\
         VSC rated power (RMS) & 3 MVA\\
         VSC terminal voltage (line to line RMS) & 600 V\\
         VSC modulation index ($m_{f}$) & 33\\
         \midrule
         $f_0$ (microgrid frequency) & 60 Hz\\
         $V_{DC}$ (DC bus voltage) & 2000 V\\
         Transformer voltage ratio k (Y/$\Delta$) & 0.6/13.8\\
         \midrule
         $R_{s}$ (transmission line resistance) & 1$\Omega$\\
         $L_{s}$ (transmission line inductance) & 600 mH\\
         \midrule
         $S_{base}$ (power base value) & 6 MVA (1 pu)\\
         $V_{base}$ (voltage base value) & $\frac{13000\sqrt 2}{\sqrt 3}$ (1 pu)\\
         \bottomrule
       \end{tabular}
     \end{table}

     \begin{table}[!htb]
       \caption{VSC filter parameters for DGUs $\subss{\hat{\Sigma}}{i}^{DGU}$, $i=\{1,\dots,11\}$ in Scenario 2.}	
       \label{tbl:diffpar10}
       \centering
       \begin{tabular}{*{3}{c}}
         \toprule
         DGU & Resistance $R_t (m\Omega)$ & Inductance $L_t (\mu H)$ \\
         \midrule
         $\subss{\hat{\Sigma}}{1}^{DGU}$& 1.2 & 93.7\\
         $\subss{\hat{\Sigma}}{2}^{DGU}$& 1.6 & 94.8\\
         $\subss{\hat{\Sigma}}{3}^{DGU}$& 1.5 & 107.7\\
         $\subss{\hat{\Sigma}}{4}^{DGU}$& 1.5 & 90.6\\
         $\subss{\hat{\Sigma}}{5}^{DGU}$& 1.7 & 99.8\\
         $\subss{\hat{\Sigma}}{6}^{DGU}$& 1.6 & 93.4\\
         $\subss{\hat{\Sigma}}{7}^{DGU}$& 1.6 & 109.6\\
         $\subss{\hat{\Sigma}}{8}^{DGU}$& 1.7 & 104.3\\
         $\subss{\hat{\Sigma}}{9}^{DGU}$& 1.7 & 100.0\\
         $\subss{\hat{\Sigma}}{10}^{DGU}$& 1.5 & 99.4\\
         \midrule
         $\subss{\hat{\Sigma}}{11}^{DGU}$& 1.5 & 100.0\\
         \bottomrule
       \end{tabular}
     \end{table}

     \begin{table}[!htb]
       \caption{Connection three-phase lines parameters for Scenario 2.}	
       \label{tbl:linespar10}
       \centering
       \begin{tabular}{*{3}{c}}
         \toprule
         Connected DGUs (i,j) & Resistance $R_s (\Omega)$ & Inductance $L_s (mH)$ \\
         \midrule
         $(1,2)$ & 1.1 & 600 \\
         $(1,3)$ & 0.9 & 400 \\
         $(3,4)$ & 1.0 & 500 \\
         $(2,4)$ & 1.2 & 700 \\
         $(4,5)$ & 1.0 & 550 \\
         $(5,7)$ & 0.7 & 350 \\
         $(5,6)$ & 1.3 & 800 \\
         $(5,9)$ & 1.2 & 650 \\
         $(7,8)$ & 1.0 & 450 \\
         $(6,10)$ & 0.4 & 600 \\
         \midrule
         $(1,11)$ & 1.0 & 700 \\
         $(6,11)$ & 1.1 & 600 \\
         \bottomrule
       \end{tabular}
     \end{table}

     \begin{table}[!htb]
       \caption{Common parameters of DGUs $\subss{\hat{\Sigma}}{i}^{DGU}$, $i=\{1,\dots,11\}$ in Scenario 2.}	
       \label{tbl:commpar10}
       \centering
       \begin{tabular}{*{2}{c}}
         \toprule
         Quantity & Values \\
         \midrule
         VSC rated power (RMS) & 6 MVA\\
         VSC terminal voltage (line to line RMS) & 600 V\\
         VSC modulation index ($m_{f}$) & 33\\
         \midrule
         $f_0$ (microgrid frequency) & 60 Hz\\
         $V_{DC}$ (DC bus voltage) & 2000 V\\
         Transformer voltage ratio k (Y/$\Delta$) & 0.6/13.8\\
         \midrule
         $S_{base}$ (power base value) & 6 MVA (1 pu)\\
         $V_{base}$ (voltage base value) & $\frac{13000\sqrt 2}{\sqrt 3}$ (1 pu)\\
         \bottomrule
       \end{tabular}
     \end{table}

        \begin{table}[!htb]
          \caption{Parameter of step references in simulations of Scenario 2. $V_{d,ref}$ and $V_{q,ref}$ are the final values of reference signals after time Tstep.}	
          \label{tbl:parsim10}
          \centering
          \begin{tabular}{*{6}{c}}
            \toprule
            DGU & $R (\Omega)$ & $L (mH)$ & $V_{d,ref}$ (pu) & $V_{q,ref}$ (pu) & Tstep (s) \\
            \midrule
            $\subss{\hat{\Sigma}}{1}^{DGU}$& 76 & 111.9 & 0.9 & 0.1 & 0.8\\
            $\subss{\hat{\Sigma}}{2}^{DGU}$& 85 & 134.3 & 0.9 & -0.1 & 0.9\\
            $\subss{\hat{\Sigma}}{3}^{DGU}$& 93 & 123.1 & 0.8 & 0.6 & 1.0\\
            $\subss{\hat{\Sigma}}{4}^{DGU}$& 80 & 167.9 & 0.8 & -0.6 & 1.1\\
            $\subss{\hat{\Sigma}}{5}^{DGU}$& 125 & 223.8 & 0.8 & 0.1 & 1.2\\
            $\subss{\hat{\Sigma}}{6}^{DGU}$& 90 & 156.7 & 0.6 & 0.8 & 1.3\\
            $\subss{\hat{\Sigma}}{7}^{DGU}$& 103 & 145.5 & 0.7 & 0.7 & 1.4\\
            $\subss{\hat{\Sigma}}{8}^{DGU}$& 150 & 179.0 & 0.9 & 0.2 & 1.5\\
            $\subss{\hat{\Sigma}}{9}^{DGU}$& 81 & 190.2 & 0.9 & -0.3 & 1.6\\
            $\subss{\hat{\Sigma}}{10}^{DGU}$& 76 & 111.9 & 0.8 & 0.4 & 1.7\\
            \bottomrule
          \end{tabular}
	\end{table}
     \clearpage

     \bibliographystyle{IEEEtran}
     \bibliography{microgrids-report}

\end{document}